\documentclass{article}
\usepackage{fullpage}

\PassOptionsToPackage{numbers, compress}{natbib}

\usepackage{wrapfig}

\usepackage{graphicx} %
\usepackage{subcaption}

\usepackage{ktmacros-lncs}

\declaretheorem[name=Theorem, style=plain, numberwithin=section]{theorem}
\declaretheorem[name=Definition, numberwithin=section]{definition}
\declaretheorem[name=Lemma, sibling=theorem]{lemma}
\declaretheorem[name=Remark, sibling=theorem]{remark}

\newif\ifconf
\conffalse

\begin{document}
\title{PREAMBLE: Private and Efficient Aggregation via Block Sparse Vectors}
\author{Hilal Asi\footnote{Apple} \and Vitaly Feldman\footnotemark[1] \and Hannah Keller\footnote{Aarhus University. Research done while at Apple.} \and Guy N. Rothblum\footnotemark[1] \and Kunal Talwar\footnotemark[1]}

\maketitle
\begin{abstract}

We revisit the problem of secure aggregation of high-dimensional vectors in a two-server system such as Prio. These systems are typically used to aggregate vectors such as gradients in private federated learning, where the aggregate itself is protected via noise addition to ensure differential privacy. Existing approaches require communication scaling with the dimensionality, and thus limit the dimensionality of vectors one can efficiently process in this setup. %

We propose PREAMBLE: {\bf Pr}ivate {\bf E}fficient {\bf A}ggregation {\bf M}echanism via  {\bf BL}ock-sparse {\bf E}uclidean Vectors. PREAMBLE builds on an extension of distributed point functions that enables communication- and computation-efficient aggregation of {\em block-sparse vectors}, which are sparse vectors where the non-zero entries occur in a small number of clusters of consecutive coordinates. We show that these block-sparse DPFs can be combined with random sampling and privacy amplification by sampling results, to allow asymptotically optimal privacy-utility trade-offs for vector aggregation, at a fraction of the communication cost. When coupled with recent advances in numerical privacy accounting, our approach incurs a negligible overhead in noise variance, compared to the Gaussian mechanism used with Prio.

\end{abstract}

\section{Introduction}
\label{sec:intro}

Secure Aggregation is a fundamental primitive in multiparty communication, and underlies several large-scale deployments of federated learning and statistics. Motivated by applications to private federated learning, we study the problem of aggregation of high-dimensional vectors, each of bounded Euclidean norm. We will study this problem in the same trust model as Prio~\cite{Corrigan-GibbsB17}, where at least one of two servers is assumed to be honest. This setup has been deployed at large scale in practice and our goal in this work is to design algorithms for estimating the sum of a large number of high-dimensional vectors, while keeping the device-to-server communication, as well as the client and server computation small.

This problem was studied in the original Prio paper, as well as in several subsequent works~\cite{BBCGI19,Talwar22,AddankiGJOP22,RatheeSWP22,RohrigU23,RothblumOCT24}. In Prio, the client creates additive secret-shares of its vector and sends those to the two servers. One of the shares can be replaced by a short seed, and thus the communication out of the client for sending $\fulldim$-dimensional vector is $\fulldim + O(1)$ field elements. Additionally, such deployments typically require resistance to malicious clients, which can be ensured by sending zero-knowledge proofs showing that the secret-shared vector has bounded norm. Existing protocols~\cite{RothblumOCT24} allow for very efficient proofs that incur negligible communication overhead.%

In recent years, the size of models that are used on device has significantly increased. Recent work has shown that in some contexts, larger models are easier to train in the private federated learning setup~\cite{AzamPFTSL23,ChauhanCTTN24}. While approaches have been developed to fine-tuning models while training a fraction of the model parameters (e.g.~\cite{Hu22}), increasing model sizes often imply that the number of trained parameters is in the range of millions to hundreds of millions. The communication cost is further exacerbated in the secret-shared setting, as one must communicate $\fulldim$ field elements, which are typically 64 or at least 32 bits, even when the gradients themselves may be low-precision. As an example, with a 64-bit field, an eight million-dimensional gradient will require at least 64MB of communication from each client to the servers.

In applications to private federated learning, noise is typically added to the sum of gradients from hundreds of thousands of devices, to provide a provable differential privacy guarantee. Even absent added noise, the aggregate is often inherently noisy in statistical settings. In such setups, computing
the exact aggregate may be overkill, and indeed, previous work in the single-trusted-server setting has proposed reducing the communication cost by techniques such as random projections~\cite{SureshYuKuMc17,VargaftikBaPoGaBeYaMi21,VargaftikBaPoMeItMi22,AsiFNNT23,ChenIKNOX24}. However, random projections increase the sensitivity of the gradient estimate, and hence naively, would require more noise to be added to protect privacy. %

This additional overhead can be reduced if each client picks a different random subset of coordinates, and this subset remain hidden from the adversary. Thus while each client could send a {\em sparse} vector, the sparsity pattern must remain hidden. This constraint prevents any reduction in communication costs when using Prio. A beautiful line of work~\cite{GilboaI14,EC:BoyGilIsh15,CCS:BoyGilIsh16,BonehBCGI21} on {\em Function Secret Sharing} addresses questions of this kind. In particular, {\em Distributed Point Functions} (DPFs) allow for low-communication secret-sharing of sparse vectors in a high-dimensional space. However, this approach is primarily designed to allow the servers to efficiently compute any one coordinate of the aggregate. The natural extension of their approach to aggregating high-dimensional vectors incurs a non-trivial overheads for the parameter settings of interest, where the sparse vectors have tens of thousands of non-zero entries.

In this work, we address this problem of high-dimensional vector aggregation in a two-server setting. We give the first protocol for this problem that has sub-linear communication cost, reasonable %
client and server computation costs, and gives near-optimal trade-offs between utility and (differential) privacy. Our approach builds on the ability to efficiently handle the class of {\em $k$-block sparse} vectors %
(see~\cref{fig:sparsity}). For a parameter $\blocksize$, we group the $\fulldim$ coordinates into $\numblocks = \fulldim/\blocksize$ blocks of $\blocksize$ coordinates each. A $k$-block-sparse vector is one where at most $k$ of these blocks take non-zero values. We make the following contributions:%

\begin{itemize}
\setlength{\itemsep}{2pt}
    \setlength{\parskip}{0pt}
  \item We identify block-sparseness as the ``right'' abstraction, that effectively balances the expressiveness needed for accuracy with the structure needed for efficient cryptographic aggregation protocols. 
  \item We propose an extension of the distributed point function construction of~\cite{CCS:BoyGilIsh16} that can secret-share $k$-block-sparse vectors while communicating $\approx kB$ field elements. For typical parameter settings, our approach that uses blocks is significantly more communication- and computation-efficient, compared to schemes for sparse vectors.
  \item We show how an aggregation scheme for $k$-block-sparse inputs combines with sampling analyses for utility, and privacy-amplification-by-sampling analyses for privacy accounting. Combined with recent advances in numerical privacy accounting, we show that for reasonable settings of parameters, our approach leads to privacy-utility trade-offs comparable to the Gaussian mechanism, while providing significantly smaller communication costs. For instance, in the case of an eight million-dimensional vector with 64 bit field size, our approach reduces the communication from $64MB$ to about $1MB$, while increasing the noise standard deviation by about $10\%$ for $(1,10^{-6})$-DP when aggregating $100K$ vectors.
\end{itemize}

\begin{figure}
\begin{center}
\begin{tikzpicture}[scale=0.4, font=\small]

\fill[gray] (3,3) rectangle (3.5, 3.5);
\draw (5,2) node {(a) 1-sparse} (5,2);

\filldraw[gray] (17,3) rectangle (17.5, 3.5);
\filldraw[gray] (20.5,3) rectangle (21, 3.5);
\filldraw[gray] (23,3) rectangle (23.5, 3.5);
\draw (20,2) node {(b) $k$-sparse} (20,2);

\filldraw[gray] (4,0) rectangle (6,0.5);
\draw[thick] (2,0) -- (2,0.5);
\draw[thick] (4,0) -- (4,0.5);
\draw[thick] (6,0) -- (6,0.5);
\draw[thick] (8,0) -- (8,0.5);
\draw (5,-1) node {(c) $1$-block-sparse} (5,-1);

\filldraw[gray] (17,0) rectangle (19,0.5);
\filldraw[gray] (21,0) rectangle (23,0.5);
\filldraw[gray] (23,0) rectangle (25,0.5);
\draw[thick] (17,0) -- (17,0.5);
\draw[thick] (19,0) -- (19,0.5);
\draw[thick] (21,0) -- (21,0.5);
\draw[thick] (23,0) -- (23,0.5);
\draw (20,-1) node {(d) $k$-block-sparse} (20,-1);

\draw[thick] (0,0) rectangle (10,0.5);
\draw[thick] (15,0) rectangle (25,0.5);
\draw[thick] (0,3) rectangle (10,3.5);
\draw[thick] (15,3) rectangle (25,3.5);

\end{tikzpicture}
\caption{\small Possible non-zero patterns (gray) of $1$-sparse, $k$-sparse, $1$-block-sparse and $k$-block-sparse vectors.}\label{fig:sparsity}
\end{center}
\end{figure}
\subsection*{Overview of Techniques}

Compressing high-dimensional vectors for noisy aggregation is a standard %
consequence of the beautiful advances in randomized sketching algorithms. For example, one can use standard random projection techniques, including efficient versions of the Johnson-Lindenstrauss lemma~\cite{JohnsonL84,AilonC06}, to construct a sparse unbiased estimator for each vector. A random rotation of the vector, followed by subsampling an appropriate number (say $m \approx 50,000$) of coordinates is sufficient to get a good approximation to the aggregate. Indeed this approach has been proposed for vector aggregation in prior works.

This approach has two significant challenges. Firstly, to get a good approximation for practical parameters, the number of non-zero coordinates $m$ is in the tens of thousands. While sparse vectors can be communicated efficiently using distributed point functions (DPF) constructions in the literature, the overheads in terms of client computation, communication, and server computation are fairly significant. Secondly, this reduction to sparse vectors significantly increases the {\em sensitivity} of the final estimate and results in larger privacy noise. Indeed for aggregating vectors of norm $1$ in $\Re^{\fulldim}$, using a $\gamma \fulldim$-sparse vector increases the sensitivity, and thus the required privacy noise, by a multiplicative $1/\gamma$. This leads to an unappealing trade-off between communication and privacy noise.

We show that constraining the sparsity structure to have $k$ non-zero blocks (each of size $B=m/k$) instead of $m$ non-zero coordinates overall can allow for significant performance gains for aggregation protocols in the two-party setting. We further show that this constraint comes at little cost: analyses of privacy amplification by sampling allow us to handle the increased sensitivity with little impact to the privacy-utility tradeoff. We give some details of each of these pieces next.

We show how to efficiently secret-share $k$-block-sparse vectors. Note that any such vector has at most $m$ non-zero coordinates, and thus is a sum of $m$ $1$-sparse vectors. Each of these can be shared using the Distributed Point Function (DPF) construction of~\cite{CCS:BoyGilIsh16}. This approach however leads to a communication cost of $m\lambda\log \fulldim$, where $\lambda$ is the security parameter. Thus e.g. when $\fulldim$ is more than a million and $\lambda=128$, the communication cost is at least 300 bytes per non-zero coordinate of the vector. For $m=50,000$, this amounts to $15MB$ of communication. (This basic scheme also incurs a very significant compute overhead, with the server cost being at least $O(m\fulldim)$, though that can be reduced by using optimizations based on probabilistic batch codes \cite{CCS:BCGI18} to $O(\fulldim+ m\log\fulldim)$ at a small additional overhead in communication cost.) Thus for $m$ being in the tens of thousands, approaches that only exploit sparsity (rather than {\em block}-sparsity) are far from feasible.

\ifconf
\begin{wrapfigure}[20]{L}{0.5\textwidth}
\else
\begin{figure}
\fi
    \begin{center}\vspace{-25pt}
    \includegraphics[width=0.7\linewidth]{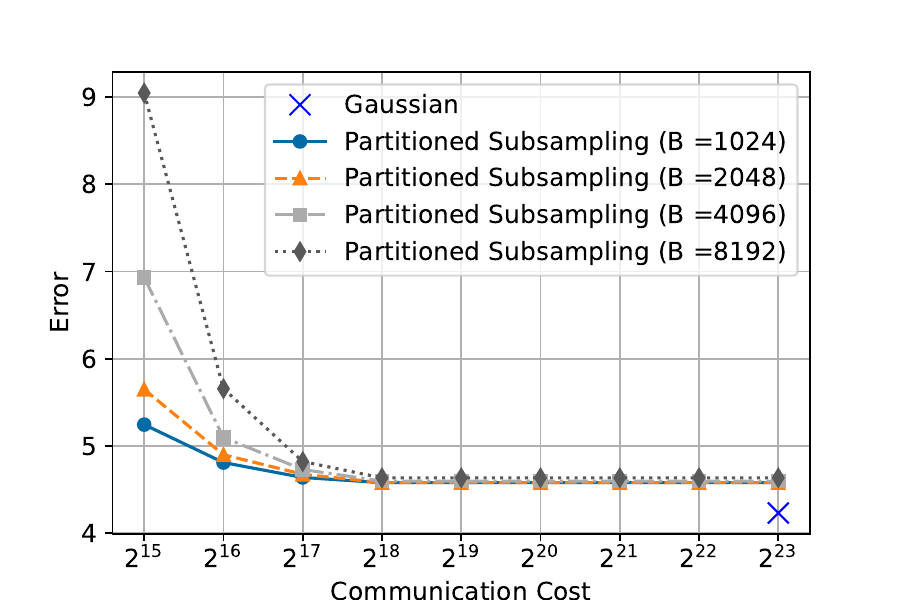}
    \caption{\small{The trade-off between expected squared error and per-client communication, when computing the sum of $n=10^5$ vectors in $\fulldim=2^{23}$ dimensions with $(1.0, 10^{-6})$-DP. The curves show our algorithm using different block sizes $B \in [10^3,10^4]$, and the blue 'x' shows the baseline approach of sending the whole vector. %
    We limit ourselves to algorithms for which the server run time is linear or near-linear in $\fulldim$.}}
    \label{fig:intro_plot}
   \end{center}
\ifconf
\end{wrapfigure}
\else
\end{figure}
\fi
We first observe that a simple modification of the DPF construction can exploit the block structure, reducing the communication cost to $O(m + k\lambda\log\fulldim)$. This also allows a bulk of the PRG evaluations to be in {\em counter mode} which can be more computationally efficient. This approach still requires $O(k\fulldim)$ server-side PRG evaluations. The use of probabilistic batch codes can reduce this computational overhead. We propose a different approach that uses cuckoo hashing within the DPF construction, that reduces the server's computation to $O(\fulldim)$ and reduces by $3\times$ the number of PRG evaluations the server has to do. See Section \ref{sec:related} for a comparison with prior work and with a concurrent and independent work.

We also provide efficient zero-knowledge proofs of validity for our construction of secret-shared $k$-block-sparse vectors. This allows the client to prove that all but $k$ of the blocks are zero vectors. The proofs do not change the asymptotic communication or the client (prover) or server (verifiers)  runtimes. In particular, the client runtime and communication depend only on the sparsity and on $\numblocks$, and not on $\fulldim$. The proof system combines techniques from prior works~\cite{BonehBCGI21} with a novel interactive proof showing that secret shares of PRG 
seeds in a DPF tree only differ in a small number of tree nodes (when the seeds are identical the PRG outputs ``cancel  out'', corresponding to a zero block). Finally, we can use an efficient zero-knowledge proof system from prior work ~\cite{RothblumOCT24} to also prove that  that the Euclidean norm of the vector of non-zero values is small.

As discussed above, naively applying random projection techniques to reduce communication results in an unappealing trade-off between communication and privacy noise. We avoid this overhead by exploiting the fact that the sparsity pattern of vectors is hidden from each server in our construction, as in some recent work~\cite{chen2023privacy,ChenIKNOX24}. This allows us to use {\em privacy amplification by sampling} techniques, at a block-by-block level, coupled with composition. For a large range of parameters, it allows us to reduce the overhead in the privacy noise, and asymptotically recover the bounds one would get if we communicated the full vector. \cref{fig:algsketch} presents an informal outline of our approach. We defer to \cref{sec:sampling} a discussion of different approaches to sampling, and their trade-offs. Coupled with numerical privacy analyses, we get the  reduction in communication cost essentially for free (\cref{fig:intro_plot}). %

 While our approach of subsampling blocks rather than coordinates is motivated here for the two-server setting, some of the benefits extend easily to the single trusted server setting where the cost of sending the index would now get amortized over a large block. The privacy analysis of the subsampling approaches requires some control of the norm of each coordinate, which gets relaxed in our approach to a bound on the $\ell_2$ norm of each block. The latter is a significantly weaker requirement. Our privacy analysis shows that block-based sampling nearly matches the privacy-utility trade-off of the Gaussian mechanism for a large range of block sizes.

\begin{Boxfig}{\small{Informal Description of our approach to Approximate Aggregation via $k$-block sparse vectors}}{algsketch}{Aggregate (Informal)}{{\bf Client Algorithm.} Input:  vector $v \in \Re^\fulldim$. Parameters: dimension $\fulldim$, blocksize $B$, sparsity $k$.}
    \begin{enumerate}
        \item Randomly select $I \subseteq \{1,2,\ldots,\fulldim/\blocksize\}$ with $|I|=k$.
        \item Define a $k$-block sparse vector $w$ which is equal to $(\tfrac{\fulldim}{k\blocksize})\cdot v$ in the blocks indexed by $I$ (i.e. coordinates $(j-1)B+1,\ldots,jB$ for $j \in I$) and $0$ everywhere else. This rescaling ensures that the expected value of $w$ is $v$.
        \item Use $k$-block-sparse DPF construction to communicate $w$ to the server. This needs communication $O(kB+k\lambda \log\fulldim/B)$.
    \end{enumerate}
    {\bf Servers} will decrypt to recover (secret-shares of) each vector $w$. They collaboratively compute the sum and add noise $\mathcal{N}(0, \sigma^2\mathbb{I}_d)$.
\end{Boxfig}

\ifconf
\else
\subsection*{Organization}
After discussion of additional related work (\cref{sec:related}), we start with some preliminaries in~\cref{sec:prelims}. \cref{sec:ksparsedf} overviews some ideas from our $k$-block-sparse DPF construction, including proofs of validity. \cref{sec:sampling} shows how this algorithm can be used to efficiently and privately estimate the sum of bounded-norm vectors, with sub-linear communication, including further efficiency improvements arising from numerical privacy accounting. We conclude in \cref{sec:conclusions} with further research directions.
\fi

 \subsection{Additional Related Work}
 \label{sec:related}
 As discussed above, Prio was introduced in~\cite{Corrigan-GibbsB17}, and triggered a long line of research on efficient validity proofs for various predicates of interest. Perhaps most related to our work is POPLAR~\cite{BonehBCGI21}, which uses Distributed Point Functions (DPFs) to allow clients to communicate $1$-sparse vectors in a high-dimensional space. They are used in that work to solve heavy hitters over a large alphabet, where the servers together run a protocol to compute the heavy hitters. In their work, the aggregate vectors is never constructed, and hence the DPFs there are optimized for {\em query} access: at any step, the servers want compute a specific entry of the aggregate vector.  In contrast, we are interested in the setting when the whole aggregate vector is needed, and this leads to different trade-offs in our use of DPFs.

We consider not just single-point DPFs, as used in POPLAR, but rather a generalization for a larger number of non-zero evaluations. \cite{CCS:BCGI18} also consider such multi-point functions and improve upon the naive implementation of a $k$-sparse DPF using $k$ single-point DPF instantiations. 
Multi-point function secret sharing has also been optimized using cuckoo hashing~\cite{SP:ACLS18,CCS:SGRR19,EC:dCaPol22}. They use cuckoo hashing to break one vector with $k$ non-zero entries into $k' > k$ smaller vectors with one non-zero entry each, after which standard DPF keys are constructed for those $k'$ vectors. Applying cuckoo hashing in this way incurs a roughly $2 \times$-$3\times$ overhead in the total number of vector entries, dependent on the number of hash functions chosen and the cuckoo hashing memory overhead.  While these prior work didn't focus on block-sparse DPFs, applying their methodlogy to block-sparse functions would result in a similar $2 \times$-$3\times$ overhead in terms of the total number of blocks and thus also the total number of large PRG evaluations (with output length $B$). Our construction applies cuckoo hashing differently and avoids this computational overhead. Finally, we remark that similar cuckoo hashing-based approaches~\cite{CCS:CheLaiRin17,PoPETS:DRRT18,JC:FHNP16,USENIX:PinSchZoh14} have also been used in the context of private set intersection.

Another recent work in this line of research~\cite{C:BBCGI23} uses projections, which may at first seem related to our work. Unlike our work, these arithmetic sketches are designed for the servers to be able to verify certain properties efficiently, without any help from the client. In contrast, our use of sketching is extraneous to the cryptographic protocol, and helps us reduce the client to server communication.

The fact that differential privacy guarantees get amplified when the mechanism is run on a random subsample (that stays hidden) was first shown by Kasiviswanathan et al.~\cite{Kasiviswanathan:2008} and has come to be known as privacy amplification by sampling. Abadi et al.~\cite{CCS:ACGMMT16}
first showed that careful privacy accounting tracking the moments of the privacy loss random variable, and numerical privacy accounting techniques can provide significantly better privacy bounds. In effect, this approach tracks the Renyi DP parameters~\cite{Mironov2019RnyiDP} or Concentrated DP~\cite{DworkR16,Bun:2016} parameters. Subsequent works have further improved numerical accounting techniques for the Gaussian mechanism~\cite{Balle2018ImprovingTG} and for various subsampling methods~\cite{Wang:2021,Balle:2018}. Tighter bounds on composition of mechanisms can be computed by more carefully tracking the distribution of the privacy loss random variable, and a beautiful line of work~\cite{KairouzOVa,DongRS22,MeiserM18,SommerMM19,KoskelaH21,GopiLW21,GhaziKKM22,ZhuDW22}. Numerical accounting for privacy amplification for a specific sampling technique we use has been studied recently in \cite{ChuaGHKKLMSZ25,FeldmanS25}.

The Prio architecture has been deployed at scale for several applications, including by Mozilla for private telemetry measurement~\cite{firefox} and by several parties to enable private measurements of pandemic data in Exposure Notification Private Analytics (ENPA)~\cite{ENPA:2021}. Talwar et al.~\cite{Talwar24} show how Prio can be combined with other primitives to build an aggregation system for differentially private computations.

We remark that the problem of vector aggregation has attracted a lot of attention in different models of differential privacy. Vector aggregation is a crucial primitive for several applications, such as deep learning in the federated setting~\cite{CCS:ACGMMT16,McMahanMRHA17}, frequent itemset mining~\cite{sun2014personalized}, linear regression~\cite{NguyenXYHSS16}, and stochastic optimization~\cite{ChaudhuriMS11,CheuJMP22}. The privacy-accuracy trade-offs for vector aggregation are well-understood for central DP \cite{ODOpaper}, for local DP~\cite{BeimelNO08,ChanSS12,DuchiR19,DuchiJW18,BhowmickDFKR18,AsiFT22}, as well as for shuffle DP~\cite{BalleBGN19,BalleBGN20,GhaziMPV20,GhaziKMPS21}. Several works have addressed the question of reducing the communication cost of private aggregation, in the local model~\cite{chen2020breaking,FeldmanTalwar21,AsiFNNT23}, and under sparsity assumptions~\cite{BassilyS15,FantiPE16,YeB18,AcharyaS19,Zhou0CFS22}. While the shuffle model can allow for accurate vector aggregation~\cite{GhaziKMPS21}, recent work~\cite{AsiFNNTZ24} has shown that for large $\fulldim$, the number of messages per client must be very large, thus motivating an aggregation functionality.

Our use of subsampling to reduce communication in private vector aggregation is closely related to work on aggregation in the single trusted server, secure multi-party aggregation and multi-message shuffling models ~\cite{chen2023privacy,ChenIKNOX24}. Aside from a different trust model, our work crucially relies on blocking to reduce the sparsity. Sparsity constraints also make regular Poisson subsampling less suitable for our application and necessitate the use of sampling schemes that require a more involved privacy analysis.%

\paragraph{Very recent works.} After a preliminary version of our work was made public, we were made aware of an independent and  concurrent work~\cite{BoyleGHIY25} that studies the problem of communicating $m$-sparse vectors in the two-server setting. They compare three different schemes for this task and empirically compare them. Their ``big-state'' DMPF is equivalent to our naive scheme and their probabilistic batch codes (PBC) construction uses cuckoo hashing in a black-box way, similarly to the prior works discussed above.
They also propose a new scheme based on Oblivious Key-value stores (OKVS), which obtains constant-factor improvements in the server runtime compared to the other two schemes for some range of sparsity. To keep the server-side computation low, the OKVS-based schemes (which allow a range of tradeoff) incur at least a $2 \times$ communication overhead.

Beyond these differences in various performance measures, one of our main contributions is identifying block-sparsity as a property that lends itself both to significant savings in DPF constructions and to efficient privacy-preserving aggregation. \cite{BoyleGHIY25} don't focus on block-sparse DPFs, but their approaches can be applied towards block-sparse constructions. The cuckoo-hashing based constructions would incur a $3\times$ overhead in the number of large PRG evaluations (as discussed above). The OKVS-based schemes incur the above $2 \times$ communication overhead. %

Another recent work \cite{EPRINT:KKEP24} also provides an efficient multi-point DPF construction (they also don't focus on block-sparsity). Their scheme requires the client to solve linear systems with $k$ equations and at least $(k + \lambda)$ unknowns over a field of size $2^{\lambda}$. So even if the blocks are small, for large $k$ the client work is larger than in our scheme.

\section{Preliminaries}\label{sec:prelims}

\begin{definition}{Function Secret Sharing~\cite{EC:BoyGilIsh15}}
    Let $\mathcal{F}=\{f:I\rightarrow\mathbf{G}\}$ be a class of functions with input domain $I$ and output group $\mathbb{G}$, and let $\lambda\in\mathbb{N}$ be a security parameter. A 2-party function secret sharing (FSS) scheme is a pair of algorithms with the following syntax:
    \begin{itemize}
        \item $Gen(1^\lambda,f)$ is a probabilistic, polynomial-time key generation algorithm, which on input $1^\lambda$ and a description of a function $f$ outputs a tuple of keys $(k_0,k_1)$. 
        \item $Eval(i,k_i,x)$ is a polynomial-time evaluation algorithm, which on input server index $i\in\{0,1\}$, key $k_i$, and input $x\in I$, outputs a group element $y\in\mathbb{G}$.
    \end{itemize}    
    Given some allowable leakage function $Leak:\{0,1\}^*\rightarrow\{0,1\}^*$ and a parameter $\gamma \in [0,1]$, we require the following two properties:
    \begin{itemize}
        \item Correctness: For any $f\in\mathcal{F}$ and any $x\in I$, we have that \[Pr[\sum_{b\in\{0,1\}}Eval(b,k_b,x)=f(x)] \geq 1 - \gamma.\]
        \item Security: For any $b\in\{0,1\}$, there exists a ppt simulator such that for any polynomial-size function $f\in\mathcal{F}$:
        \[ \{k_b|(k_0,k_1)\leftarrow Gen(1^\lambda,f)\} \sim \{k_b \leftarrow Sim_b(1^\lambda,Leak_b(f))\}  \]
    \end{itemize}
\end{definition}
Note that we have defined a relaxed notion of correctness here, where we allow a small probability $\gamma$ of incorrect evaluation. While $\gamma$ will be $0$ for some of our constructions, the most efficient version of our protocol will have a $\gamma$ that is polynomially small in the sparsity $k$. 

We define 1-sparse and $k$-block-sparse distributed point functions, which are instantiations of FSS for specific families of sparse functions:

\begin{definition}[Distributed Point Function (DPF)]\label{def:dpf}
    A point function $f_{\alpha,\beta}$ for $\alpha\in\{0,1\}^d$ and $\beta\in\mathbb{G}$ is defined to be the function $f:\{0,1\}^{d}\rightarrow\mathbb{G}$ such that $f(\alpha)=\beta$ and $f(x)=0$ for $x\neq \alpha$. A DPF is an FSS for the family of all point functions.
\end{definition}

\begin{definition}[$k$-block-sparse DPF]\label{def:kblockdpf}
    A $k$-block-sparse function $f_{\alpha,\beta}$ with block size $B$ for $\alpha=\{\alpha^0,\dots,\alpha^{k-1}\}$, where $\alpha^i\in\{0,1\}^d$ and $\beta=\{\beta^0,\dots,\beta^{k-1}\}$ and $\beta^i=\{\beta^i_0,\dots,\beta^i_{\blocksize-1}\}\in\mathbb{G}^{B}$ is defined to be the function $f:\{0,1\}^{d+\log \blocksize}\rightarrow\mathbb{G}$ such that $f(\alpha^i||j)=\beta^i_j$ and $f(x)=0$ for $x\neq \alpha^i||j$ with $j\in[\blocksize]$ and $i\in[k]$. A $k$-block-sparse DPF is an FSS for the family of all $k$-block-sparse functions.
\end{definition}

\paragraph{Cuckoo Hashing. }
Cuckoo Hashing~\cite{PF01} is an algorithm for building hash tables with worst case constant lookup time. The hash table depends on two (or more) hash functions,  and each item is placed in location specified by one of the hash values. When inserting $k$ items from a universe $U$ into a cuckoo hash table of size $\tilde{k}$, the hash functions map $U$ to $[\tilde{k}]$. When the hash functions are truly random, as long as $k$ is a constant fraction of $\tilde{k}$, there is a way to assign each item to one of its hash locations while avoiding any collision: %

\begin{theorem}{Cuckoo Hashing~\cite{DK12}}
    Suppose that $c\in\{0,1\}$ is fixed. The probability that a cuckoo hash of $k=\lfloor (1-c)\tilde{k}\rfloor$ data points into two tables of size $\tilde{k}$ succeeds is equal to:
    \[
    1-\frac{(2c^2-5c+5)(1-c)^3}{12(2-c)^2c^3}\frac{1}{\tilde{k}} +\mathcal{O}\left(\frac{1}{\tilde{k}^2}\right)
    \]
\end{theorem}

Note that if $c\approx 0.32$, this simplifies to about $1-1/\tilde{k}-\mathcal{O}(1/\tilde{k}^2)$. Therefore, this choice of $c$ leads to a failure probability of $\Tilde{\mathcal{O}}(1/\tilde{k})$. Variants of cuckoo hashing where more than 2 hash functions are used can improve the space efficiency of the data structure. E.g. using 4 hash functions allows for a an efficiency close to $0.97$ with probability approaching 1~\cite{FoutoulakisP12,FriezeM12}.

We will be working with vectors $v \in \Re^D$, and in the linear algebraic parts of the paper, we view them as $v_1,\ldots,v_\fulldim$. As is standard, the cryptographic parts of the paper will view the vectors as coming from a finite field $\mathbb{G}=\mathbb{F}_q$; for a suitably large $q$, the quantized versions of $n$ real vectors can be added without rollover so that we get an estimate of the sum of quantized vectors. 

We will use $\blocksize$ to represent the block size and $\Delta = \fulldim/\blocksize$ will be the number of blocks. We will assume $\Delta$ is a power of $2$ and $d$ will denote $\log_2 \Delta$. $\lambda$ will denote our security parameter. In the cryptographic part of the paper we will view a vector as a function from $\{0,1\}^d \times [B] \rightarrow \mathbb{G}$ where $[B]=\{0,1,\ldots,B-1\}$.

We recall the definition of $(\eps,\delta)$-indistinguishability and differential privacy.
\begin{definition}
    Let $\eps>0, \delta \in [0,1)$, and let $Y$ and $Y'$ be two random variables. We say that $Y$ and $Y'$ are $(\eps, \delta)$-indistinguishable %
    if for any measurable set $S$, it is the case that
    \begin{align*}
        \Pr[Y' \in S] &\leq e^{\eps} \Pr[Y \in S] + \delta \\
        \mbox{ and }         \Pr[Y \in S] &\leq e^{\eps} \Pr[Y' \in S] + \delta.
    \end{align*}
\end{definition}

\begin{definition}[Differential Privacy~\cite{Dwork:2006}]
    Let $\mathcal{A} : \mathcal{D}^* \rightarrow \mathcal{R}$ be a randomized algorithm that maps a dataset $X \in \mathcal{D}^*$ to a range $\mathcal{R}$. We say two datasets $X,X'$ are neighboring if $X$ can be obtained from $X'$ by adding or deleting one element. We say that $\mathcal{A}$ is $(\eps,\delta)$-differentially private if for any pair of neighboring datasets $X$ and $X'$, the distributions $\mathcal{A}(X)$ and $\mathcal{A}(X')$ are $(\eps, \delta)$-indistinguishable. 
\end{definition}
We will use the following standard result.
\begin{theorem} [Gaussian Mechanism~\cite{ODOpaper,DworkR14}]
\label{thm:gaussian}
Let $\eps,\delta \in (0,1)$. Let $\mathcal{A} : (\Re^\fulldim)^* \rightarrow \Re^\fulldim$ be the mechanism that for a sequence of vectors $v_1,\ldots,v_n \in \Re^\fulldim$ outputs $\sum_i v_i + \mathcal{N}(0, \sigma^2\mathbb{I}_\fulldim)$. If for all $i\in [n]$, the input $v_i$ is restricted to $\|v_i\|_2 \leq s$ and $(\sigma/s)^2 \geq 2 \log (1.25/\delta) /\eps^2$, then $\mathcal{A}$ is $(\eps,\delta)$-differentially private. We refer to $\sigma$ as the scale of the mechanism. In this context $s$ is the sensitivity of the sum to adding/deleting an element. 
\end{theorem}

\section{Private aggregation via $k$-block-sparse vectors}
\label{sec:sampling}
In this section we propose two instantiations of our sampling-based sparsification, and describe the formal privacy and utility guarantees of the resulting aggregation algorithms.
In both schemes each user is given a vector $v \in \Re^{\fulldim}$ which consists of $\Delta= \fulldim/B$ blocks, each of size $B$. We refer to the value of $v$ on block $i \in [\Delta]$ by $v_i \in \Re^{\blocksize}$. Our subsampling schemes are parametrized by an upper bound on the $\ell_2$ norm of each block $L$ and the number of blocks $k$ to be sent by each user. We will also assume, for simplicity, that $k$ divides $\Delta$. The bound $L$ is somewhat stronger than a total bound on the $\ell_2$ norm of the input that is typically assumed in mean estimation. A number of standard techniques are known for converting an $\ell_2$ norm bounded vector to an $\ell_\infty$-bounded vector such as Kashin representation~\cite{Lyubarskii:2010}. We show that techniques from ~\cite{AsiFNNT23} can be used to convert a vector of $\ell_2$ norm 1 to a vector in which each block has norm $\sqrt{B/D}$ while incurring expected squared error on the order of $1/B$. Note that this result relies crucially on the fact that ensuring that block norms are upper-bounded is easier than ensuring that each coordinate norm is upper-bounded. We formally prove that the expected error due to truncation falls as $\tilde{O}(\frac 1 B)$ in \cref{app:truncation_analysis}, and evaluate this impact empirically in ~\cref{sec:experiments}.

\ifconf
\else
Our subsampling schemes differ in how the $k$ blocks are subsampled.  In \cref{fig:kwiseonesparse}, we partition the blocks into $k$ groups, and pick one block out of each group, randomly and independently. In \cref{fig:truncatedpoisson}, we select each block with probability $q$, to get in expectation $q\Delta$ blocks. If the number of blocks that end up getting subsampled is larger than $k$, we keep at random $k$ of them. This gives at most $k$ non-zero blocks, which can be communicated using our $k$-block-sparse DPF. Both schemes give about the same expected utility, and privacy bounds that match asymptotically. The partitioned subsampling results in a more structured $k$-block-sparse vector, which is a concatenation of $k$ $1$-block-sparse vectors, which are simpler to communicate. The truncated Poisson subsampling has more randomness, and can be analyzed numerically using a larger set of tools, though it may have slightly higher communication cost. As we discuss in~\cref{sec:experiments}, each of these may be preferable over the other for some range of parameters.%
\fi

\ifconf\else
\begin{Boxfig}{$k$-wise 1-sparse sampling scheme}{kwiseonesparse}{SampledVector (Partitioned Subsampling)}{Input: vector $v \in \Re^\fulldim$. Parameters: dimension $\fulldim$, blocksize $B$, sparsity $k$, $\Delta = \fulldim/\blocksize$, blockwise $\ell_2$ bound $L$.}
\begin{enumerate}
\item Split the block indices into $k$ consecutive subsets $S_j = \{(j-1)\Delta/k+1,\ldots, j\Delta/k\}$ for $j\in [k]$.
\item Select an index $i_j$ randomly and uniformly from $S_j$ and define $I = \{i_j\}_{j\in[k]}$
\item Define the subsampled (and clipped) $v^I$ as 
$$v^I_i = \begin{cases} 
\frac{\Delta}{k} \cdot \mbox{clip}_L(v_i) & \text{if } i \in I \\
0 & \text{otherwise}
\end{cases}, 
$$
where $\mbox{clip}_L(z)$ is defined as $z$ if $\|z\|_2 \leq L$ and $\frac{L}{\|z\|_2} \cdot z$, otherwise. 
\item Output $w = v^I$.
\end{enumerate}
\end{Boxfig}

\begin{Boxfig}{$k$-sparse sampling scheme}{truncatedpoisson}{SampledVector (Truncated Poisson Subsampling) }{Input: vector $v \in \Re^\fulldim$. Parameters: dimension $\fulldim$, blocksize $B$, sparsity $k$, sampling probability $q$, $\Delta = \fulldim/\blocksize$, blockwise $\ell_2$ bound $L$.}
\begin{enumerate}
\item Select a subset $I_0$ by picking each coordinate in $\{1,\ldots,\Delta\}$ independently with probability $q$.
\item If $|I_0| > k$, let $I$ be random subset of $I_0$ of size $k$. Else set $I=I_0$.
\item Let $\kappa = \E[|I|] = \E[\min(\texttt{Bin}(\Delta,q), k)]$ under this sampling process.
\item Define the subsampled (and clipped) $v^I$ as 
$$v^I_i = \begin{cases} 
\frac{\Delta}{\kappa} \cdot \mbox{clip}_L(v_i) & \text{if } i \in I \\
0 & \text{otherwise}
\end{cases}. $$
\item Output $w = v^I$.
\end{enumerate}
\end{Boxfig}
\fi

\begin{theorem}[Utility of the subsampling schemes]
\label{thm:utility}
Let $v^1,\ldots, v^n$ be a collection of vectors such that each $v^j=v^j_1,\ldots,v^j_\Delta \in \Re^\fulldim$ and $\|v^j_i\|_2 \leq L$ for all $j\in [n]$ and $i\in[\Delta]$. Let $w^j$ denote the (randomized) report of either of our subsampling algorithms for each $j \in [n]$ and let $W = \sum_{j\in [n]} w^j + N(\bar{0}, \sigma^2 I_\fulldim)$ be their noisy aggregate. Then
$$\mathbf{E}[\|W - \sum_{j\in [n]} v^j\|_2^2] \leq  n L^2 \frac{\Delta^2}{k} + \sigma^2 \cdot \fulldim .$$
\end{theorem}
\ifconf\else
\begin{proof}
    First, note that $\mathbf{E}[w^j] = v^j$ for all $j \in [n]$.
    Therefore we have 
    \begin{align*}
    \mathbf{E}[\|W - \sum_{j\in [n]} v^j\|_2^2] 
        & = \mathbf{E}\left[\|\sum_{j\in [n]} w^j - v^j + N(\bar{0}, \sigma^2 I_\fulldim)\|_2^2\right] \\
        & = \sum_{j=1}^n \mathbf{E}\left[\|w^j - v^j\|_2^2\right] + \sigma^2 \fulldim \\
        & = n \mathbf{E}\left[\|w^1 - v^1\|_2^2\right] + \sigma^2 \fulldim \\
        & \le n \Delta \tfrac{k}{\Delta}(1-\tfrac{k}{\Delta}) L^2 (\frac{\Delta}{k})^2 + \sigma^2 \fulldim\\
        & \le n  \cdot L^2 (\Delta^2/k) + \sigma^2 \fulldim.
    \end{align*}
    Here we have used the fact that the variance of the $w^1$ decomposes across $\Delta$ blocks, where each block contributed $p(1-p)$ times the square of its value when non-zero (which is $L(\Delta/k)$) and $p=\tfrac k \Delta$ is the probability of a block being non-zero.
\end{proof}
\fi

Theorem~\ref{thm:utility}\  provides guidance on how to set the smallest communication cost so that the sub-sampling error is negligible compared to the privacy error. Indeed, assume for simplicity that our vectors lie in the $\ell_\infty$-ball $v^j \in \left[\frac{-1}{\sqrt{\fulldim}},\frac{1}{\sqrt{\fulldim}}\right]^\fulldim$. This implies that the norm of each block is upper bounded by $L = \sqrt{B/\fulldim}$. Noting that $\Delta = \fulldim/B$, the (upper bound above on the) error of the algorithm is $n\Delta/k + \sigma^2 \fulldim $. This implies that we should set $ k B\ge c \cdot n / \sigma^2$ for some constant $c > 0$. In other words, the number of coordinates that are non-zero must be set to $cn/\sigma^2$, which is independent on $\fulldim$, as well as of the blocksize. Thus block-based subsampling does not impact the sampling noise, and any setting of $k$ and $B$ with the product $kB \geq cn/\sigma^2$ would suffice. 
    
We also note that the proof is oblivious to the subsampling method and only uses the fact that marginally, each coordinate has the right expectation, and is non-zero with probability $k/\Delta$. Thus it applies to a range of subsampling methods.

We now analyze the privacy of our subsampling methods using the standard notion of differential privacy \cite{Dwork:2006} with respect to deletion of user data. In the context of aggregation, we can also think of this adjacency notion as replacing the user's data with the all-0 vector.
We state the asymptotic privacy guarantees for our partitioned subsampling method below (the proof can be found in Appendix \ref{sec:privacy-proof}). Similar guarantees hold for the Truncated Poisson subsampling with an appropriate choice of parameters and can be derived from the known analyses \cite{CCS:ACGMMT16,steinke2022composition} together with the fact that the truncation operation does not degrade the privacy guarantees \cite{FeldmanS25}. 
\ifconf
\else
\begin{theorem}[Privacy of partitioned subsampling]
\label{thm:asymptotic--partitioned-privacy}
	Let $v^1,\ldots, v^n$ be a collection of vectors such that each $v^j=v^j_1,\ldots,v^j_\Delta \in \Re^\fulldim$ and $\|v^j_i\|_2 \leq L$ for all $j\in [n]$ and $i\in[\Delta]$. Let $w^j$ denote the (randomized) report using Partitioned Subsampling, for each $j \in [n]$ and let $A$ be an algorithm that outputs $W = \sum_{j\in [n]} w^j + N(\bar{0}, \sigma^2 I_\fulldim)$. Then there exists a constant $c$ such that for every $\delta >0$, if $$\frac{\sigma}{L} \geq c \cdot \max \left\{\frac{ \Delta \sqrt{\log(\Delta/\delta)}}{k} , \sqrt{\frac{\Delta}{k}} \log(\Delta/\delta), \sqrt{\Delta \log(1/\delta)} \right\} ,$$ then $A$ is $(\epsilon,\delta)$-differentially private for 
	$$ \eps = O\left(\frac{L\sqrt{\Delta} \sqrt{\log(1/\delta)}}{\sigma  }\right).$$ 
\end{theorem}
\fi

Note that this implies that when $L  = \sqrt{1/\Delta}$ and we set $kB$ to be $n/\sigma^2$, the resulting $\eps$ is $O(\sqrt{\log(1/ \delta)}/\sigma)$, which asymptotically matches the bound for the Gaussian mechanism. In other words, the privacy cost of our algorithm is close to that of the Gaussian mechanism, when $kB \geq n/\sigma^2$ and $k \geq \sqrt{\Delta\log \frac 1 \delta}/\sigma$. For a fixed $kB$, this constraint translates to an upper bound on the block size. %

This asymptotic analysis demonstrates the importance of hiding the sparsity pattern. Specifically, without hiding the pattern we cannot appeal to privacy amplification by subsampling and need to rely on the sensitivity of the aggregated value. The sensitivity is equal to $\sqrt{k} \cdot \frac{L\Delta}{k}=\frac{L\Delta}{\sqrt{k}}$. By the properties of the Gaussian noise addition (Thm.~\ref{thm:gaussian}), we obtain that the algorithm is $\left(O\left(\frac{L\Delta\sqrt{\log(1/\delta)}}{\sigma \sqrt{k}}\right),\delta\right)$-DP. This bound is worse by a factor of $\sqrt{\frac{\Delta}{k}}$ than the bound we get in Theorem \ref{thm:asymptotic--partitioned-privacy}. A similar gain was obtained for private aggregation of Poisson subsampled vectors in \cite{chen2023privacy}.

\ifconf\else
\paragraph{Communicating $1$-sparse vectors:} A common application of secure aggregation systems is to aggregate vectors that are $1$-sparse (often known as $1$-hot vectors) or $k$-sparse for a small $k$. Directly using DPFs for these vectors requires $O(\fulldim)$ PRG re-seedings and thus can be expensive. Noting that  $k$-sparse vector is also $k$-block sparse, one can directly use PREAMBLE to reduce the server computation cost at a modest increase in communication cost. Alternately, in settings where we want to add noise in a distributed setting, {\sc RAPPOR}~\cite{ErlingssonPK14} and its lower-communication variants such as {\sc PI-RAPPOR}~\cite{FeldmanTalwar21} and ProjectiveGeometryResponse~\cite{FeldmanNNT22} can be used along with Prio.  Since $1$-hot vectors become vectors  in $\left\{\frac{1}{\sqrt{\fulldim}}, \frac{-1}{\sqrt{\fulldim}}\right\}^\fulldim$ after a Hadamard transform, one can view these vectors as Euclidean vectors in $\Re^\fulldim$ and use PREAMBLE to efficiently communicate them.
ProjectiveGeometryResponse is particularly well-suited for this setup even without sampling, as the resulting message space is $O(\fulldim)$-dimensional, and a linear transformation of the input space. Thus one can aggregate in ``message space'': each message is a $1$-hot vector in message space, and we can add up these vectors using PREAMBLE. The linear transform to go back to data space is a simple post-processing and the privacy guarantee here can use privacy amplification by shuffling. Since we avoid sampling in this approach, it can scale to larger $n$ without incurring any additional utility overhead.
\fi
\ifconf
\else
\subsection{Numerical Privacy Analysis}
\label{sec:numerical}
As discussed above, under appropriate assumptions on the block size, we can prove asymptotic privacy guarantees for our approach that matches that of the Gaussian mechanism itself. In this section, we evaluate the privacy-utility trade-off of our approach using numerical techniques to compute the privacy cost.

We will analyze the two approaches described above. For both schemes, the noise due to subsampling can be upper bounded using~\cref{thm:utility}. For partitioned subsampling, we use the recent work on Privacy Amplification by Random Allocation~\cite{FeldmanS25} to analyze the privacy cost. The authors show that the R\'enyi DP parameters of the one-out-of-$m$ version of the Gaussian mechanism can be bounded using numerical methods. Composing this across the $k$ groups gives us the Renyi DP parameters, and hence the overall privacy cost of the full mechanism.

For the approach based on truncated Poisson subsampling, it is shown in~\cite{FeldmanS25} that the privacy cost is no larger than that of the Poisson subsampling version without the truncation. The Poisson subsampling can then be analyzed using the PRV accountant of ~\cite{GopiLW21}. While the PRV accountant usually gives better bounds than one can get from RDP-based accounting, we suffer a multiplicative overhead as the scaling factor $\kappa$ in Truncated Poisson is smaller than $k$. For the numerical analysis, we optimize over $q$ to control the overall variance one gets from this process. Note that when $k$ is large and $q \approx \frac{k-O(\sqrt{k})}{\Delta}$, one would expect truncation to be rare, and thus $\kappa$ to be close to $q\Delta$ and thus to $k$. 

One can also study a different subsampling process where $I$ is a uniformly random subset of size $k$. 
We may expect better privacy bounds to hold for this version, intuitively as there is more randomness compared to partitioned subsampling. Feldman and Shenfeld~\cite{FeldmanS25} show that the privacy bound of this variant is no larger than that of partitioned subsampling. We conjecture that the privacy cost of this version is closer to (untruncated) Poisson with sampling rate $q=k/\Delta$. Since the latter can be more precisely accounted for using the PRV accountant, we expect careful numerical accounting of this version to do better than the RDP-based bounds for partitioned subsampling.
\fi
\ifconf\else
\paragraph{Ensuring block-wise norm bound:}
We now formally show that one can reduce the problem of computing means of $\ell_2$ bounded vectors to our setting of block-wise bounded norm. Without loss of generality, we can assume that each input vector has $\ell_2$ norm 1. Our main reduction is based on the techniques developed by Asi {\em et al.}~\cite{AsiFNNT23} in the context of communication-efficient algorithms for mean estimation with local differential privacy. Their randomizer for mean estimation relies on a randomized dimensionality reduction followed by an optimal differentially private randomizer in lower dimension referred to as {\tt PrivUnit}. {\tt PrivUnit} requires a vector of unit length as an input, whereas randomized dimensionality reductions they use result in vectors of varying lengths. Asi {\em et al.}~ apply scaling to ensure that the norm condition is satisfied and developed several techniques for the analysis of the error resulting from this  step. We observe that their dimensionality reductions can be used just as (randomized) linear maps (in the same dimension) with each block of $B$ coordinates in the image corresponding to projection of the original input into $B$ dimensions. Thus we can apply clipping/scaling to ensure that block norms are upper bounded and then analyze the resulting error in essentially the same way as in \cite{AsiFNNT23}. We defer full proofs to \cref{app:truncation_analysis}.
Our first application of this approach shows that a random rotation with simple block norm clipping to $\sqrt{B/D}$ achieves expected squared error of $1/B$.

\begin{theorem}
For a vector $v=v_1,\ldots,v_\Delta\in \Re^D$, where $v_i\in \Re^B$ let $\mathtt{blkscale}_B(v)$ denote the vector $u=u_1,\ldots,u_\Delta$, such that for every $i\in [\Delta]$, $u_i=\sqrt{B/D} \cdot \frac{v_i}{\|v_i\|}$. Let $U\in \Re^{D\times D}$ be a randomly and uniformly chosen rotation matrix. Then 
for every $v\in \Re^D$ such that $\|v\|_2=1$, $$\E_U\left[\left\|U^\top \mathtt{blkscale}_B(Uv) - v\right\|_2^2 \right]= O\left(\frac{1}{B}\right) .$$
\end{theorem}
While this method is simple to describe and analyze, it is relatively inefficient as it requires $D^2$ multiplications. We also show that a significantly more efficient scheme from \cite{AsiFNNT23} based on Subsampled Randomized Hadamard Transform (SHRT) can also be easily adapted to our setting. Specifically, let $W_H=SHT$ denote the following distribution over random matrices: $H \in \Re^{D\times D}$ is the Hadamard matrix, $S \in \Re^{D\times D}$ is a random permutation matrix, and $T \in \Re^{D\times D}$ is a diagonal matrix where $T_{i,i}$ are independent samples from the Rademacher distribution (that is, uniform over $\pm1$). An important (and well-known) property of this family of matrices is that multiplication by $W_H$ and $W_H^\top$ can be performed in time $O(D\log D)$.
\begin{theorem}
Let $W_H\in \Re^{D\times D}$ be a randomly and uniformly chosen $SHT$ matrix as described above. Then 
for every $v\in \Re^D$ such that $\|v\|_2=1$, $$\E_{W_H}\left[\left\|W_H^\top \mathtt{blkscale}_B(W_Hv) - v\right\|_2^2 \right]= O\left(\frac{\log^2 D}{B}\right) .$$ Further, multiplication by $W_H$ and $W_H^\top$ can be performed in time $O(D\log D)$.
\end{theorem}
\fi

\section{$k$-block-sparse-DPF Construction: Main Ideas}
\label{sec:ksparsedf-overview}
We sketch here the main ideas of our DPF construction, deferring full details to \cref{sec:ksparsedf}.
Our construction builds on the tree-based DPF construction from~\cite{CCS:BoyGilIsh16} for sharing a 1-sparse function $f : \{0,1\}^d \rightarrow \{0,1\}$. $f$ outputs` 0 on all but at most one input $\alpha \in \{0,1\}^d$. At a very high level, in their construction the client shares with each server a seed to a pseudo-random generator. Each server uses their seed to expand out an entire tree with $2^d$ leaves, where each node in the tree contains the seed to a PRG and some additional ``control bits''. The client also sends a (public) ``correction word'' for each layer of the tree. The control bits and the correction word are used to ensure that in each layer, the strings in each node are secret shares of the zero string, except for the single node in that layer that is on the path from the root to the leaf $\alpha$ corresponding to the non-zero output of $f$.

In our construction of $k$-sparse DPFs, there are $k$ non-zero nodes in each layer: the nodes that are on the paths from the root to the $k$ non-zero leaves. A naive suggestion could be to also use $k$ correction words for each layer, but this increases the server's computation by a $k$-fold multiplicative overhead (in a nutshell, each server would need to apply a correction procedure for each node and for each correction word). Instead, we use cuckoo hashing and a collection of slightly more than $k$ correction words per layer, where each node in the layer only needs to apply a correction procedure to two correction words that are relevant to it (the relevant correction words are determined by two hash functions that are public). This reduces the server's computation to 2 corrections per node (the correction procedure is just an XOR: it is quite lightweight). We can also use 3 or 4 hash functions to reduce the number of correction words per layer to be very close to $k$, see Remark~\ref{remark:numfuncs}.

To handle {\em block-sparse} functions we modify the tree: each leaf now corresponds to an entire block (and all but $k$ of them will be zero). The servers expand the seed corresponding to each leaf into an entire block, and the correction words for the last layer ensure that these expansions are secret shares of the correct output (thus the correction words for the last layer are of length $B$). This avoids the high cost of repeating the bit-output construction $\blocksize$ times. Rather than paying a $\blocksize$-multiplicative overhead in the key size, we pay the cost of a single bit-output construction, plus $O(k \cdot \blocksize)$ for the correction bits in the last layer. We also only have only a single PRG evaluation per block/leaf (with a large output). This approach reduces the number of large PRG evaluations by a factor of 2-3 compared to techniques from prior work that used cuckoo hashing to construct sparse DPFs (see Section \ref{sec:related} and Remark \ref{remark:cuckoo}). As noted above, if we use 3 or 4 hash functions for cuckoo hashing, then the number of correction words for the last layer is very close to $k$ and the communication complexity approaches $k \cdot B$.

\paragraph{Zero-knowledge proofs of validity.} We construct an efficient proof-system that allows a client to prove that it shared a valid block-sparse DPF. The proof is divided into two components:
\begin{enumerate}
\item $k$ correction-bit sparse. The client proves that at most $k$ of the (secret shared) correction bits corresponding to leaves in the tree are non-zero. 

These are just bits so here we can use an efficient proof systems of \cite{CCS:BoyGilIsh16} for sharing standard DPFs (the client also needs to send 1-sparse DPFs whose sum is the vector of control bits).

\item $k$ block-sparse. Given that at most $k$ of the correction bits for the leaves are non-zero, the client proves that there are at most $k$ non-zero blocks in the output.

Consider the final layer of the DPF tree: in a zero block, the two PRG seeds held by the servers are identical (shares of the zero string), whereas in a non-zero block, they are different. Rather than expanding the seeds to $B$ group elements (as in the vanilla construction above), we add another $\lambda$ bits to the output, and we also add $\lambda$ corresponding bits to each correction word. We refer to these as the check-bits of the PRG outputs / correction words, and we refer to the original outputs as the payload bits. In the zero blocks the check-bits of the outputs should be identical: subtracting them should results in a zero vector. In each non-zero block, the check-bits of the (appropriate) correction word are chosen so that the correction procedure will result in an all-zero check-bit string. Thus, in our proof system, the servers verify that, in each block, the secret-shared check-bits are indeed all zero. This can be done quite efficiently. 

\end{enumerate}

The additional cost for the proof (on top of the construction above) is $O(k \cdot d \cdot \lambda)$ communication, $(k \cdot \poly(d,\lambda))$ client work, $(k \cdot 2^d \cdot \poly(d,\lambda))$ server work for each server. The proof system is sound against a malicious client, but we assume semi-honest behavior by the servers.

\section{Experimental Evaluation}
\label{sec:experiments}
In this section, we give empirical evidence demonstrating that PREAMBLE is accurate, private, and communication efficient.

\paragraph{\bf Blocking improves communication cost:}

Figure~\ref{fig:comm_cost} shows the required communication, or DPF key size, for $\lambda=128$, a group of size $2^{64}$ in the final tree layer. Communication depends on the sparsity $k$, block size $\blocksize$, and data dimension $D$. We vary $\blocksize$ and hold $k\blocksize=2^{18}$ constant. %
We also include a baseline communication comparison, which is the minimum communication required to communicate $k\blocksize$ group elements of size $\log |\mathbb{G}|=\lambda$ and the $\fulldim/\blocksize$ indices of the non-zero blocks to a single trusted server. Our plots are for $4$-way cuckoo hashing, where the cuckoo-hashing overhead is about 1.03. Using fewer hashes would increase the key size, but lead to more efficient server computations.

\paragraph{Blocking reduces Truncation Error}
We analyze empirically the ease of ensuring a block-wise norm bound. Theoretically, a norm bound of $O(\frac{1}{\sqrt{D}})$ on each entry can be ensured if one transforms to a $O(D)$-dimensional space; the hidden constants in the two $O(\cdot)$ notations are related, where making one smaller makes the other larger. In practice, a simple approach often used is to apply a random rotation to the vector, followed by truncating any entry that is larger than $\frac{c}{\sqrt{D}}$, for an appropriate constant $c$. For moderate values of $c$, the {\em truncation error} (i.e. the norm of the induced error) is small for a random rotation. Imposing the weaker condition that each block has $\ell_2$ norm at most $\frac{c}{\sqrt{D/B}}$ results in a lower truncation error. In \cref{fig:truncation-error}, we plot the truncation error as a function of $c$, when we apply a random rotation, for different values of the block size $B$. It is easy to see that $c<1$ will result in non-trivial truncation error for any block size. Our plots show that even moderately large $B$ allow us to take $c$ very close to $1$ for a negligible truncation error.

\begin{figure}[H]
    \centering
    \begin{subfigure}{0.29\textwidth}
        \centering
        \includegraphics[width=\textwidth]{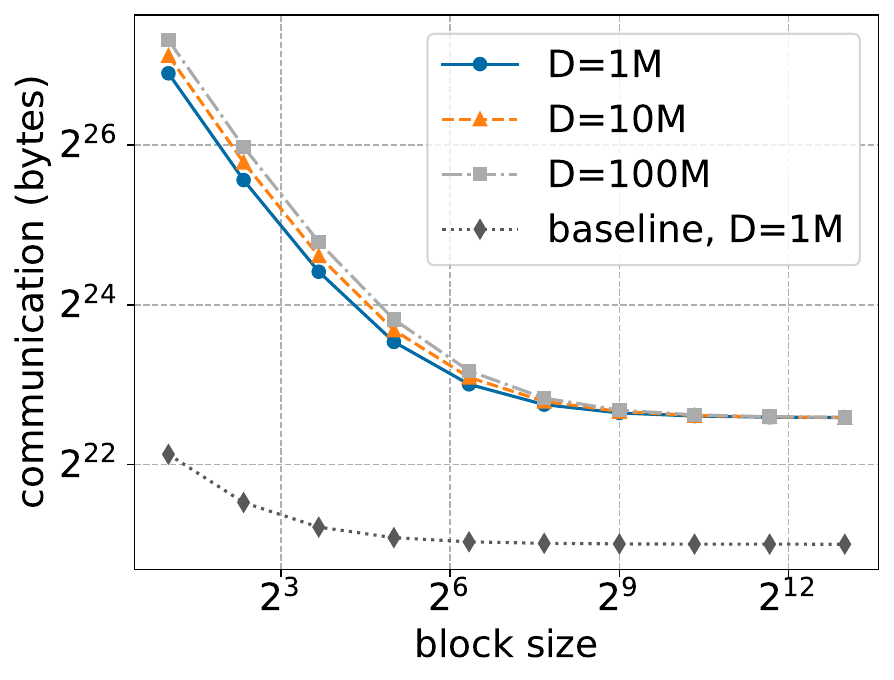}
    \caption{}
    \label{fig:comm_cost}
    \end{subfigure}
    \begin{subfigure}{0.30\textwidth}
        \centering
        \includegraphics[width=\textwidth]{figures/truncation_error_1M.pdf}
        \caption{}
        \label{fig:truncation-error}
    \end{subfigure}
        \begin{subfigure}{0.36\textwidth}
        \centering
        \includegraphics[width=\textwidth]{figures/comparison_eps_1_d_8388608.pdf}
        \caption{}
        \label{fig:error_v_comm}
    \end{subfigure}
   
    \caption{\small (Left) Communication vs. Block size for our $k$-sparse DPF construction, compared to a single-trusted-server baseline. (Middle)
    The trade-off between the truncation error $\|Trunc_c(Gv) - Gv\|_2$ and constant $c$, where $v \in \Re^\fulldim$ is arbitrary, and $G$ is a random rotation matrix. The plots show the error for various block sizes, for  $\fulldim = 2^{20}$.  (Right) The trade-off between the standard deviation of the error and per-client communication $C = k B$, when computing the sum of $n=10^5$ vectors with dimension $\fulldim = 2^{23}$, with $(1.0, 10^{-6})$-DP. The blue 'x' shows the baseline approach of sending the whole vector.}
    \label{fig:blocking_helps}
\end{figure}

\paragraph{\bf Blocking is compatible with Privacy:}

We next evaluate the privacy-utility trade-off of our approach, compared to not using any sampling. For our baseline Gaussian mechanism, we use the Analytic Gaussian mechanism analysis from~\cite{Wang:2021}. For these numerical results, we assume that $L$ is fixed to $\sqrt{B/D}$. For each of the approaches, we compute the total variance of the error in the sum, which includes the privacy error that results from the numerical privacy analysis, and the sampling error as bounded by~\cref{thm:utility}. For the communication cost, we simply plot $kB$ as it is a good proxy for the actual communication cost for a large range of parameters. Based on the evaluation above, we consider values of $B$ that are in the range $(2^{10}, 2^{14})$.

\cref{fig:error_v_comm} shows the trade-off between communication cost and the standard deviation of the error for our algorithm, using partitioned subsampling, as well as the Gaussian baseline. As is clear from the plots, our approach allows us to significantly reduce the communication costs, at the price of a minor increase in the error. This holds for a range of $\fulldim$ from 1M to 8M, and for a range of $\eps$ values. (additional plots in the Appendix).

As we decrease the communication $kB$, the error in \cref{fig:error_v_comm} essentially stays constant until a point, and then rapidly increases. This is largely due to the fact that for large block sizes, the norm of each block is larger so that the required lower bound on $\sigma$ to ensure privacy amplification by subsampling is larger. While the plots are derived from the numerical analysis which is tighter, intuition for this can be derived from the condition $k > \sqrt{\Delta\log 1/\delta}/\sigma$, or equivalently $\sigma > \sqrt{\Delta\log 1/\delta}/k$ in \cref{thm:asymptotic--partitioned-privacy}. Thus for the case of small $kB$, one would prefer smaller block sizes. Our plots show that block sizes in the range $[2^{10},2^{13}]$ provide low error across a range of parameters. Recall that in \cref{fig:comm_cost}, we saw that block sizes above $2^{7}$ are sufficient to get most of the communication benefits.

Next we turn to experiments simulating the use of PREAMBLE in private model training. In \cref{fig:sigma_dpsgd}, we plot the overhead in the per-batch noise standard deviation vs. the block size, if we were to analyse DPSGD~\cite{CCS:ACGMMT16} with the Gaussian mechanism for each batch replaced by partitioned subsampling. Due to the more complex accounting that uses general Renyi subsampling, the increase here is larger.
Finally, we show some end-to-end experiments for private learning (\cref{fig:mnist} and \cref{fig:cifar}). We report results for using DP-SGD (with momentum) on MNIST~\cite{lecun2010mnist} and CIFAR-10~\cite{Krizhevsky09learningmultiple}. For MNIST, we use the model from~\cite{AsiFNNT23} which has $69050$ trainable parameters. For CIFAR, we train a simple two-layer neural network with $66954$ parameters on CLIP~\cite{clip} embeddings. 
We give additional details of the setup, including all hyperparameters in \cref{app:experiment_details}. Our results show that PREAMBLE allows for significant reduction in communication while incurring a small increase in accuracy. Indeed, for both datasets, PREAMBLE with the chosen parameters allows to communicate around $2\cdot 10^4$ parameters, compared to roughly $6 \cdot 10^4$ parameters using the Gaussian mechanism. Additional experiments are deferred to the supplement.

\begin{figure}
    \centering
    \begin{subfigure}{0.32\textwidth}
        \centering
        \includegraphics[width=\textwidth]{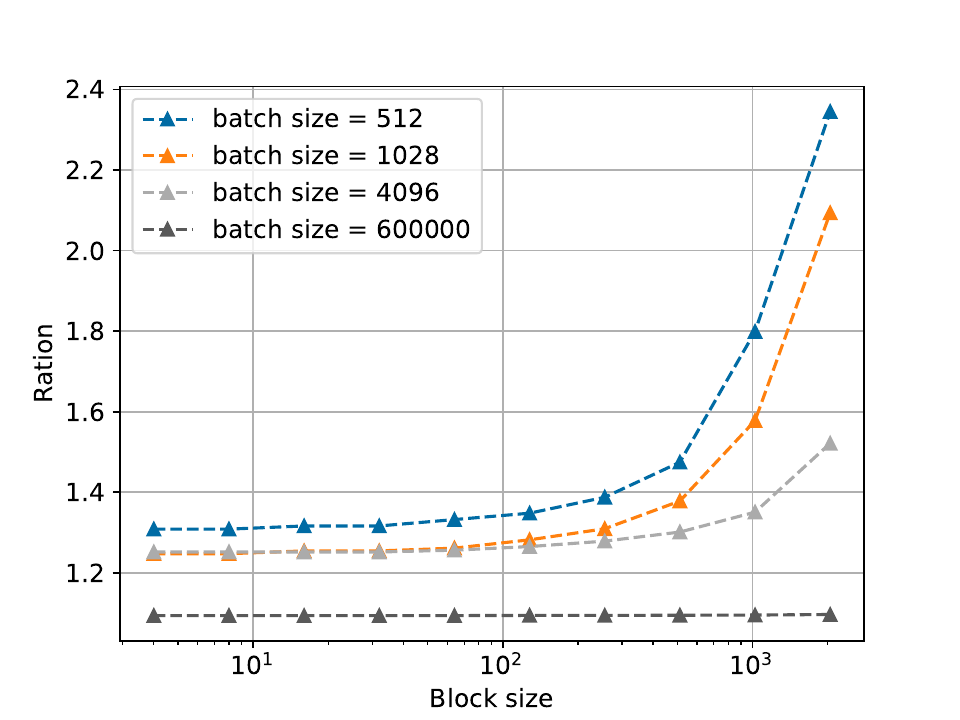}
        \caption{}
        \label{fig:sigma_dpsgd}
    \end{subfigure}
    \begin{subfigure}{0.32\textwidth}
        \centering
        \includegraphics[width=\textwidth]{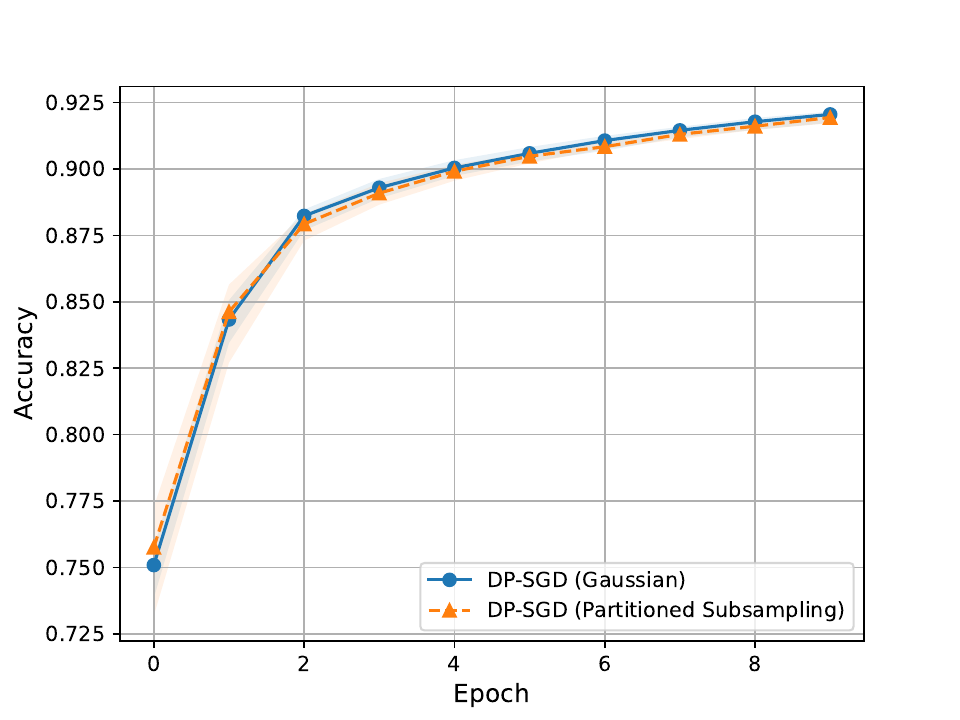}
        \caption{}
        \label{fig:mnist}
    \end{subfigure}
    \begin{subfigure}{0.32\textwidth}
        \centering
        \includegraphics[width=\textwidth]{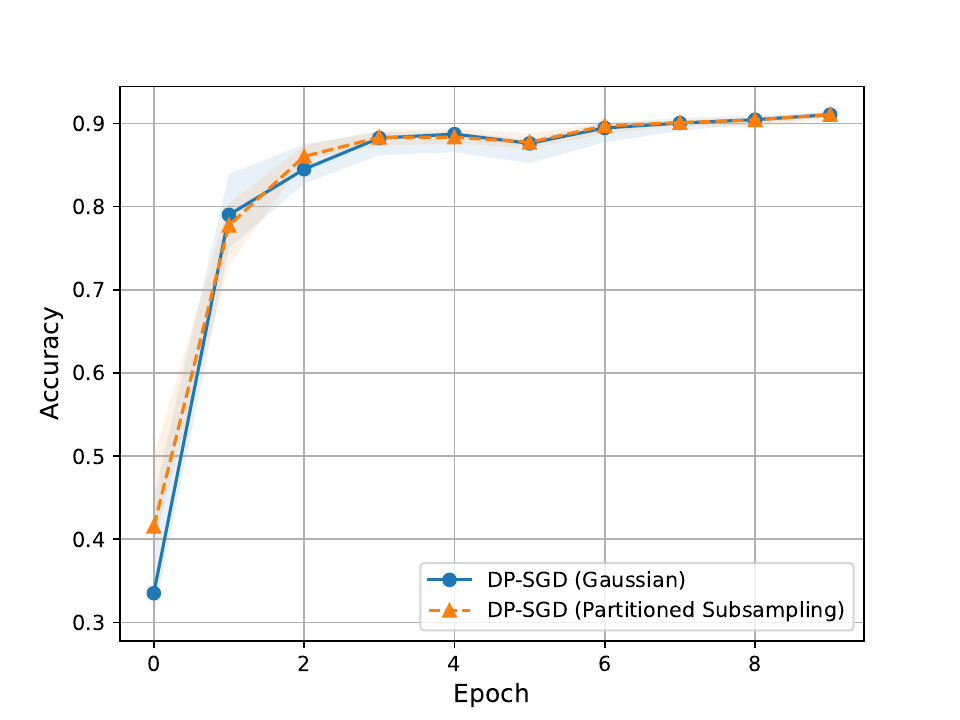}
        \caption{}
        \label{fig:cifar}
    \end{subfigure}
    \caption{\small
    Comparison between PREAMBLE and the Gaussian mechanism for private model training. (Left) Ratio of the per-batch noise standard deviation of PREAMBLE over the noise of the Gaussian mechanism. (Middle) Accuracy for PREAMBLE and Gaussian mechanism on MNIST with 90\% confidence intervals. (Right) Accuracy for PREAMBLE and Gaussian mechanism on CIFAR10 with 90\% confidence intervals.  }
    \label{fig:plot-privacy-tradeoff}
\end{figure}

\section{Conclusions}
\label{sec:conclusions}
In this work, we have described PREAMBLE, an efficient algorithm for communicating high-dimensional Euclidean vectors in the Prio model. Our construction reduces this problem to aggregating $k$-block-sparse vectors, using random sampling, and privacy amplification-by-sampling type analyses to allow private aggregation with a small overhead in accuracy. We showed how to efficiently communicate such vectors, and construct zero-knowledge proofs to validate a bound on the Euclidean norm and $k$-block-sparsity of these vectors. Our algorithms require client communication proportional to the sparsity $kB$ of these vectors, and our client computation also scales only with $kB$ for parameters of interest. Our construction allows the servers to reconstruct each contribution using $O(\fulldim)$ field operations and PRG evaluations in counter mode. 

We leave open some natural research directions. %
Our numerical privacy analyses are close to tight but still have gaps. We conjecture that the $k$-out-of-$\Delta$ sampling approach should admit better privacy analyses than the partitioning-based approach, and it should be no worse than Poisson sampling.
Our approach based on Cuckoo hashing with two hash functions incurs a constant factor communication overhead, and has a $O(\frac 1 k)$ failure probability. While the overhead can be reduced using more hash functions, the failure probability %
remains $k^{-c}$ for a small constant $c$~\cite{KirschMW08}. While for our application to approximate aggregation, this has little impact, it would be natural to design a version of our scheme that has a negligible failure probability without increasing the server compute cost.

\bibliographystyle{alpha}
\bibliography{extrarefs,abbrev3,crypto,refs}

\newcommand{\etalchar}[1]{$^{#1}$}
\begin{thebibliography}{BBCG{\etalchar{+}}21}

\bibitem[AC06]{AilonC06}
Nir Ailon and Bernard Chazelle.
\newblock Approximate nearest neighbors and the fast johnson-lindenstrauss
  transform.
\newblock In {\em Proceedings of the Thirty-Eighth Annual ACM Symposium on
  Theory of Computing}, STOC '06, page 557–563, New York, NY, USA, 2006.
  Association for Computing Machinery.

\bibitem[ACG{\etalchar{+}}16]{CCS:ACGMMT16}
Mart{\'i}n Abadi, Andy Chu, Ian~J. Goodfellow, H.~Brendan McMahan, Ilya
  Mironov, Kunal Talwar, and Li~Zhang.
\newblock Deep learning with differential privacy.
\newblock In Edgar~R. Weippl, Stefan Katzenbeisser, Christopher Kruegel,
  Andrew~C. Myers, and Shai Halevi, editors, {\em ACM CCS 2016}, pages
  308--318. {ACM} Press, October 2016.

\bibitem[ACLS18]{SP:ACLS18}
Sebastian Angel, Hao Chen, Kim Laine, and Srinath T.~V. Setty.
\newblock {PIR} with compressed queries and amortized query processing.
\newblock In {\em 2018 {IEEE} Symposium on Security and Privacy}, pages
  962--979. {IEEE} Computer Society Press, May 2018.

\bibitem[AFN{\etalchar{+}}23]{AsiFNNT23}
Hilal Asi, Vitaly Feldman, Jelani Nelson, Huy Nguyen, and Kunal Talwar.
\newblock Fast optimal locally private mean estimation via random projections.
\newblock In A.~Oh, T.~Naumann, A.~Globerson, K.~Saenko, M.~Hardt, and
  S.~Levine, editors, {\em Advances in Neural Information Processing Systems},
  volume~36, pages 16271--16282. Curran Associates, Inc., 2023.

\bibitem[AFN{\etalchar{+}}24]{AsiFNNTZ24}
Hilal Asi, Vitaly Feldman, Jelani Nelson, Huy~L. Nguyen, Kunal Talwar, and
  Samson Zhou.
\newblock Private vector mean estimation in the shuffle model: Optimal rates
  require many messages, 2024.

\bibitem[AFT22]{AsiFT22}
Hilal Asi, Vitaly Feldman, and Kunal Talwar.
\newblock Optimal algorithms for mean estimation under local differential
  privacy.
\newblock In {\em International Conference on Machine Learning, {ICML}, {USA}},
  pages 1046--1056, 2022.

\bibitem[AG21]{ENPA:2021}
Apple and Google.
\newblock Exposure notification privacy-preserving analytics ({ENPA}) white
  paper.
\newblock
  \url{https://covid19-static.cdn-apple.com/applications/covid19/current/static/contact-tracing/pdf/ENPA_White_Paper.pdf},
  2021.

\bibitem[AGJ{\etalchar{+}}22]{AddankiGJOP22}
Surya Addanki, Kevin Garbe, Eli Jaffe, Rafail Ostrovsky, and Antigoni
  Polychroniadou.
\newblock Prio+: Privacy preserving aggregate statistics via boolean shares.
\newblock In Clemente Galdi and Stanislaw Jarecki, editors, {\em Security and
  Cryptography for Networks - 13th International Conference, {SCN} 2022,
  Amalfi, Italy, September 12-14, 2022, Proceedings}, volume 13409 of {\em
  Lecture Notes in Computer Science}, pages 516--539. Springer, 2022.

\bibitem[APF{\etalchar{+}}23]{AzamPFTSL23}
Sheikh~Shams Azam, Martin Pelikan, Vitaly Feldman, Kunal Talwar, Jan Silovsky,
  and Tatiana Likhomanenko.
\newblock Federated learning for speech recognition: Revisiting current trends
  towards large-scale {ASR}.
\newblock In {\em International Workshop on Federated Learning in the Age of
  Foundation Models in Conjunction with NeurIPS 2023}, 2023.

\bibitem[AS19]{AcharyaS19}
Jayadev Acharya and Ziteng Sun.
\newblock Communication complexity in locally private distribution estimation
  and heavy hitters.
\newblock In {\em Proceedings of the 36th International Conference on Machine
  Learning, {ICML}}, pages 51--60, 2019.

\bibitem[BBC{\etalchar{+}}19]{BBCGI19}
Dan Boneh, Elette Boyle, Henry Corrigan{-}Gibbs, Niv Gilboa, and Yuval Ishai.
\newblock Zero-knowledge proofs on secret-shared data via fully linear pcps.
\newblock In Alexandra Boldyreva and Daniele Micciancio, editors, {\em Advances
  in Cryptology - {CRYPTO} 2019 - 39th Annual International Cryptology
  Conference, Santa Barbara, CA, USA, August 18-22, 2019, Proceedings, Part
  {III}}, volume 11694 of {\em Lecture Notes in Computer Science}, pages
  67--97. Springer, 2019.

\bibitem[BBC{\etalchar{+}}23]{C:BBCGI23}
Dan Boneh, Elette Boyle, Henry {Corrigan-Gibbs}, Niv Gilboa, and Yuval Ishai.
\newblock Arithmetic sketching.
\newblock In Helena Handschuh and Anna Lysyanskaya, editors, {\em CRYPTO~2023,
  Part~I}, volume 14081 of {\em {LNCS}}, pages 171--202. Springer, Cham, August
  2023.

\bibitem[BBCG{\etalchar{+}}21]{BonehBCGI21}
Dan Boneh, Elette Boyle, Henry Corrigan-Gibbs, Niv Gilboa, and Yuval Ishai.
\newblock Lightweight techniques for private heavy hitters.
\newblock In {\em 2021 IEEE Symposium on Security and Privacy (S \& P)}, pages
  762--776, 2021.

\bibitem[BBG20]{Balle:2018}
Borja Balle, Gilles Barthe, and Marco Gaboardi.
\newblock Privacy profiles and amplification by subsampling.
\newblock {\em Journal of Privacy and Confidentiality}, 10(1), 2020.

\bibitem[BBGN19]{BalleBGN19}
Borja Balle, James Bell, Adri{\`{a}} Gasc{\'{o}}n, and Kobbi Nissim.
\newblock The privacy blanket of the shuffle model.
\newblock In {\em Advances in Cryptology - {CRYPTO} 2019 - 39th Annual
  International Cryptology Conference, Proceedings, Part {II}}, pages 638--667,
  2019.

\bibitem[BBGN20]{BalleBGN20}
Borja Balle, James Bell, Adri{\`{a}} Gasc{\'{o}}n, and Kobbi Nissim.
\newblock Private summation in the multi-message shuffle model.
\newblock In {\em {CCS} '20: 2020 {ACM} {SIGSAC} Conference on Computer and
  Communications Security}, pages 657--676, 2020.

\bibitem[BCGI18]{CCS:BCGI18}
Elette Boyle, Geoffroy Couteau, Niv Gilboa, and Yuval Ishai.
\newblock Compressing vector {OLE}.
\newblock In David Lie, Mohammad Mannan, Michael Backes, and XiaoFeng Wang,
  editors, {\em ACM CCS 2018}, pages 896--912. {ACM} Press, October 2018.

\bibitem[BDF{\etalchar{+}}18]{BhowmickDFKR18}
Abhishek Bhowmick, John~C. Duchi, Julien Freudiger, Gaurav Kapoor, and Ryan
  Rogers.
\newblock Protection against reconstruction and its applications in private
  federated learning.
\newblock {\em CoRR}, abs/1812.00984, 2018.

\bibitem[BGH{\etalchar{+}}25]{BoyleGHIY25}
Elette Boyle, Niv Gilboa, Matan Hamilis, Yuval Ishai, and Yaxin Tu.
\newblock { Improved Constructions for Distributed Multi-Point Functions }.
\newblock In {\em 2025 IEEE Symposium on Security and Privacy (SP)}, pages
  2414--2432, Los Alamitos, CA, USA, May 2025. IEEE Computer Society.

\bibitem[BGI15]{EC:BoyGilIsh15}
Elette Boyle, Niv Gilboa, and Yuval Ishai.
\newblock Function secret sharing.
\newblock In Elisabeth Oswald and Marc Fischlin, editors, {\em EUROCRYPT~2015,
  Part~II}, volume 9057 of {\em {LNCS}}, pages 337--367. Springer, Berlin,
  Heidelberg, April 2015.

\bibitem[BGI16]{CCS:BoyGilIsh16}
Elette Boyle, Niv Gilboa, and Yuval Ishai.
\newblock Function secret sharing: Improvements and extensions.
\newblock In Edgar~R. Weippl, Stefan Katzenbeisser, Christopher Kruegel,
  Andrew~C. Myers, and Shai Halevi, editors, {\em ACM CCS 2016}, pages
  1292--1303. {ACM} Press, October 2016.

\bibitem[BNO08]{BeimelNO08}
Amos Beimel, Kobbi Nissim, and Eran Omri.
\newblock Distributed private data analysis: Simultaneously solving how and
  what.
\newblock In David Wagner, editor, {\em Advances in Cryptology -- CRYPTO 2008},
  pages 451--468, Berlin, Heidelberg, 2008. Springer Berlin Heidelberg.

\bibitem[BS15]{BassilyS15}
Raef Bassily and Adam~D. Smith.
\newblock Local, private, efficient protocols for succinct histograms.
\newblock In {\em Proceedings of the Forty-Seventh Annual {ACM} on Symposium on
  Theory of Computing, {STOC}}, pages 127--135, 2015.

\bibitem[BS16]{Bun:2016}
Mark Bun and Thomas Steinke.
\newblock Concentrated differential privacy: Simplifications, extensions, and
  lower bounds.
\newblock In Martin Hirt and Adam Smith, editors, {\em Theory of Cryptography},
  pages 635--658, Berlin, Heidelberg, 2016. Springer Berlin Heidelberg.

\bibitem[BW18]{Balle2018ImprovingTG}
Borja Balle and Yu-Xiang Wang.
\newblock Improving the gaussian mechanism for differential privacy: Analytical
  calibration and optimal denoising.
\newblock In {\em International Conference on Machine Learning}, 2018.

\bibitem[CB17]{Corrigan-GibbsB17}
Henry Corrigan{-}Gibbs and Dan Boneh.
\newblock Prio: Private, robust, and scalable computation of aggregate
  statistics.
\newblock In Aditya Akella and Jon Howell, editors, {\em 14th {USENIX}
  Symposium on Networked Systems Design and Implementation, {NSDI} 2017,
  Boston, MA, USA, March 27-29, 2017}, pages 259--282. {USENIX} Association,
  2017.

\bibitem[CCT{\etalchar{+}}24]{ChauhanCTTN24}
Geeticka Chauhan, Steve Chien, Om~Thakkar, Abhradeep Thakurta, and Arun
  Narayanan.
\newblock Training large asr encoders with differential privacy.
\newblock In {\em 2024 IEEE Spoken Language Technology Workshop (SLT)}, pages
  102--109. IEEE, 2024.

\bibitem[CGH{\etalchar{+}}25]{ChuaGHKKLMSZ25}
Lynn Chua, Badih Ghazi, Charlie Harrison, Pritish Kamath, Ravi Kumar,
  Ethan~Jacob Leeman, Pasin Manurangsi, Amer Sinha, and Chiyuan Zhang.
\newblock Balls-and-bins sampling for {DP}-{SGD}.
\newblock In {\em The 28th International Conference on Artificial Intelligence
  and Statistics}, 2025.

\bibitem[CIK{\etalchar{+}}24]{ChenIKNOX24}
Wei-Ning Chen, Berivan Isik, Peter Kairouz, Albert No, Sewoong Oh, and Zheng
  Xu.
\newblock Improved communication-privacy trade-offs in l2 mean estimation under
  streaming differential privacy.
\newblock In {\em Proceedings of the 41st International Conference on Machine
  Learning}, ICML'24. JMLR.org, 2024.

\bibitem[CJMP22]{CheuJMP22}
Albert Cheu, Matthew Joseph, Jieming Mao, and Binghui Peng.
\newblock Shuffle private stochastic convex optimization.
\newblock In {\em The Tenth International Conference on Learning
  Representations, {ICLR}}, 2022.

\bibitem[CK{\"O}20]{chen2020breaking}
Wei-Ning Chen, Peter Kairouz, and Ayfer {\"O}zg{\"u}r.
\newblock Breaking the communication-privacy-accuracy trilemma.
\newblock {\em arXiv preprint arXiv:2007.11707}, 2020.

\bibitem[CLR17]{CCS:CheLaiRin17}
Hao Chen, Kim Laine, and Peter Rindal.
\newblock Fast private set intersection from homomorphic encryption.
\newblock In Bhavani~M. Thuraisingham, David Evans, Tal Malkin, and Dongyan Xu,
  editors, {\em ACM CCS 2017}, pages 1243--1255. {ACM} Press,
  October~/~November 2017.

\bibitem[CMS11]{ChaudhuriMS11}
Kamalika Chaudhuri, Claire Monteleoni, and Anand~D. Sarwate.
\newblock Differentially private empirical risk minimization.
\newblock {\em Journal of Machine Learning Research}, 12(29):1069--1109, 2011.

\bibitem[CSOK23]{chen2023privacy}
Wei-Ning Chen, Dan Song, Ayfer Ozgur, and Peter Kairouz.
\newblock Privacy amplification via compression: Achieving the optimal
  privacy-accuracy-communication trade-off in distributed mean estimation.
\newblock In {\em Thirty-seventh Conference on Neural Information Processing
  Systems}, 2023.

\bibitem[CSS12]{ChanSS12}
T.~H.~Hubert Chan, Elaine Shi, and Dawn Song.
\newblock Privacy-preserving stream aggregation with fault tolerance.
\newblock In Angelos~D. Keromytis, editor, {\em Financial Cryptography and Data
  Security}, pages 200--214, Berlin, Heidelberg, 2012. Springer Berlin
  Heidelberg.

\bibitem[dCP22]{EC:dCaPol22}
Leo de~Castro and Antigoni Polychroniadou.
\newblock Lightweight, maliciously secure verifiable function secret sharing.
\newblock In Orr Dunkelman and Stefan Dziembowski, editors, {\em
  EUROCRYPT~2022, Part~I}, volume 13275 of {\em {LNCS}}, pages 150--179.
  Springer, Cham, May~/~June 2022.

\bibitem[DJW18]{DuchiJW18}
John~C. Duchi, Michael~I. Jordan, and Martin~J. Wainwright.
\newblock Minimax optimal procedures for locally private estimation.
\newblock {\em Journal of the American Statistical Association},
  113(521):182--201, 2018.

\bibitem[DK12]{DK12}
Michael Drmota and Reinhard Kutzelnigg.
\newblock A precise analysis of cuckoo hashing.
\newblock {\em ACM Trans. Algorithms}, 8(2), April 2012.

\bibitem[DKM{\etalchar{+}}06]{ODOpaper}
Cynthia Dwork, Krishnaram Kenthapadi, Frank McSherry, Ilya Mironov, and Moni
  Naor.
\newblock Our data, ourselves: Privacy via distributed noise generation.
\newblock In {\em Advances in Cryptology (EUROCRYPT 2006)}, volume 4004 of {\em
  Lecture Notes in Computer Science}, pages 486--503. Springer Verlag, 2006.

\bibitem[DMNS06]{Dwork:2006}
Cynthia Dwork, Frank McSherry, Kobbi Nissim, and Adam Smith.
\newblock Calibrating noise to sensitivity in private data analysis.
\newblock In Shai Halevi and Tal Rabin, editors, {\em Theory of Cryptography},
  pages 265--284, Berlin, Heidelberg, 2006. Springer Berlin Heidelberg.

\bibitem[DR14]{DworkR14}
Cynthia Dwork and Aaron Roth.
\newblock The algorithmic foundations of differential privacy.
\newblock {\em Found. Trends Theor. Comput. Sci.}, 9(3-4):211--407, 2014.

\bibitem[DR16]{DworkR16}
Cynthia Dwork and Guy~N. Rothblum.
\newblock Concentrated differential privacy.
\newblock {\em CoRR}, abs/1603.01887, 2016.

\bibitem[DR19]{DuchiR19}
John~C. Duchi and Ryan Rogers.
\newblock Lower bounds for locally private estimation via communication
  complexity.
\newblock In {\em Conference on Learning Theory, {COLT}}, pages 1161--1191,
  2019.

\bibitem[DRRT18]{PoPETS:DRRT18}
Daniel Demmler, Peter Rindal, Mike Rosulek, and Ni~Trieu.
\newblock {PIR}-{PSI}: Scaling private contact discovery.
\newblock {\em {PoPETs}}, 2018(4):159--178, October 2018.

\bibitem[DRS22]{DongRS22}
Jinshuo Dong, Aaron Roth, and Weijie~J. Su.
\newblock Gaussian differential privacy.
\newblock {\em Journal of the Royal Statistical Society Series B: Statistical
  Methodology}, 84(1):3--37, 02 2022.

\bibitem[EPK14]{ErlingssonPK14}
\'{U}lfar Erlingsson, Vasyl Pihur, and Aleksandra Korolova.
\newblock Rappor: Randomized aggregatable privacy-preserving ordinal response.
\newblock In {\em Proceedings of the 2014 ACM SIGSAC Conference on Computer and
  Communications Security}, CCS '14, page 1054–1067, New York, NY, USA, 2014.
  Association for Computing Machinery.

\bibitem[FHNP16]{JC:FHNP16}
Michael~J. Freedman, Carmit Hazay, Kobbi Nissim, and Benny Pinkas.
\newblock Efficient set intersection with simulation-based security.
\newblock {\em Journal of Cryptology}, 29(1):115--155, January 2016.

\bibitem[FM12]{FriezeM12}
Alan Frieze and Páll Melsted.
\newblock Maximum matchings in random bipartite graphs and the space
  utilization of cuckoo hash tables.
\newblock {\em Random Structures \& Algorithms}, 41(3):334--364, 2012.

\bibitem[FNNT22]{FeldmanNNT22}
Vitaly Feldman, Jelani Nelson, Huy Nguyen, and Kunal Talwar.
\newblock Private frequency estimation via projective geometry.
\newblock In Kamalika Chaudhuri, Stefanie Jegelka, Le~Song, Csaba Szepesvari,
  Gang Niu, and Sivan Sabato, editors, {\em Proceedings of the 39th
  International Conference on Machine Learning}, volume 162 of {\em Proceedings
  of Machine Learning Research}, pages 6418--6433. PMLR, 2022.

\bibitem[FP12]{FoutoulakisP12}
Nikolaos Fountoulakis and Konstantinos Panagiotou.
\newblock Sharp load thresholds for cuckoo hashing.
\newblock {\em Random Structures \& Algorithms}, 41(3):306--333, 2012.

\bibitem[FPE16]{FantiPE16}
Giulia Fanti, Vasyl Pihur, and {\'{U}}lfar Erlingsson.
\newblock Building a {RAPPOR} with the unknown: Privacy-preserving learning of
  associations and data dictionaries.
\newblock {\em Proc. Priv. Enhancing Technol.}, 2016(3):41--61, 2016.

\bibitem[FS25]{FeldmanS25}
Vitaly Feldman and Moshe Shenfeld.
\newblock Privacy amplification by random allocation, 2025.

\bibitem[FT21]{FeldmanTalwar21}
Vitaly Feldman and Kunal Talwar.
\newblock Lossless compression of efficient private local randomizers.
\newblock In Marina Meila and Tong Zhang, editors, {\em {ICML}}, volume 139 of
  {\em Proceedings of Machine Learning Research}, pages 3208--3219. {PMLR},
  2021.

\bibitem[GI14]{GilboaI14}
Niv Gilboa and Yuval Ishai.
\newblock Distributed point functions and their applications.
\newblock In Phong~Q. Nguyen and Elisabeth Oswald, editors, {\em Advances in
  Cryptology -- EUROCRYPT 2014}, pages 640--658, Berlin, Heidelberg, 2014.
  Springer Berlin Heidelberg.

\bibitem[GKKM22]{GhaziKKM22}
Badih Ghazi, Pritish Kamath, Ravi Kumar, and Pasin Manurangsi.
\newblock Faster privacy accounting via evolving discretization.
\newblock In Kamalika Chaudhuri, Stefanie Jegelka, Le~Song, Csaba Szepesvari,
  Gang Niu, and Sivan Sabato, editors, {\em Proceedings of the 39th
  International Conference on Machine Learning}, volume 162 of {\em Proceedings
  of Machine Learning Research}, pages 7470--7483. PMLR, 17--23 Jul 2022.

\bibitem[GKM{\etalchar{+}}21]{GhaziKMPS21}
Badih Ghazi, Ravi Kumar, Pasin Manurangsi, Rasmus Pagh, and Amer Sinha.
\newblock Differentially private aggregation in the shuffle model: Almost
  central accuracy in almost a single message.
\newblock In {\em Proceedings of the 38th International Conference on Machine
  Learning, {ICML}}, pages 3692--3701, 2021.

\bibitem[GLW21]{GopiLW21}
Sivakanth Gopi, Yin~Tat Lee, and Lukas Wutschitz.
\newblock Numerical composition of differential privacy.
\newblock In M.~Ranzato, A.~Beygelzimer, Y.~Dauphin, P.S. Liang, and J.~Wortman
  Vaughan, editors, {\em Advances in Neural Information Processing Systems},
  volume~34, pages 11631--11642. Curran Associates, Inc., 2021.

\bibitem[GMPV20]{GhaziMPV20}
Badih Ghazi, Pasin Manurangsi, Rasmus Pagh, and Ameya Velingker.
\newblock Private aggregation from fewer anonymous messages.
\newblock In {\em Advances in Cryptology - {EUROCRYPT} 2020 - 39th Annual
  International Conference on the Theory and Applications of Cryptographic
  Techniques, Proceedings, Part {II}}, pages 798--827, 2020.

\bibitem[HMR18]{firefox}
Robert Helmer, Anthony Miyaguchi, and Eric Rescorla.
\newblock Testing privacy-preserving telemetry with prio.
\newblock
  \url{https://hacks.mozilla.org/2018/10/testing-privacy-preserving-telemetry-with-prio/},
  2018.

\bibitem[HSW{\etalchar{+}}22]{Hu22}
Edward~J Hu, Yelong Shen, Phillip Wallis, Zeyuan Allen-Zhu, Yuanzhi Li, Shean
  Wang, Lu~Wang, and Weizhu Chen.
\newblock Lo{RA}: Low-rank adaptation of large language models.
\newblock In {\em International Conference on Learning Representations}, 2022.

\bibitem[JL84]{JohnsonL84}
William~B. Johnson and Joram Lindenstrauss.
\newblock Extensions of lipschitz mappings into hilbert space.
\newblock {\em Contemporary mathematics}, 26:189--206, 1984.

\bibitem[KH21]{KoskelaH21}
Antti Koskela and Antti Honkela.
\newblock Computing differential privacy guarantees for heterogeneous
  compositions using fft, 2021.

\bibitem[KKEPR24]{EPRINT:KKEP24}
Erki Külaots, Toomas Krips, Hendrik Eerikson, and Pille Pullonen-Raudvere.
\newblock Slamp-fss: Two-party multi-point function secret sharing from simple
  linear algebra.
\newblock Cryptology ePrint Archive, Report 2024/1394, 2024.
\newblock \url{https://eprint.iacr.org/2024/1394.pdf}.

\bibitem[KLN{\etalchar{+}}08]{Kasiviswanathan:2008}
S.~P. {Kasiviswanathan}, H.~K. {Lee}, K.~{Nissim}, S.~{Raskhodnikova}, and
  A.~{Smith}.
\newblock What can we learn privately?
\newblock In {\em 2008 49th Annual IEEE Symposium on Foundations of Computer
  Science}, pages 531--540, 2008.

\bibitem[KMW08]{KirschMW08}
Adam Kirsch, Michael Mitzenmacher, and Udi Wieder.
\newblock More robust hashing: Cuckoo hashing with a stash.
\newblock In Dan Halperin and Kurt Mehlhorn, editors, {\em Algorithms - ESA
  2008}, pages 611--622, Berlin, Heidelberg, 2008. Springer Berlin Heidelberg.

\bibitem[KOV15]{KairouzOVa}
Peter Kairouz, Sewoong Oh, and Pramod Viswanath.
\newblock Secure multi-party differential privacy.
\newblock In C.~Cortes, N.~D. Lawrence, D.~D. Lee, M.~Sugiyama, and R.~Garnett,
  editors, {\em Advances in Neural Information Processing Systems 28}, pages
  2008--2016. Curran Associates, Inc., 2015.

\bibitem[Kri09]{Krizhevsky09learningmultiple}
Alex Krizhevsky.
\newblock Learning multiple layers of features from tiny images.
\newblock Technical report, 2009.

\bibitem[LCB10]{lecun2010mnist}
Yann LeCun, Corinna Cortes, and CJ~Burges.
\newblock Mnist handwritten digit database.
\newblock {\em ATT Labs [Online]. Available: http://yann.lecun.com/exdb/mnist},
  2, 2010.

\bibitem[LV10]{Lyubarskii:2010}
Y.~Lyubarskii and R.~Vershynin.
\newblock Uncertainty principles and vector quantization.
\newblock {\em Information Theory, IEEE Transactions on}, 56(7):3491--3501,
  2010.

\bibitem[MM18]{MeiserM18}
Sebastian Meiser and Esfandiar Mohammadi.
\newblock Tight on budget? tight bounds for r-fold approximate differential
  privacy.
\newblock In {\em Proceedings of the 2018 ACM SIGSAC Conference on Computer and
  Communications Security}, CCS '18, page 247–264, New York, NY, USA, 2018.
  Association for Computing Machinery.

\bibitem[MMR{\etalchar{+}}17]{McMahanMRHA17}
Brendan McMahan, Eider Moore, Daniel Ramage, Seth Hampson, and
  Blaise~Ag{\"{u}}era y~Arcas.
\newblock Communication-efficient learning of deep networks from decentralized
  data.
\newblock In {\em Proceedings of the 20th International Conference on
  Artificial Intelligence and Statistics, {AISTATS}}, pages 1273--1282, 2017.

\bibitem[MTZ19]{Mironov2019RnyiDP}
Ilya Mironov, Kunal Talwar, and Li~Zhang.
\newblock R{\'e}nyi differential privacy of the sampled gaussian mechanism.
\newblock {\em ArXiv}, abs/1908.10530, 2019.

\bibitem[NXY{\etalchar{+}}16]{NguyenXYHSS16}
Th{\^{o}}ng~T. Nguy{\^{e}}n, Xiaokui Xiao, Yin Yang, Siu~Cheung Hui, Hyejin
  Shin, and Junbum Shin.
\newblock Collecting and analyzing data from smart device users with local
  differential privacy.
\newblock {\em CoRR}, abs/1606.05053, 2016.

\bibitem[PR01]{PF01}
Rasmus Pagh and Flemming~Friche Rodler.
\newblock Cuckoo hashing.
\newblock In {\em Proceedings of the 9th Annual European Symposium on
  Algorithms}, ESA '01, page 121–133, Berlin, Heidelberg, 2001.
  Springer-Verlag.

\bibitem[PSZ14]{USENIX:PinSchZoh14}
Benny Pinkas, Thomas Schneider, and Michael Zohner.
\newblock Faster private set intersection based on {OT} extension.
\newblock In Kevin Fu and Jaeyeon Jung, editors, {\em USENIX Security 2014},
  pages 797--812. {USENIX} Association, August 2014.

\bibitem[RKH{\etalchar{+}}21]{clip}
Alec Radford, Jong~Wook Kim, Chris Hallacy, Aditya Ramesh, Gabriel Goh,
  Sandhini Agarwal, Girish Sastry, Amanda Askell, Pamela Mishkin, Jack Clark,
  Gretchen Krueger, and Ilya Sutskever.
\newblock Learning transferable visual models from natural language
  supervision.
\newblock In Marina Meila and Tong Zhang, editors, {\em Proceedings of the 38th
  International Conference on Machine Learning}, volume 139 of {\em Proceedings
  of Machine Learning Research}, pages 8748--8763. PMLR, 18--24 Jul 2021.
\newblock Model at \url{https://github.com/openai/CLIP/}.

\bibitem[ROCT24]{RothblumOCT24}
Guy~N. Rothblum, Eran Omri, Junye Chen, and Kunal Talwar.
\newblock {PINE}: Efficient verification of a euclidean norm bound of a
  {Secret-Shared} vector.
\newblock In {\em 33rd USENIX Security Symposium (USENIX Security 24)}, pages
  6975--6992, Philadelphia, PA, August 2024. USENIX Association.

\bibitem[RSWP22]{RatheeSWP22}
Mayank Rathee, Conghao Shen, Sameer Wagh, and Raluca~Ada Popa.
\newblock {ELSA:} secure aggregation for federated learning with malicious
  actors.
\newblock {\em {IACR} Cryptol. ePrint Arch.}, page 1695, 2022.

\bibitem[RU23]{RohrigU23}
Olivia R\"ohrig and Maxim Urschumzew.
\newblock \texttt{dpsa4fl}: Differential privacy for federated machine learning
  with {PRIO}, 2023.

\bibitem[SFZ{\etalchar{+}}14]{sun2014personalized}
Chongjing Sun, Yan Fu, Junlin Zhou, Hui Gao, et~al.
\newblock Personalized privacy-preserving frequent itemset mining using
  randomized response.
\newblock {\em The Scientific World Journal}, 2014, 2014.

\bibitem[SGRR19]{CCS:SGRR19}
Phillipp Schoppmann, Adri{\`a} Gasc{\'o}n, Leonie Reichert, and Mariana
  Raykova.
\newblock Distributed vector-{OLE}: Improved constructions and implementation.
\newblock In Lorenzo Cavallaro, Johannes Kinder, XiaoFeng Wang, and Jonathan
  Katz, editors, {\em ACM CCS 2019}, pages 1055--1072. {ACM} Press, November
  2019.

\bibitem[SMM19]{SommerMM19}
David Sommer, Sebastian Meiser, and Esfandiar Mohammadi.
\newblock Privacy loss classes: The central limit theorem in differential
  privacy.
\newblock {\em Proceedings on Privacy Enhancing Technologies}, 2019:245--269,
  04 2019.

\bibitem[Ste22]{steinke2022composition}
Thomas Steinke.
\newblock Composition of differential privacy \& privacy amplification by
  subsampling, 2022.

\bibitem[SYKM17]{SureshYuKuMc17}
Ananda~Theertha Suresh, Felix~X. Yu, Sanjiv Kumar, and H.~Brendan McMahan.
\newblock Distributed mean estimation with limited communication.
\newblock In {\em Proceedings of the 34th International Conference on Machine
  Learning}, 2017.

\bibitem[Tal22]{Talwar22}
Kunal Talwar.
\newblock Differential secrecy for distributed data and applications to robust
  differentially secure vector summation.
\newblock In L.~Elisa Celis, editor, {\em 3rd Symposium on Foundations of
  Responsible Computing, {FORC} 2022, June 6-8, 2022, Cambridge, MA, {USA}},
  volume 218 of {\em LIPIcs}, pages 7:1--7:16. Schloss Dagstuhl -
  Leibniz-Zentrum f{\"{u}}r Informatik, 2022.

\bibitem[TWM{\etalchar{+}}24]{Talwar24}
Kunal Talwar, Shan Wang, Audra McMillan, Vitaly Feldman, Pansy Bansal, Bailey
  Basile, Aine Cahill, Yi~Sheng Chan, Mike Chatzidakis, Junye Chen, Oliver
  R.~A. Chick, Mona Chitnis, Suman Ganta, Yusuf Goren, Filip Granqvist,
  Kristine Guo, Frederic Jacobs, Omid Javidbakht, Albert Liu, Richard Low, Dan
  Mascenik, Steve Myers, David Park, Wonhee Park, Gianni Parsa, Tommy Pauly,
  Christian Priebe, Rehan Rishi, Guy~N. Rothblum, Congzheng Song, Linmao Song,
  Karl Tarbe, Sebastian Vogt, Shundong Zhou, Vojta Jina, Michael Scaria, and
  Luke Winstrom.
\newblock Samplable anonymous aggregation for private federated data analysis.
\newblock In {\em Proceedings of the 2024 on ACM SIGSAC Conference on Computer
  and Communications Security}, CCS '24, page 2859–2873, New York, NY, USA,
  2024. Association for Computing Machinery.

\bibitem[VBBP{\etalchar{+}}21]{VargaftikBaPoGaBeYaMi21}
Shay Vargaftik, Ran Ben-Basat, Amit Portnoy, Gal Mendelson, Yaniv Ben-Itzhak,
  and Michael Mitzenmacher.
\newblock Drive: One-bit distributed mean estimation.
\newblock In M.~Ranzato, A.~Beygelzimer, Y.~Dauphin, P.S. Liang, and J.~Wortman
  Vaughan, editors, {\em Advances in Neural Information Processing Systems},
  volume~34, pages 362--377. Curran Associates, Inc., 2021.

\bibitem[VBBP{\etalchar{+}}22]{VargaftikBaPoMeItMi22}
Shay Vargaftik, Ran Ben-Basat, Amit Portnoy, Gal Mendelson, Yaniv~Ben Itzhak,
  and Michael Mitzenmacher.
\newblock Eden: Communication-efficient and robust distributed mean estimation
  for federated learning.
\newblock In {\em Proceedings of the 39th International Conference on Machine
  Learning}, 2022.

\bibitem[WBK21]{Wang:2021}
Yu-Xiang Wang, Borja Balle, and Shiva Kasiviswanathan.
\newblock Subsampled rényi differential privacy and analytical moments
  accountant.
\newblock {\em Journal of Privacy and Confidentiality}, 10(2), 2021.

\bibitem[YB18]{YeB18}
Min Ye and Alexander Barg.
\newblock Optimal schemes for discrete distribution estimation under locally
  differential privacy.
\newblock {\em {IEEE} Trans. Inf. Theory}, 64(8):5662--5676, 2018.

\bibitem[ZDW22]{ZhuDW22}
Yuqing Zhu, Jinshuo Dong, and Yu-Xiang Wang.
\newblock Optimal accounting of differential privacy via characteristic
  function.
\newblock In Gustau Camps-Valls, Francisco J.~R. Ruiz, and Isabel Valera,
  editors, {\em Proceedings of The 25th International Conference on Artificial
  Intelligence and Statistics}, volume 151 of {\em Proceedings of Machine
  Learning Research}, pages 4782--4817. PMLR, 28--30 Mar 2022.

\bibitem[ZWC{\etalchar{+}}22]{Zhou0CFS22}
Mingxun Zhou, Tianhao Wang, T.{-}H.~Hubert Chan, Giulia Fanti, and Elaine Shi.
\newblock Locally differentially private sparse vector aggregation.
\newblock In {\em 43rd {IEEE} Symposium on Security and Privacy, {SP}}, pages
  422--439, 2022.

\end{thebibliography}
\appendix
\section{Proof of Theorem \ref{thm:asymptotic--partitioned-privacy}}
\label{sec:privacy-proof}
We will  need the following asymptotic bound on the privacy of Poisson subsampling of the Gaussian noise addition.
\begin{lemma}[\cite{CCS:ACGMMT16}]
	\label{lem:poisson:privacy}
	Let $A_1,\ldots,A_T: (\Re^{B})^n\rightarrow \Re^{B}$ be a sequence of Gaussian noise addition algorithms with sensitivity $s$ and noise scale $\sigma$. Let $P_\eta(A_1,\ldots,A_T)$ be the Poisson subsampling scheme in which each user's data is used in each step with probability $\eta$ randomly and independently (of other users and steps). Then, there exist a constant $c$ such that for every $\delta >0$, if $\sigma/s \geq c \eta \sqrt{T\log(1/\delta)}$ then $P_\eta(A_1,\ldots,A_T)$  satisfies $(\eps,\delta)$-DP for $\eps=O(\eta \frac{s}{\sigma} \sqrt{T\log(1/\delta)}$.
\end{lemma}

\begin{proof}[of \cref{thm:asymptotic--partitioned-privacy}]
We first observe that we can analyze algorithm as an independent composition of $A_1,\ldots,A_k$, with the instance $A_j$ outputting the coordinates of $W$ in the set $S_j = \{(j-1)\Delta/k+1,\ldots, j\Delta/k\}$ for $j\in [k]$. For convenience of notation we will analyze $A_1$ (with the rest being identical). Observe that the output of $A_1$ is $W_1,\ldots,W_{\Delta/k}$. Now, by the definition of the sampling scheme any user's data is summed in exactly one (uniformly chosen) of $W_1,\ldots,W_{\Delta/k}$.
Each of  these algorithms is Gaussian noise addition with sensitivity $L \Delta/k$ and scale $\sigma$. This implies that we can apply results for privacy amplification by allocation for Gaussian noise from \cite{FeldmanS25} to analyze this algorithm. For the analytic results we use an upper bound on the privacy parameters of random allocation in terms of  Poisson subsampling for Gaussian noise. Specifically,  for every $\eps$,  $k$-wise composition of 1 out of  $\Delta/k$ random allocation for Gaussian noise addition algorithms satisfies $(\eps, \delta_P + \Delta \delta_0 + \delta')$-DP, where $(\eps,\delta_P)$ are the privacy parameters of $\Delta$-step Poisson subsampling scheme with rate $\eta = \frac{k}{\Delta (1-\gamma)}$ for $\gamma = (e^{\eps_0}-e^{\eps_0})\sqrt{\frac{k}{2\Delta} \ln \left(\frac{k}{\delta'} \right)}$.  
Here $(\eps_0,\delta_0)$ are the privacy parameters of each Gaussian noise addition. Note that the sensitivity of the aggregate in each block is $L\Delta/k$. Therefore, by setting $\delta_0=\delta/(3\Delta)$ we get $\eps_0 = \frac{L \Delta \sqrt{2\ln(15\Delta/(4\delta))}}{k\sigma}$ (Thm.~\ref{thm:gaussian}). We set $\delta'=\delta/3$ and note that by the first part of our assumption on $\sigma/L$, $\eps_0 \leq 1$ and therefore $e^{\eps_0}-e^{\eps_0} \leq 3 \eps_0$. Now the second part of our assumption of $\sigma/L$ implies that 
\begin{align*}
\gamma &\leq 3 \eps_0 \sqrt{\frac{k}{2\Delta} \ln \left(\frac{3k}{\delta} \right)}  \\&
\leq  \frac{3 L \Delta \sqrt{2\ln(15\Delta/(4\delta))}}{k\sigma} \sqrt{\frac{k}{2\Delta} \ln \left(\frac{3k}{\delta} \right)} \\
&\leq \frac{3 L \sqrt{\Delta} \ln(15\Delta/(4\delta))}{\sqrt{k}\sigma} \leq 1/2.
\end{align*}

 This implies that $\eta \leq \frac{2k}{\Delta}$. Now, by Lemma \ref{lem:poisson:privacy},  the $\Delta$-step Poisson subsampling scheme with subsampling rate $\eta$ is $(\eps,\delta/3)$-DP for $$\eps = O\left(\frac{L\sqrt{\Delta} \sqrt{\log(1/\delta)}}{\sigma  }\right).$$ Here we note that the conditions of the lemma translate to 
 $$\frac{\sigma k}{L \Delta} \geq c \frac{k}{\Delta} \sqrt{\Delta \log (1/\delta)}$$ or equivalently, $\frac{\sigma }{L} \geq c \sqrt{\Delta \log (1/\delta)}$ (which is ensured by the third part of our assumption). Finally, noting that $\Delta \delta_0 + \delta' \leq 2\delta/3$, we get the claimed bound.
\end{proof} 

\section{Ensuring block-wise norm bound}
\newcommand{\clipop}{\mathtt{blkclip}}

\label{app:truncation_analysis}
We now formally show that one can reduce the problem of computing means of $\ell_2$ bounded vectors to our setting of block-wise bounded norm. Without loss of generality, we can assume that each input vector has $\ell_2$ norm 1. Our main reduction is based on the techniques developed by Asi {\em et al.}~\cite{AsiFNNT23} in the context of communication-efficient algorithms for mean estimation with local differential privacy. Their randomizer for mean estimation relies on a randomized dimensionality reduction followed by an optimal differentially private randomizer in lower dimension referred to as {\tt PrivUnit}. {\tt PrivUnit} requires a vector of unit length as an input, whereas the randomized dimensionality reductions they use result in vectors of varying lengths. Asi {\em et al.}~ apply scaling to ensure that the norm condition is satisfied and develop several techniques for the analysis of the error resulting from this  step. We observe that their dimensionality reductions can be used just as (randomized) linear maps (in the same dimension) with each block of $B$ coordinates in the image corresponding to the projection of the original input into $B$ dimensions. Thus we can apply clipping/scaling to ensure that block norms are upper bounded and then analyze the resulting error in essentially the same way as in \cite{AsiFNNT23}.
Our first application of this approach shows that a random rotation with simple block norm clipping to $\sqrt{B/D}$ achieves expected squared error of $1/B$.

\begin{theorem}
\label{thm:truncated-rotation-error}
For a vector $v=v_1,\ldots,v_\Delta\in \Re^D$, where $v_i\in \Re^B$ let $\clipop_B(v)$ denote the vector $u=u_1,\ldots,u_\Delta$, such that for every $i\in [\Delta]$, $u_i= \mbox{clip}_{\sqrt{B/D}}(v_i)$. Let $U\in \Re^{D\times D}$ be a randomly and uniformly chosen rotation matrix. Then 
for every $v\in \Re^D$ such that $\|v\|_2\leq 1$, $$\E_U\left[\left\| U^\top \clipop_B(Uv) - v\right\|_2^2\right]= O\left(\frac{1}{B}\right) .$$
\end{theorem}
Our proof of this result relies on the following lemma from ~\cite{AsiFNNT23}.
\begin{lemma}
\label{lemma:unit-proj}
Let $x$ be a random unit vector on the unit ball of $\Re^D$ and $z$ be the projection of $x$ onto the last $B$ coordinates. We have
\[
    \left|\E[\|z\|_2]-\sqrt{B/D}\right| = O\left(\frac{1}{\sqrt{BD}}\right)
\]
\end{lemma}
\begin{proof}[Proof of Theorem \ref{thm:truncated-rotation-error}]
We first note that the squared error scales quadratically with the norm of $v$ and therefore it is sufficient to prove the theorem for unit norm $v$. For a given $U$, Let $w = Uv$ and $w_1,\ldots,w_\Delta$ be the blocks of size $B$ in $w$. Observe that when $U$ is a randomly chosen rotation matrix, $w$ is a random and uniform unit vector. Naturally, the uniform distribution over random unit vectors is not affected by permuting coordinates and therefore for every $i\in [\Delta]$ we can apply Lemma~\ref{lemma:unit-proj} to get
\[
    \left|\E_U[\|w_i\|_2]-\sqrt{B/D}\right| = O\left(\frac{1}{\sqrt{BD}}\right) .
\]
In addition, by the same symmetry, we have that
\[
    \E_U[\|w_i\|^2] = \frac{B}{D} .
\]
Combining these two results, we have
\begin{align*}
\E_U\left[\left\| U^\top \clipop_B(Uv) - v\right\|_2^2\right] &= \E_U\left[\left\| UU^\top \clipop_B(Uv) - Uv\right\|_2^2\right]\\
    &=\E_U\left[\left\| \clipop_B(w) - w\right\|_2^2\right]\\
    &= \sum_{i\in[\Delta]} \E_U\left[\left\| \mbox{clip}_{\sqrt{B/D}}(w_i) - w_i\right\|_2^2\right]\\
    &\leq \sum_{i\in[\Delta]} \E_U\left[\left(\|w_i\|_2 - \sqrt{B/D}\right)_2^2\right]\\
    &= \sum_{i\in[\Delta]} \E_U\left[\left(2\frac{B}{D} - 2\sqrt{B/D} \|w_i\|_2 \right)\right]\\
    &= 2 \sqrt{B/D} \cdot \sum_{i\in[\Delta]} \E_U\left[\left(\sqrt{B/D}  - \|w_i\|_2 \right)\right]\\
    &= O\left(\frac{D}{B} \cdot \sqrt{B/D} \cdot \frac{1}{\sqrt{BD}}\right) = O(1/B) \ .
\end{align*}
\end{proof}

While this method is simple to describe and analyze, it is relatively inefficient as it requires $D^2$ multiplications. We also show that a significantly more efficient scheme from \cite{AsiFNNT23} based on Subsampled Randomized Hadamard Transform (SHRT) can also be easily adapted to our setting at the expense of somewhat worse $\tilde O(1/\sqrt{B})$ expected squared error. Specifically, let $W=SHT$ denote the following distribution over random matrices: $H \in \Re^{D\times D}$ is the Hadamard matrix, $S \in \Re^{D\times D}$ is a random permutation matrix, and $T \in \Re^{D\times D}$ is a diagonal matrix where $T_{i,i}$ are independent samples from the Rademacher distribution (that is, uniform over $\pm1$). An important (and well-known) property of this family of matrices is that multiplication by $W$ and $W^\top$ can be performed in time $O(D\log D)$ \cite{AilonC06}.
\begin{theorem}
\label{thm:shrt-trunc-error}
Let $W\in \Re^{D\times D}$ be a randomly and uniformly chosen $SHT$ matrix as described above. Then 
for every $v\in \Re^D$ such that $\|v\|_2\leq 1$, $$\E_{W} \left[\left\|W^\top \clipop_B(Wv) - v\right\|_2^2\right] = O\left(\frac{\log^2 D}{B}\right) .$$ 
Further, multiplication by $W$ and $W^\top$ can be performed in time $O(D\log D)$.
\end{theorem}
To prove this result, we first establish some relevant properties of SHRT.
\begin{lemma}[\cite{AsiFNNT23}]\label{lem:shrt-conc}
Suppose $W_B = S_BHT$ is obtained with $S_B$ being a $B\times D$ uniform sampling matrix without replacement, $H$ being Hadamard matrix and $T$ being a Rademacher diagonal matrix as above. Then for some constant $C>0$, for any fixed $u\in\Re^D$ of unit Euclidean norm and $\delta\in(0,1)$,
 $$
    \Pr_{W_B}\left[\left| \|W_B u\|_2^2 - \sqrt{\frac BD} \right| > C\sqrt{\log^2(B/\delta)/D}\right] < \delta .
    $$
In particular, choosing $\delta = 1/D$ implies that for some constant $C_1>0$,
 $$
    \left| \E_{W_B}\left[ \|W_B u\|\right]  - \sqrt{\frac BD} \right| \leq C_1 \sqrt{\log^2(D)/D} .
    $$
\end{lemma}

\begin{proof}[Proof of Theorem \ref{thm:shrt-trunc-error}]
As in the proof of Theorem \ref{thm:truncated-rotation-error}, we restrict our attention to the case $\|v\|=1$ and let
 $w = w_1,\ldots,w_\Delta=  Wv$. We note that matrix $S$ being a random uniform permutation implies that every $B$-block of coordinates in $Wv$ corresponds to picking $B$ coordinates of $HT$ randomly and uniformly without replacement. Therefore, we can apply Lemma~\ref{lem:shrt-conc} to obtain that for every $i\in [\Delta]$:
\[
    \left|\E_W[\|w_i\|_2]-\sqrt{B/D}\right| = O\left(\frac{\log(D)}{\sqrt{D}}\right) .
\]
In addition, by the permutation symmetry of the distribution of $S$, we have that
\[
    \E_W[\|w_i\|^2] = \frac{B}{D} .
\]
Now, following the same steps as in the proof of Theorem \ref{thm:truncated-rotation-error}, we have
\begin{align*}
\E_W\left[\left\| W^\top \clipop_B(Wv) - v\right\|_2^2\right] &\leq 2 \sqrt{B/D} \cdot \sum_{i\in[\Delta]} \E_W\left[\left(\sqrt{B/D}  - \|w_i\|_2 \right)\right]\\
    &= O\left(\frac{D}{B} \cdot \sqrt{B/D} \cdot \frac{\log(D)}{\sqrt{D}} \right) = O\left(\frac{\log(D)}{\sqrt{B}}\right) \ .
\end{align*}
\end{proof}

\section{$k$-block-sparse-DPF Construction Details}
\label{sec:ksparsedf}

\subsection{Secret-Sharing $k$-block-sparse Vectors}

We first review the tree-based DPF construction from~\cite{CCS:BoyGilIsh16}, both for single bit outputs and for group element outputs (using our notation). We then describe how we build on their techniques to construct efficient $k$-block-sparse DPFs. 

\paragraph{Tree-based DPF of ~\cite{CCS:BoyGilIsh16}.} The original tree-based DPF construction is formulated for one-sparse vectors without blocks. Let us suppose that the client wants to send a 1-sparse function $f : \{0,1\}^d \rightarrow \{0,1\}$ with an input domain size of $\fulldim=2^d$, where the output is non zero only on input $\alpha\in\{0,1\}^d$. Let us define $f_\layerind  : \{0,1\}^\layerind  \rightarrow \{0,1\}$ as the function that computes the sum $f_\layerind (x) = \sum_{y \in \{0,1\}^{d-\layerind }} f(x || y)$ where $||$ denotes concatenation. Note that $f_d =f$ and each $f_\layerind $ is $1$-sparse. Let $\alpha_\layerind $ be the input that produces a non-zero output for $f_\layerind $.

In the tree-based construction an invariant holds at every layer $\layerind $ in the tree.
Servers 1 and 2 hold functions $s_\layerind , t_\layerind  : \{0,1\}^\layerind  \rightarrow \{0,1\}^\lambda$ whose vectors of outputs are secret shares of $r_\layerind  \cdot e_{\alpha_\layerind }$, where $r_\layerind $ is a (pseudo) random $\lambda$-bit string and $e_{\alpha_\layerind }$ is the basis vector with value 1 at position $\alpha_\layerind $ and 0 elsewhere. In other words, $s_\layerind (x) - t_\layerind (x)$ is zero for $x \neq \alpha_\layerind $ and is a pseudorandom value $r_i$ for $x=\alpha_\layerind $. The servers also hold functions $u_\layerind , v_\layerind  : \{0,1\}^\layerind  \rightarrow \{0,1\}$, whose vectors of outputs are secret shares of $e_\layerind $. The client knows all secret-shared values.

When $\layerind =0$, defining these functions is simple: a client with input $x$ sets $s_0, t_0$ to return a constant $\lambda$-bit string chosen at random, and sets $r_0 = s_0(x)-t_0(x)$. It also sets $u_0$ to return a random bit and $v_0(x)= 1 - u_0(x)$.

For the inductive step, we do the following:
\begin{itemize}
    \item Each server provisionally expands out their seeds using a PRG $G:\{0,1\}^\lambda\rightarrow\{0,1\}^{2(\lambda+1)}$. We can think of the first half of the PRG's output as the left and the second half as the right child in the tree. They parse for 
    $x \in \{0,1\}^\layerind $:%
    \begin{align*}
        s'_{\layerind +1}(x ||0) || u'_{\layerind +1}(x|| 0) || s'_{\layerind +1}(x ||1) || u'_{\layerind +1}(x|| 1) &= G(s_\layerind (x))\\
        t'_{\layerind +1}(x ||0) || v'_{\layerind +1}(x|| 0) || t'_{\layerind +1}(x ||1) || v'_{\layerind +1}(x|| 1) &= G(t_\layerind (x))
    \end{align*}

    \item Let $\bar{b}_\layerind $ be such that $\alpha_{\layerind +1} \neq \alpha_\layerind  || \bar{b}_\layerind $, so that this is the term that needs to be corrected to zero. The client computes the correction according to:
    \begin{align*}%
        c_{\layerind +1} = s'_{\layerind +1}(\alpha_\layerind  || \bar{b}_\layerind ) - t'_{\layerind +1}(\alpha_\layerind  || \bar{b}_\layerind )
    \end{align*}
    and sends it to both servers. The servers then set, for each $x \in \{0,1\}^{\layerind }$, and $b\in \{0,1\}$:
    \begin{align*}
        s_{\layerind +1}(x || b) = s'_{\layerind +1}(x || b) - u_\layerind (x) c_{\layerind +1},\\
        t_{\layerind +1}(x || b) = t'_{\layerind +1}(x || b) - v_\layerind (x) c_{\layerind +1}.
    \end{align*}
    It is then easy to check that for $x \neq \alpha_\layerind $, the equality $u_\layerind (x)=v_\layerind (x)$ implies that $s_{\layerind +1}(x || b) = t_{\layerind +1}(x || b)$. Moreover for $y = \alpha_\layerind  || \bar{b}_\layerind $, we have
    \begin{align*}
        s_{\layerind +1}(y) - t_{\layerind +1}(y) &= s'_{\layerind +1}(y) - u_\layerind (\alpha_\layerind ) c_{\layerind +1}
                                - t'_{\layerind +1}(y) + v_\layerind (\alpha_\layerind ) c_{\layerind +1}\\
                                &= s'_{\layerind +1}(y) - t'_{\layerind +1}(y) - (u_\layerind (\alpha_\layerind ) -v_\layerind (\alpha_\layerind )) c_{\layerind +1}\\
                                &= s'_{\layerind +1}(y) - t'_{\layerind +1}(y) - 1 \cdot c_{\layerind +1}\\
                                &= 0.                                
    \end{align*}
    Here the last step follows by definition of $c_{\layerind +1}$.
    \item Finally, we need to correct the bit components. For this purpose, we compute two bit corrections. Recall that $u_{\layerind +1}(y) = v_{\layerind +1}(y)$ for each $y \neq \alpha_\layerind  || b$ for $b\in\{0,1\}$. 
    We compute $m_{\layerind +1}(0) = u_{\layerind +1}(\alpha_\layerind  || 0) - v_{\layerind +1}(\alpha_\layerind  || 0) + \bar{b}_\layerind$ and similarly $m_{\layerind +1}(1) = u_{\layerind +1}(\alpha_\layerind  || 1) - v_{\layerind +1}(\alpha_\layerind  || 1) + (1-\bar{b}_\layerind )$, and send both of these to the two servers. Similarly to above, the servers now do, for each $x \in \{0,1\}^\layerind $ and each $b \in \{0,1\}$:
    \begin{align*}
        u_{\layerind +1}(x || b) &=  u'_{\layerind +1}(x || b) - u_\layerind (x)m_{\layerind +1}(b)\\
        v_{\layerind +1}(x || b) &=  v'_{\layerind +1}(x || b) - v_\layerind (x)m_{\layerind +1}(b)
    \end{align*}
    The correctness proof is identical to the one for the $s-t$ case.
\end{itemize}

Note that the multiplications above all have one of the arguments being bits so this is point-wise multiplication. We have now verified that the invariant holds for $(\layerind +1)$. 

\hannah{We refer to the tuple $(c_i,m_i(0),m_i(1))$ as a correction word, where $c_i$ is the seed correction and $m_i(0),m_i(1)$ are the correction bits. Intuitively, since $u_{\layerind }(y) = v_{\layerind }(y)$ for each $y \neq \alpha_\layerind  $, each correction word is only applied to one expanded seed in each level. For all other expanded seeds, the correction word is either never applied (if $u_{\layerind }(y) = v_{\layerind }(y)=0$) or applied twice (if $u_{\layerind }(y) = v_{\layerind }(y)=1$), in which case these two applications cancel out.}

\paragraph{Non-zero output from group $\mathbb{G}$.} 
\cite{CCS:BoyGilIsh16} define a variant of their construction where the output of the DPF on input $\alpha$ outputs not 1, but rather a group element $\beta\in\mathbb{G}$. Given a $convert(\cdot)$ function, which converts a $\lambda$-bit string to a group element in $\mathbb{G}$, the changes to the construction are minimal. Since the invariant holds, $s_{n}(\alpha) - t_{n}(\alpha)\neq 0$ and $s_{n}(y) - t_{n}(y)= 0$ for all $y\neq\alpha$. Since the DPF according to Definition~\ref{def:dpf} should output $\beta$ 
on input $\alpha$, an additional correction word $c_{d+1}$, which consists only of a seed correction, is constructed such that either $convert(s_{d}(\alpha)) - convert(t_{d}(\alpha))+c_{d+1}=\beta$ or $convert(s_{d}(\alpha)) - convert(t_{d}(\alpha))-c_{d+1}=\beta$, depending on whether $u_{d}(\alpha)$ or $v_{d}(\alpha)$ is one.

\paragraph{The block-sparse case with block size $\blocksize>2$.}
We adapt the DPF construction of \cite{CCS:BoyGilIsh16} from point functions to block-sparse functions, where the output on any number of input values from a single block will be a non-zero group element. More formally, a block-sparse function $f_{\alpha,\beta}$ with block size $B$ for $\alpha\in\{0,1\}^d$ and $\beta=\{\beta_0,\dots,\beta_{\blocksize-1}\}\in\mathbb{G}^{B}$ is defined to be the function $f:\{0,1\}^{d+\log \blocksize}\rightarrow\mathbb{G}$ such that $f(\alpha||j)=\beta_j$ and $f(x)=0$ for $x\neq \alpha||j$ with $j\in[\blocksize]$.

To formulate a block-sparse DPF based on tree-based DPF construction, we can use a different PRG for the final tree layer compared to the previous tree layers, such that the output is $B\log|\mathbb{G}|$ bits rather than $2(\lambda+1)$ bits in the original construction. We call this $G'$.
\begin{align*}
        s'_{d+1}(x ||0) || \dots || s'_{d+1}(x ||\blocksize-1) &= G'(s_{d}(x))\\
        t'_{d+1}(x ||0) || \dots ||  t'_{d+1}(x || \blocksize-1) &= G'(t_d(x))
\end{align*}

The correction word can be constructed analogously to the original construction, and we make use of a $convert(\cdot)$ function, which maps a $\log|\mathbb{G}|$-length bit string to an element in $\mathbb{G}$. For any input $x=\alpha||j$ for any $j\in[\blocksize]$, we set $c_{d+1,j}$ such that $convert(s'_{d+1}(x)) - convert(t'_{d+1}(x))+c_{d+1,j}=\beta_j$ or $convert(s'_{d+1}(x)) - convert(t'_{d+1}(x))-c_{d+1,j}=\beta_j$, depending on whether $u_d(\alpha)$ or $v_d(\alpha)$ is one.

\hannah{This construction avoids the high cost of using the original DPF construction $\blocksize$ times, both in terms of computation and communication. In particular, the original construction would involve $\blocksize$ DPF keys of size $O(\lambda(d+\log \blocksize))$
each, while this construction yields a single DPF key of size $O(\lambda d+\blocksize)$. %
DPF key generation with the original construction will involve $O(d+\log \blocksize)$ PRG evaluations for each of the $\blocksize$ keys, while our optimization involves $O(d)$ PRG evaluations, as well as one larger PRG evaluation, where the size of this PRG output scales with $B$. In the evaluation step, when the entire tree is evaluated, naively using the original DPF construction $\blocksize$ times would result in $\blocksize\cdot 2^{d+\log \blocksize}$ PRG evaluations, while our optimization involves only $O(2^d)$ smaller and $O(2^d)$ larger PRG evaluations.
}

\paragraph{The $k$-block-sparse case}

We now describe the idea of a construction for $k$-block-sparse DPFs, as specified in Definition~\ref{def:kblockdpf}. Instead of a single index $\alpha_\layerind $ at each level, we have a set $\{\alpha^0_\layerind ,\dots,\alpha^{\Tilde{k}-1}_\layerind \}$, where $\Tilde{k}$ corresponds to the number of distinct $\layerind $-bit prefixes in $\alpha$, at most $k$. We will begin by formulating a change to the construction that allows us to share $k$-block-sparse vectors, which is a naive extension to the block-sparse version of the DPF construction of \cite{CCS:BoyGilIsh16}. Later, we introduce an optimization to reduce the number of correction words applied to each node. 

The invariant on all tree layers except the lowest one is essentially the same as before, except that there are now $k$ separate $u$ and $v$ functions per correction word at each layer, and the client will send $k$ correction words per layer. \hannah{The PRG at the upper level now outputs $2(k-1)$ additional bits, and the bit components of the expanded and corrected seeds are secret shares of an indicator, which specifies which correction word, if any, will be applied next.  As in the original construction, we maintain the goal that each correction word is applied to only at most one expanded seed in that layer. In particular, the correction word with index $\ell\in[\Tilde{k}]$ at level $i\in[d]$ will be applied to the expanded seed at position $\alpha_i^\ell\in[2^i]$. In the lowest layer, we can formulate a correction word, which is interpreted as a group element, by applying an idea analogous to that of the block-sparse case.}

\hannah{For the goal of generating DPF keys for a $k$-block-sparse vector, this construction avoids the overhead of generating $k\blocksize$ DPF keys from the original construction. We inherit all advantages of using blocks of size $B$ from the block-sparse construction and obtain further savings by allowing $k$ non-zero blocks. We can compare the the costs of naively using $k$ instantiations of a block-sparse DPF to those of a single $k$-block-sparse DPF instantiation. Asymptotically, the two approaches require the same amount of communication, with identical DPF key sizes, and the DPF key generation requires the same number of PRG evaluations. The savings for the construction come in the form of computation savings for server evaluation, where the number of PRG evaluations decreases by a factor of $k$, since servers must now evaluate a single tree, rather than $k$ trees when instantiating a 1-block-sparse DPF $k$ times. However, since each server applies up to $k$ correction words at each level, the total server computation still depends linearly on $k$.  This yields the following result

\begin{theorem}\label{thrm:construction_basic}
    Let $G:\{0,1\}^\lambda\rightarrow\{0,1\}^{2(\lambda+2)}$ and $G':\{0,1\}^\lambda\rightarrow\{0,1\}^{\blocksize \lceil\log |\mathbb{G}|\rceil}$ be pseudorandom generators. Then there is a scheme $(Gen,Eval)$ that defines a $k$-block-sparse DPF for the family of $k$-block-sparse functions $f'_{\alpha,\beta}:\{0,1\}^{d+\log \blocksize}\rightarrow\mathbb{G}$ with correctness error $0$. The key size is $kd(\lambda+4) + k\blocksize\lceil\log|\mathbb{G}|\rceil$. In $Gen$ the number of invocations of $G$ is at most $kd$, and the number of invocations of $G'$ is at most $k$. In $Eval$ the number of invocations of $G$ is at most $d$, and  $G'$ is invoked once. Evaluating the full vector requires $2^d$ invocations of $G$ and $2^d$ invocations of $G'$, and $O(k\fulldim)$ additional operations.
\end{theorem}

$O(2^{d+\log\blocksize})$ PRG calls for each of the $k$ DPF keys, while using our $k$-block-sparse DPF key generation requires $O(2^d)$ small and $O()$ larger PRG calls.}

\paragraph{Cuckoo Hashing.}

We next show how we can reduce the $k$-fold multiplicative overhead in the servers' computation using cuckoo hashing. The idea is to only have $w$ control bits (rather than) $k$ for each node in the tree as follows. In practice, we can set the constant $w$ to be between 2 and 5.

At every layer $i$ in the tree, there are $k$ non-zero nodes, which have indices $\{ \alpha_i^{\ell} \}_{\ell \in [k]}$. We use $\tilde{k}$ correction words per layer. Each tree node in the $i$-th layer is assigned $w$ correction words. These correction words are selected using $w$ hash functions (per layer $i$), where each hash function maps the $2^{ i}$ tree nodes to the set of $\tilde{k}$ correction words. The hash functions are public and known to both servers (they are chosen independently of the values of the DPF). The client assigns each non-zero node to one of the $w$ correction words specified (for that node) by the hash functions. This assignment does depend on the non-zero indices of the DPF and must not be known to the servers. Cuckoo hashing~\cite{PF01} shows that for any set of $k$ non-zero nodes (specified by the values  $\{ \alpha_i^{\ell} \}$), except with probability $\Tilde{\mathcal{O}}(\frac{1}{k})$ over the choice of the hash functions, the client can choose the assignment so that there are no ``collisions'': no two non-zero nodes are mapped to the same correction word. In the case of failure, which occurs with probability at most $\Tilde{\mathcal{O}}(\frac{1}{k})$, the client will generate a outputs keys corresponding to the zero vector, which can trivially be realized by picking an arbitrary assignment. This does not affect the security of the construction as the failure of cuckoo hashing is not revealed. It does however mean that the correct vector is sent with probability $1-O(\frac 1 k)$, rather than $1$. For statistical applications, this small failure probability has little impact.  %

The correction words are now constructed and applied as usual, except for the fact that each correction word now has only $w$ correction bits instead of k on each side and that one of only $w$ correction words is applied per expanded seed/node. The bit components of an expanded and corrected seed still correspond to an indicator specifying which single correction word is applied at the next layer, as before; however, it is no longer up to one of all $k$ possible correction words that will be applied, but rather one of the $w$ possible correction words specified by the $w$ hash functions. 

For simplicity, we present the formal construction by setting $w=2$. The details can be found in Figures~\ref{fig:genConst} and~\ref{fig:evalConst}, using a helper function for DPF key generation in Figure~\ref{fig:GenNext} to specify the constructions at each of the upper tree layers. Note that we formulate evaluation in Figure~\ref{fig:evalConst} for a single path in the tree for simplicity; to reconstruct the entire vector instead of just one entry, all nodes in the tree can be evaluated using the same approach. In that case, the number of invocations of $G$ is at most $2^d$, and the number of invocations of $G'$ is at most $2^d$. %

\begin{theorem}\label{thrm:construction}
    Let $G:\{0,1\}^\lambda\rightarrow\{0,1\}^{2(\lambda+2)}$ and $G':\{0,1\}^\lambda\rightarrow\{0,1\}^{\blocksize \lceil\log |\mathbb{G}|\rceil}$ be pseudorandom generators. Also suppose $\{\hashes_i\}_{i\in [d+1]}:[2^{i-1}]\rightarrow [ck]^2$ describes a set of random hash functions. Then the scheme $(Gen,Eval)$ from Figures~\ref{fig:genConst} and~\ref{fig:evalConst} is a $k$-block-sparse DPF for the family of $k$-block-sparse functions $f'_{\alpha,\beta}:\{0,1\}^{d+\log \blocksize}\rightarrow\mathbb{G}$ with correctness error $\tilde{O}(\frac 1 k)$. The key size is $3kd(\lambda+4) + 3k\blocksize\lceil\log|\mathbb{G}|\rceil$. In $Gen$ the number of invocations of $G$ is at most $kd$, and the number of invocations of $G'$ is at most $k$. In $Eval$ the number of invocations of $G$ is at most $d$, and  $G'$ is invoked once. Evaluating the full vector requires $2^d$ invocations of $G$ and $2^d$ invocations of $G'$, and $O(\fulldim)$ additional operations.
\end{theorem}

\begin{remark}\label{remark:numfuncs}
    Note that cuckoo hashing can yield different trade-offs from those in Theorem~\ref{thrm:construction} if more than 2 hash functions are used. For $w=2$, the total number of required correction words to achieve a low failure probability is approximately $3k$. If $w=4$, the total number of required correction words to achieve a low failure probability can be reduced to be only about $1.03k$~\cite{FoutoulakisP12,FriezeM12}, leading to a key size of $1.03kd(\lambda+4) + 1.03k\blocksize\lceil\log|\mathbb{G}|\rceil$. Increasing $w$ decreases the total number of correction words, and therefore the total key size and required communication, by a constant factor, at the cost of increasing the total number of field operations per node by a constant factor. The number of PRG evaluations does not depend on $w$.
\end{remark}

\begin{remark}\label{remark:cuckoo}
    Recall from related work the application of cuckoo hashing to multi-point function secret sharing~\cite{SP:ACLS18,PoPETS:DRRT18,CCS:SGRR19,EC:dCaPol22}, where cuckoo hashing is applied directly to the $k$-sparse $2^d$-dimensional secret vector, constructing $k$ smaller vectors with a single entry each before generating DPF keys. In this approach, the total number of vector entries, and therefore also the number of PRG evaluations, is $2\cdot2^d$ or $3\cdot2^d$ when 2 or 3 hash functions are used for cuckoo hashing, as suggested by ~\cite{PoPETS:DRRT18}. Because our application of cuckoo hashing is at the level of correction words within the DPF construction, it avoids this overhead and requires only $2^d$ vector entries. 
\end{remark}

\textit{Proof of Theorem~\ref{thrm:construction}.} We prove both correctness and security of the scheme.
    
\paragraph{Correctness.} We describe and argue correctness of our optimized $k$-block-sparse DPF construction in a way that is analogous to our arguments for the original construction of~\cite{CCS:BoyGilIsh16}.     The invariant for the $k$-sparse case is that in layer $i$ of the tree construction, the $\Tilde{k}$ nodes corresponding to $\alpha_i$ are non-zero, and all others are zero. In addition, we maintain the invariant that exactly one of the two bit components on the non-zero path is 1, and the other is 0. More formally, we would like that if $x\notin\alpha_i$, $s_i(x) = t_i(x)$, $u_{\layerind }(x)=v_{\layerind }(x)$, and $q_{\layerind }(x)=r_{\layerind }(x)$. Furthermore, we would like that for $x\in\alpha_i$, exactly one of $u_{\layerind }(x)=v_{\layerind }(x)$, and $q_{\layerind }(x)=r_{\layerind }(x)$ should hold.

The function $\hashes_i:[2^{i-1}]\rightarrow[\cuckoo k]^2$ maps one $\alpha^\ell_i$ in tree layer $i$ to two correction words. We use cuckoo hashing to determine which of these two correction words will be applied for $\alpha^\ell_i$, defining function $g_i:[2^{i-1}]\rightarrow\{0,1\}$. Due to cuckoo hashing, we know that this mapping exists for any $k$-block-sparse function given all $\hashes_i$ with probability $1-\mathcal{O}(\frac{1}{k})$. 
In the upper levels, the PRG $G:\{0,1\}^\lambda\rightarrow\{0,1\}^{2\lambda+4}$ output is parsed as follows:
\begin{align*}
    s'_{\layerind +1}(x ||0) || u'_{\layerind +1}(x|| 0) || q'_{\layerind +1}(x|| 0) || s'_{\layerind +1}(x ||1) || u'_{\layerind +1}(x|| 1) || q'_{\layerind +1}(x|| 1) &= G(s_\layerind (x))\\
    t'_{\layerind +1}(x ||0) || v'_{\layerind +1}(x|| 0) || r'_{\layerind +1}(x|| 0) || t'_{\layerind +1}(x ||1) || v'_{\layerind +1}(x|| 1) || r'_{\layerind +1}(x|| 1) &= G(t_\layerind (x))
\end{align*}

In the original construction, the seed portion of the correction word was set in such a way as to set to zero the node corresponding to $\alpha_i||\bar{b}_\layerind $, where $\bar{b}_\layerind $ is defined such that $\alpha_{\layerind +1} \neq \alpha_\layerind  || \bar{b}_\layerind $. Let us now analogously define $\bar{b}_\layerind^\ell$, defined such that $\alpha_{\layerind +1}^\ell \neq \alpha_\layerind^\ell  || \bar{b}_\layerind^\ell $. It is possible that there exist indices $\ell\neq\ell'\in[k]$ such that $\alpha^\ell_i||\bar{b}_\layerind^\ell=\alpha^{\ell'}_{i+1}$, in which case we do not want to set the seed portion of the node corresponding to $\alpha^\ell_i||\bar{b}_\layerind^\ell$ to 0. For such an $\ell$, we set the seed portion of the corresponding correction word to a random bit-string instead. For other $\ell$, we set the corresponding seed correction, specified by the $g_i(\alpha_i^\ell)$th output of $\hashes_i(\alpha_i^\ell)$, as expected:
\begin{align*}
    c_{\hashes_{i+1}(\alpha_i^\ell)[g_i(\alpha_i^\ell)]} = s'_{i+1}(\alpha^\ell_i||\bar{b}_\layerind^\ell) + t'_{i+1}(\alpha^\ell_i||\bar{b}_\layerind^\ell)
\end{align*}
The corrected seed components are then for each $x\in\{0,1\}^i$ and $b\in\{0,1\}$:
\begin{align*}
    s_{i+1}(x||b) = s'_{i+1}(x||b) - u_i(x)c_{\hashes_{i+1}(x)[0]} - q_i(x)c_{\hashes_{i+1}(x)[1]} \\
    t_{i+1}(x||b) = t'_{i+1}(x||b) - v_i(x)c_{\hashes_{i+1}(x)[0]} - r_i(x)c_{\hashes_{i+1}(x)[1]}
\end{align*}

It is then easy to check that for $x \neq \alpha_\layerind^\ell $ for all $\ell\in[k]$, the equalities $u_\layerind (x)=v_\layerind (x)$ and $q_\layerind (x)=r_\layerind (x)$ imply that $s_{\layerind +1}(x || b) = t_{\layerind +1}(x || b)$. Moreover for $y = \alpha_\layerind^\ell  || \bar{b}_\layerind^\ell $, as long as $y\notin \alpha_{i+1}$, we have
\begin{align*}
    s_{\layerind +1}(y) - t_{\layerind +1}(y) &= s'_{i+1}(y) - u_i(\alpha^\ell_i)c_{\hashes_{i+1}(\alpha_i^\ell)[0]} - q_i(\alpha^\ell_i)c_{\hashes_{i+1}(\alpha_i^\ell)[1]} \\
    &~~~- (t'_{i+1}(y) - v_i(\alpha^\ell_i)c_{\hashes_{i+1}(\alpha_i^\ell)[0]} - r_i(\alpha^\ell_i)c_{\hashes_{i+1}(\alpha_i^\ell)[1]})\\
    &= s'_{i+1}(y) - t'_{i+1}(y) - (u_i(\alpha^\ell_i) -v_i(\alpha^\ell_i)) c_{\hashes_{i+1}(\alpha_i^\ell)[0]} \\
    &~~~-  (q_i(\alpha^\ell_i) -r_i(\alpha^\ell_i)) c_{\hashes_{i+1}(\alpha_i^\ell)[1]} \\
    &= s'_{i+1}(y) - t'_{i+1}(y) - 1 c_{\hashes_{i+1}(\alpha_i^\ell)[g_{i+1}(\alpha_i^\ell)]}\\
    &= 0.                                
\end{align*}
Here the last step follows by definition of $c_{\hashes_{i+1}(\alpha_i^\ell)[g_{i+1}(\alpha_i^\ell)]}$.

Finally, we need to correct the new bit components. For this purpose, we compute two bit corrections. Note that $u_{\layerind +1}(y) = v_{\layerind +1}(y)$ and $q_{\layerind +1}(y) = r_{\layerind +1}(y)$ for each $y\notin\alpha_{\layerind+1}$. On the other hand, when $y \in \alpha_{i+1}$, we would like either $u_{\layerind +1}(y) = v_{\layerind +1}(y)$ and $q_{\layerind +1}(y) \neq r_{\layerind +1}(y)$ or $u_{\layerind +1}(y) \neq v_{\layerind +1}(y)$ and $q_{\layerind +1}(y) = r_{\layerind +1}(y)$, depending on $g_{i+1}(y)$. We set the bit corrections to:
\begin{align*}
    & m_{\hashes_{i+1}(\alpha_i^\ell)[g_{i+1}(\alpha_i^\ell)]}(0) = u'_{\layerind +1}(\alpha_\layerind^\ell  || 0) - v'_{\layerind +1}(\alpha_\layerind^\ell  || 0) + (g_{i+2}(\alpha_i^\ell||0) + 1)(\alpha^\ell_i||0\in\alpha_{i+1}) \\
    & m_{\hashes_{i+1}(\alpha_i^\ell)[g_{i+1}(\alpha_i^\ell)]}(1) = u'_{\layerind +1}(\alpha_\layerind^\ell  || 1) - v'_{\layerind +1}(\alpha_\layerind^\ell || 1) + (g_{i+2}(\alpha_i^\ell||1) + 1)(\alpha^\ell_i||1\in\alpha_{i+1})  \\
    & p_{\hashes_{i+1}(\alpha_i^\ell)[g_{i+1}(\alpha_i^\ell)]}(0) = q'_{\layerind +1}(\alpha_\layerind^\ell  || 0) - r'_{\layerind +1}(\alpha_\layerind^\ell  || 0) + g_{i+2}(\alpha_i^\ell||0) (\alpha^\ell_i||0\in\alpha_{i+1})  \\
    & p_{\hashes_{i+1}(\alpha_i^\ell)[g_{i+1}(\alpha_i^\ell)]}(1) = q'_{\layerind +1}(\alpha_\layerind^\ell  || 1) - r'_{\layerind +1}(\alpha_\layerind^\ell || 1) + g_{i+2}(\alpha_i^\ell||1) (\alpha^\ell_i||1\in\alpha_{i+1})
\end{align*}

To apply the bit correction, the following is computed:
\begin{align*}
    & u_{\layerind +1}(x || b) =  u'_{\layerind +1}(x || b) - u_\layerind (x)m_{\hashes_{i+1}(x)[0]}(b)- q_\layerind (x)m_{\hashes_{i+1}(x)[1]}(b)\\
    & v_{\layerind +1}(x || b) =  v'_{\layerind +1}(x || b) - v_\layerind (x)m_{\hashes_{i+1}(x)[0]}(b)- r_\layerind (x)m_{\hashes_{i+1}(x)[1]}(b) \\
    & q_{\layerind +1}(x || b) =  q'_{\layerind +1}(x || b) - u_\layerind (x)p_{\hashes_{i+1}(x)[0]}(b)- q_\layerind (x)p_{\hashes_{i+1}(x)[1]}(b)\\
    & r_{\layerind +1}(x || b) =  r'_{\layerind +1}(x || b) - v_\layerind (x)p_{\hashes_{i+1}(x)[0]}(b)- r_\layerind (x)p_{\hashes_{i+1}(x)[1]}(b) 
\end{align*}
The correctness proof is identical to the one for the $s-t$ case when $x||b\notin\alpha_{i+1}$. Otherwise, when $x||b\in\alpha_{i+1}$, we show that 

\begin{align*}
     u_{\layerind +1}(x || b) - v_{\layerind +1}(x || b) & =  u'_{\layerind +1}(x || b) - u_\layerind (x)m_{\hashes_{i+1}(x)[0]}(b)- q_\layerind (x)m_{\hashes_{i+1}(x)[1]}(b)  \\
    &~~~ - (v'_{\layerind +1}(x || b) - v_\layerind (x)m_{\hashes_{i+1}(x)[0]}(b)- r_\layerind (x)m_{\hashes_{i+1}(x)[1]}(b)) \\
    & = u'_{\layerind +1}(x || b) - v'_{\layerind +1}(x || b) - (u_\layerind (x) - v_\layerind (x))m_{\hashes_{i+1}(x)[0]}(b) \\
    &~~~ - (q_\layerind (x)  r_\layerind (x))m_{\hashes_{i+1}(x)[1]}(b)  \\
    & = u'_{\layerind +1}(x || b) - v'_{\layerind +1}(x || b) - m_{\hashes_{i+1}(x)[g_{i+1}(x)]}(b) \\
    & = g_{i+2}(\alpha_i^\ell||b) + 1
\end{align*}

Using analogous arguments, $q_{\layerind +1}(x || b) - r_{\layerind +1}(x || b) = g_{i+2}(\alpha_i^\ell||b)$. Therefore, exactly one of either $u_{\layerind +1}(x || b) = v_{\layerind +1}(x || b)$ or $q_{\layerind +1}(y) = r_{\layerind +1}(y)$ can hold, and the required invariant is maintained.

In the lowest layer, the PRG is evaluated for each $\ell\in[k]$:
\begin{align*}
        & s'_{d+1}(\alpha^\ell ||0) ||\dots || s'_{d+1}(\alpha^\ell ||\blocksize-1)   = G'(s_{d}(\alpha^\ell))\\
        & t'_{d+1}(\alpha^\ell ||0) ||\dots || t'_{d+1}(\alpha^\ell || \blocksize-1)  = G'(t_d(\alpha^\ell))
\end{align*}

In a next step, we call the $convert(\cdot)$ function on each component of these outputs. The goal is to formulate one correction word for each of the $k$ non-zero blocks.

More formally, we would like to choose correction words $c_{\hashes_{d+1}(\alpha^\ell)[g_{d+1}(\alpha^\ell)]}(j)$ for $j\in[B]$ such that ${s}_{d+1}(\alpha^\ell||j)-{t}_{d+1}(\alpha^\ell||j)=\beta_j^\ell$:
\begin{align*}
        {s}_{d+1}(\alpha^\ell||j) &= s'_{d+1}(\alpha^\ell||j) +  u_{d}(\alpha^\ell) c_{\hashes_{d+1}(\alpha^\ell)[0]}(j) + q_{d}(\alpha^\ell) c_{\hashes_{d+1}(\alpha^\ell)[1]}(j)\\
        {t}_{d+1}(\alpha^\ell||j) &= t'_{d+1}(\alpha^\ell||j) +  v_{d}(\alpha^\ell) c_{\hashes_{d+1}(\alpha^\ell)[1]}(j) +  r_{d}(\alpha^\ell) c_{\hashes_{d+1}(\alpha^\ell)[1]}(j).
\end{align*}

Note that for $x \not\in \alpha$, $s'_{d+1}(x||j)=t'_{d+1}(x||j)$, $u_{d}(x)=v_{d}(x)$, and $q_{d}(x)=r_{d}(x)$, so ${s}_{d+1}(x||j) = {t}_{d+1}(x||j)$ for all $j\in[\blocksize]$, regardless of the correction words.

    For $\alpha^\ell$, if $g_{d+1}(\alpha^\ell)=0$, then $u_{d}(\alpha^\ell) \neq v_{d}(\alpha^\ell)$ and $q_{d}(\alpha^\ell) = r_{d}(\alpha^\ell)$. Otherwise, $q_{d}(\alpha^\ell) \neq r_{d}(\alpha^\ell)$ and $u_{d}(\alpha^\ell) = v_{d}(\alpha^\ell)$. Exactly one of ${s}_{d+1}(x||j)$ and ${t}_{d+1}(x||j)$ is independent of $c_{\hashes_{d+1}(\alpha^\ell)[g_{d+1}(\alpha^\ell)]}$. By subtracting over $\mathbb{G}$, we can find the target value of the other one and choose the correction word accordingly.

    In a bit more detail: the client computes the sequence \\ $(s'_{d+1}(\alpha^\ell||j)-t'_{d+1}(\alpha^\ell||j))_{j\in [B]} \in \mathbb{G}$. It computes the seed correction for each $j\in[\blocksize]$ and send it to both servers.. If $g_{d+1}(\alpha^\ell)=0$
    \begin{align*}
        c_{\hashes_{d+1}(\alpha^\ell)[g_{d+1}(\alpha^\ell)]}(j)=(u_{d,\ell}(\alpha^\ell) - v_{d,\ell}(\alpha^\ell))^{-1}(\beta^\ell_j - (s'_{d+1}(\alpha^\ell||j)-t'_{d+1}(\alpha^\ell||j)))
    \end{align*}
    Otherwise:    
    \begin{align*}
        c_{\hashes_{d+1}(\alpha^\ell)[g_{d+1}(\alpha^\ell)]}(j)=(q_{d}(\alpha^\ell) - r_{d}(\alpha^\ell))^{-1}(\beta^\ell_j - (s'_{d+1}(\alpha^\ell||j)-t'_{d+1}(\alpha^\ell||j)))
    \end{align*}
    
     Since $(u_{d}(\alpha^\ell) - v_{d}(\alpha^\ell)),(q_{d}(\alpha^\ell) - r_{d}(\alpha^\ell)) \in \{-1,0,1\}$, the inverse over $\mathbb{G}$ is well defined.
    
    Now for $x = \alpha^\ell$, we can write
    \begin{align*}
         {s}&_{d+1}(\alpha^\ell||j)  -{t}_{d+1}(\alpha^\ell||j) \\ 
        & = s'_{d+1}(\alpha^\ell||j) - t'_{d+1}(\alpha^\ell||j) + (u_{d}(\alpha^\ell) -v_{d}(\alpha^\ell))  c_{\hashes_{d+1}(\alpha^\ell)[0]}(j) \\
        &~~~+ (q_{d}(\alpha^\ell) -r_{d}(\alpha^\ell))  c_{\hashes_{d+1}(\alpha^\ell)[1]}(j)\\ %
        &=  s'_{d+1}(\alpha^\ell||j) - t'_{d+1}(\alpha^\ell||j) +   (\beta^\ell_j - s'_{d+1}(\alpha^\ell||j) + t'_{d+1}(\alpha^\ell||j))\\
        &= \beta^\ell_j.
    \end{align*}

\paragraph{Security.} We argue that each server's DPF key is pseudorandom. The security proof is analogous to the proof of Theorem 3.3 in~\cite{CCS:BoyGilIsh16}. We begin by describing the high-level argument. Each server begins with a random seed, which is unknown to the other server. In each tree layer, up to $k$ random seeds are expanded using a PRG, generating 2 new seeds and 4 bits, all of which appear similarly random due to the security of the PRG and the fact that the original seed appeared random. The application of a correction word will cancel the randomness of 5 of these 6 resulting seed and bit components. Specifically, one seed component and all 4 bit components are canceled; however, given only the correction word and the tree node values held by a single server, the resulting secret shares still appear random. 

It is possible to define a series of hybrids $Hyb_{w,\ell}$, where the correction words in all levels $i<w$ are replaced by random bit strings for $i\in[d+1]$, replacing Step~\ref{step:cw_lower} in Figure~\ref{fig:GenNext}, the first $\ell$ correction word components in layer $w$ are replaced by a random bit string, and if $w=d+1$ the first $\ell$ components of the final level correction word are replaced with random group elements, replacing Step~\ref{step:cw_last} in Figure~\ref{fig:genConst}. We also consider the view of the first of the two servers, without loss of generality. Specifically:
\begin{enumerate}
    \item Choose $s_0(0),t_0(0),u_0(0),v_0(0),q_0(0),r_0(0)$ honestly.
    \item Choose $CW^0,\dots,CW^{w-1}\in\{0,1\}^{\cuckoo k(\lambda+4)}$ uniformly at random.
    \item Choose $CW^w$ such that the first $\ell$ components are uniform samples and the remaining ones are computed honestly.
    \item Update $s_i(\alpha_i^\ell)$, $u_i(\alpha_i^\ell)$, $q_i(\alpha_i^\ell)$ honestly for all $i< w$ and $\ell\in[\Tilde{k}]$.
    \item For $i=w$, set $t_i(\alpha_i^0)$,$v_i(\alpha_i^0)$,$r_i(\alpha_i^0),\dots, t_i(\alpha_i^\ell)$,$v_i(\alpha_i^\ell)$,$r_i(\alpha_i^\ell)$ to random samples. Additionally set $t_{i-1}(\alpha_{i-1}^{\ell+1})$,$v_{i-1}(\alpha_{i-1}^{\ell+1})$,$r_{i-1}(\alpha_{i-1}^{\ell+1}),\dots, t_{i-1}(\alpha_{i-1}^{\Tilde{k}-1})$,$v_{i-1}(\alpha_{i-1}^{\Tilde{k}-1})$,$r_{i-1}(\alpha_{i-1}^{\Tilde{k}-1})$, to random samples. Compute $t_i(\alpha_i^{\ell+1})$,$v_i(\alpha_i^{\ell+1})$,$r_i(\alpha_i^{\ell+1}),\dots, t_i(\alpha_i^{\Tilde{k}-1})$,$v_i(\alpha_i^{\Tilde{k}-1} )$,$r_i(\alpha_i^{\Tilde{k}-1})$ honestly.
    \item For $i>w$, compute all $CW^i$ and update all $s_i(\alpha_i^\ell)$, $t_i(\alpha_i^\ell)$, $u_i(\alpha_i^\ell)$, $v_i(\alpha_i^\ell)$, $q_i(\alpha_i^\ell)$, $r_i(\alpha_i^\ell)$ honestly for all $\ell\in[k]$.
    \item The output is $s_0(0),u_0(0),q_0(0)||CW^1||\dots||CW^{d+1}$.
\end{enumerate}

Note that when $w=\ell=0$, this experiment corresponds to the honest key distribution, whereas when $w=d+1,\ell=k-1$ this yields a completely random key. We claim that each pair of adjacent hybrids will be indistinguishable based on the security of the pseudorandom generator.

We first consider $w\leq d$ and a $Hyb$-distinguishing adversary $\mathcal{A}$ who distinguishes $Hyb_{w,\ell}$ from either $Hyb_{w,\ell+1}$ if $\ell<\Tilde{k}-1$ or $Hyb_{w+1,0}$ otherwise. Given an adversary $\mathcal{A}$ with advantage $\adv$, we can construct a corresponding PRG adversary $\mathcal{B}$. This PRG adversary is given a value $r\in\{0,1\}^{2(\lambda+2)}$ and distinguishes between the cases where $r$ is truly random and $r=G(s)$, where $s\in\{0,1\}^\lambda$ is a random seed. Given $\alpha, \beta,w,\ell$, the adversary $\mathcal{B}$ constructs a DPF key according to $Hyb_{w,\ell}$; however, instead of sampling $t_w(\alpha^\ell_w),v_w(\alpha^\ell_w),r_w(\alpha^\ell_w)$ randomly, we set:
\begin{align*}
    t_{w}&(\alpha^\ell_{w-1}||0)||v_{w}(\alpha^\ell_{w-1}||0)||r_w(\alpha^\ell_{w-1}||0)||t_{w}(\alpha^\ell_{w-1}||1)||v_{w}(\alpha^\ell_{w-1}||1)||r_w(\alpha^\ell_{w-1}||1) 
\end{align*}
to $r$. If $r$ is computed pseudorandomly, then it is clear that the resulting DPF key is generated as in $Hyb_{w,\ell}(1^\lambda,\alpha,\beta)$. We must also argue that if $r$ is random, the resulting key is distributed as in either $Hyb_{w,\ell+1}$ if $\ell<\Tilde{k}-1$ or $Hyb_{w+1,0}$ otherwise. If $t_w(\alpha^\ell_w),v_w(\alpha^\ell_w),r_w(\alpha^\ell_w)$ is random, then the corresponding correction word is also uniformly random, since it is computed as the xor of a fixed bit-string with these randomly selected bit-strings, forming a perfect one-time pad. After applying this correction word, the resulting seed and bit components are also uniformly distributed, given only the previous seed and bit components for that server, as well as the correction words. 

Combining these pieces, an adversary $\mathcal{A}$ that distinguishes between the hybrids with advantage $\adv$ yields a corresponding adversary $\mathcal{B}$ for the PRG experiment with the same advantage and only polynomial additional runtime.

Finally, we consider $w=d+1$. We can make an argument similar to the previous one that an adversary that distinguishes between the distributions $Hyb_{d+1,\ell}$ and $Hyb_{d+1,\ell+1}$ with advantage $\adv$ directly yields a corresponding adversary $\mathcal{B}$ for the pseudo-randomness of the PRG output, interpreted as a group element, with the same advantage and only polynomial additional runtime. $\mathcal{B}$ can embed the challenge by setting the corresponding correction word to $(-1)^{r_{i-1}(\alpha_{i-1}^\ell)}(\beta_j^\ell - s_{i}'(\prefix||j)+r)$. If $r$ is generated pseudo-randomly, this is exactly the distribution of $Hyb_{d+1,\ell}$. If $r$ is truly random, then it similarly acts as a one-time pad on the remaining terms and the corresponding correction word is uniformly distributed, as in $Hyb_{d+1,\ell+1}$.
\qed

\begin{Boxfig}{Gen generates DPF keys for a $k$-block-sparse vector of dimension $d+\log B$ with blocks of size $B$ and security parameter $\lambda$,  where $G'$ is a PRG that takes an input of size $\lambda$ bits and outputs a bit string of length $\blocksize\log |\mathbb{G}|$. The values $\beta$ of the non-zero entries in the vector correspond to elements of group $\mathbb{G}$.}{genConst}{Gen}
\begin{description}
    \item $Gen(1^\lambda,\alpha,\beta,\mathbb{G},\blocksize, k)$:%
\begin{enumerate}
    \item Let %
    $\alpha=\{\alpha^\ell\}_{\ell\in[k]}\in\{\{0,1\}^{\depth }\}^k$
    \item Sample random $s_0(0)\leftarrow\{0,1\}^\lambda$ and $t_0(0)\leftarrow\{0,1\}^\lambda$
    \item For $i$ from $1$ to $\depth +1$, use cuckoo hashing to define mapping functions $\hashes_i:[2^{i-1}]\rightarrow [ck]^2$ and $g_i:\{(i-1)\text{-bit prefixes in }\alpha\}\rightarrow\{0,1\}$
    \item If $g_1(0)=0$, let $u_0(0)=0$, $v_0(0) = 1$, $q_0(0)=0$, and $r_0(0)=0$. Else let $u_0(0)=0$, $v_0(0) = 0$, $q_0(0)=0$, and $r_0(0)=1$.
    \item For $i$ from $1$ to $\depth $:
    \begin{itemize}
        \item Compute $GenNext(\alpha,i,s_{i-1},t_{i-1},u_{i-1},v_{i-1},q_{i-1},r_{i-1},\hashes_i,g_i,g_{i+1})$ and parse the output as $CW^i,s_i,t_i,u_i,v_i,q_i,r_i$ 
    \end{itemize}
    \item group $\alpha$ by entries with the same $\depth$-bit prefix 
    \item For each distinct $(\depth$-bit prefixes in $\alpha$, denoted $\finalprefix$:%
    \begin{itemize}
        \item $s_{\depth+1}'(\finalprefix||0) ||...|| s_{\depth+1}'(\finalprefix||\blocksize-1) \leftarrow G'(s_{\depth}(\finalprefix))$ %
        \item $t_{\depth+1}'(\finalprefix||0) ||...|| t_{\depth+1}'(\finalprefix||\blocksize-1)\leftarrow G'(t_{\depth}(\finalprefix))$%
        \item For $j\in[\blocksize]$, convert $s_{\depth+1}'(\finalprefix||j) := convert(s_{\depth+1}'(\finalprefix||j) )$ and $t_{\depth+1}'(\finalprefix||j) := convert(t_{\depth+1}'(\finalprefix||j) )$ 
    \end{itemize}
    \item Parse $\beta = (\beta^0,\dots,\beta^{k-1})$
    \item For $\ell\in[k]$:\label{step:cw_last}
    \begin{itemize}
        \item Denote $\finalprefix$ the $\depth $-bit prefix associated with $\ell$. Also denote $\cwind = \hashes_{\depth+1}(\finalprefix)$.
        \item Parse $\beta^\ell=(\beta^\ell_0,\dots,\beta^\ell_{\blocksize-1})$
        \item Denote $\gamma_j^\ell = \beta_j^\ell - s_{\depth+1}'(\finalprefix||j)+t_{\depth+1}'(\finalprefix||j)$ for $j\in[B]$. %
        \item Denote $c_{\rho[g_{\depth+1}(\finalprefix)]}(j) =  (-1)^{v_{\depth}(\finalprefix)}\cdot\gamma_j^\ell$
        \item If $g_i(\finalprefix)=0$, set \\
        $CW^{\depth +1}_{\cwind [0]} \leftarrow c_{\rho[g_i(\finalprefix)]}(0) || ...|| c_{\rho[g_i(\finalprefix)]}(B-1)$.%
        \item Else, set  \\ $CW^{\depth +1}_{\cwind [1]} \leftarrow (-1)^{r_{\depth}(\finalprefix)}\cdot\gamma^\ell_0 || ...|| (-1)^{r_{\depth}(\finalprefix)}\cdot\gamma_{B-1}^\ell $.%
    \end{itemize}
    \item For remaining $\ell$, set $CW^{\depth +1}_{\ell}$ randomly
    \item Set $CW^{\depth +1}=CW^{\depth +1}_1||...||CW^{\depth +1}_k$ and $CW = CW^1 || \dots || CW^{\depth +1} || \hashes_1 || \dots || \hashes_{\depth +1}$, as well as $k_0=s_0(0)||CW,k_1=t_0(0)||CW$.
    \item return $(k_0,k_1)$, $g_1(0)$
\end{enumerate}\end{description}
\end{Boxfig}

\begin{Boxfig}{GenNext computes the seed and bit components of nodes at the next tree layer, where $G$ is a PRG that each take an input of size $\lambda$ bits and outputs a bit string of length $2(\lambda+2)$.}{GenNext}{GenNext}
\begin{description}
    \item $GenNext(\alpha,i,s_{i-1},t_{i-1},u_{i-1},q_{i-1},v_{i-1},r_{i-1},\hashes_i,g_i,g_{i+1})$:
\begin{enumerate}
        \item Group $\alpha$ by entries with the same $(i-1)$-bit prefix. For each group:
        \begin{itemize}
            \item Denote $\prefix$ the $(i-1)$-bit prefix associated with that group
            \item Expand and parse $s_{i}'(\prefix||0) ||u'_{i}(\prefix||0)||q'_{i}(\prefix||0) || s_{i}'(\prefix||1) ||$\\ $u'_{i}(\prefix||1)||q'_{i}(\prefix||1)\leftarrow G(s_{i-1}(\prefix))$
            \item Expand and parse $t_{i}'(\prefix||0) ||v_{i}'(\prefix||0)||r_{i}'(\prefix||0) || $\\ $t_{i}'(\prefix||1) ||v_{i}'(\prefix||1)||r_{i}'(\prefix||1)\leftarrow G(t_{i-1}(\prefix)) $
        \end{itemize}
        \item For each group of $(i-1)$-bit prefixes:  \label{step:cw_lower}
            \begin{itemize}
                \item Denote $\prefix$ the $(i-1)$-bit prefix associated with that group. Also denote $\cwind = \hashes_i(\prefix)$.
                \item If $\alpha$ contains values with prefixes $\prefix||0$ and $\prefix||1$, set $c_{\cwind [g_i(\prefix)]}$ random
                \item Else if $\alpha$ contains values with prefixes $\prefix||0$, set $c_{\cwind [g_i(\prefix)]}=s_{i}'(\prefix||1) + t_{i}'(\prefix||1) $
                \item Else if $\alpha$ contains values with prefixes $\prefix||1$, set $c_{\cwind [g_i(\prefix)]}=s_{i}'(\prefix||0) + t_{i}'(\prefix||0) $
                    \item For $j\in[2]$:
                    \begin{itemize}
                        \item If $\alpha$ contains a value with prefix $\prefix||j$, set $m_{\cwind [g_i(\prefix)]}(j) = u_{i}'(\prefix||j) +v_{i}'(\prefix||j) + 1 + g_{i+1}(\prefix||j)$ and $p_{\cwind [g_i(\prefix)]}(j) = q_{i}'(\prefix||j) +r_{i}'(\prefix||j) + g_{i+1}(\prefix||j)$
                        \item Else, set $m_{\cwind [g_i(\prefix)]}(j) = u_{i}'(\prefix||j) +v_{i}'(\prefix||j) $ and $p_{\cwind [g_i(\prefix)]}(j) = q_{i}'(\prefix||j) +r_{i}'(\prefix||j)$
                    \end{itemize}
            \end{itemize}
        \item For all indeces $\ell$ that have not been set yet, set $c_{\ell}$ to a new random sample, and set $m_{\ell}(j)$ and $p_{\ell}(j)$ to random bits for all $j\in[2]$.
        \item Parse $CW^i = c_{0} || m_{0}(0) ||  p_{0}(0) ||  m_{0}(1) ||  p_{0}(1) ||  ... $\\
        $||  c_{\const k-1} || m_{\const k-1}(0) || p_{\const k-1}(0) || m_{\const k-1}(1) ||  p_{\const k-1}(1)$
        \item Group $\alpha$ by entries with the same $(i-1)$-bit prefix $\prefix$. For each group: For $j\in[2]$: If $\alpha$ contains a value with prefix $\prefix||j$: %
        \begin{itemize}
            \item Denote $\cwind = \hashes_i(\prefix)$.
            \item Set $s_{i}(\prefix||j)\leftarrow s'_{i}(\prefix||j) + u_{i-1 }(\prefix) \cdot c_{\cwind [0 ]} + q_{i-1 }(\prefix) \cdot c_{\cwind [1]}$ and $t_{i}(\prefix||j)\leftarrow t'_{i}(\prefix||j) + v_{i-1 }(\prefix) \cdot c_{\cwind [0]} + r_{i-1}(\prefix) \cdot c_{\cwind [1]}$  %
                \item Set  \\ $u_{i }(\prefix||j)\leftarrow u_{i }'(\prefix||j) +  u_{i-1}(\prefix)   m_{\cwind [0] }(j) + q_{i-1}(\prefix)  m_{\cwind [1] }(j)$ \\ $v_{i }(\prefix||j)\leftarrow v_{i }'(\prefix||j) +  v_{i-1}(\prefix)   m_{\cwind [0] }(j) + r_{i-1}(\prefix)  m_{\cwind [1] }(j)$ %
                \item Set \\ $q_{i }(\prefix||j)\leftarrow q_{i }'(\prefix||j) +  u_{i-1}(\prefix) \cdot p_{\cwind [0] }(j) +   q_{i-1}(\prefix) \cdot p_{\cwind [1] }(j)$,  \\ $r_{i }(\prefix||j)\leftarrow r_{i }'(\prefix||j) +  v_{i-1}(\prefix) \cdot p_{\cwind [0] }(j) + r_{i-1}(\prefix) \cdot p_{\cwind [1] }(j)$
        \end{itemize}
    \item Return $CW^i,s_i,t_i,u_i,v_i,q_i,r_i$
\end{enumerate}
\end{description}
\end{Boxfig}

\begin{Boxfig}{Eval evaluates one path $x$ given a DPF key $k_b$ corresponding to server $b$ for $k$-block-sparse vectors with block size $B$, where $G$ and $G'$ are PRGs that each take an input of size $\lambda$ bits and output a bit string of length $2(\lambda+2)$ and $\blocksize\log|\mathbb{G}|$, respectively. $g$ defines which correction word will be applied in the first layer. Also, let $convert$ be a function that takes as input a bit string of length $\log|\mathbb{G}|$ and outputs a group element in $\mathbb{G}$.}{evalConst}{Eval, $k$-sparse DPF}
\begin{description}
    \item[$Eval(b,g,k_b,x,\blocksize,k)$]: 
    \begin{enumerate}
    \item Parse $k_0=s_{0}(0)||CW^1||CW^2||...||CW^{\depth +1}|| \hashes_1 || \dots || \hashes_{\depth +1}$.
    Let $u_{0}(0)=b\cdot(g==0)$ and $q_{0}(0)=b\cdot(g==1)$. %
    \item for $i$ from 1 to $\depth $:
    \begin{itemize}
        \item Denote $\evalprefix$ the $(i-1)$-bit prefix of $x$. Also denote $\cwind = \hashes_i(\evalprefix)$.
        \item Parse $CW^i = c_{0} || m_{0}(0) || p_{0}(0) ||  m_{0}(1) || $\\
        $p_{0}(1) ||  ... ||  c_{\const k-1} || m_{\const k-1}(0) ||  p_{\const k-1}(0) ||  m_{\const k-1}(1) ||  p_{\const k-1}(1)$
        \item Expand $\tau^i \leftarrow G(s_{i-1}(\evalprefix))$
            \item Set $CW^i_{\cwind [0]} = c_{\cwind [0]}||m_{{\cwind [0]}}(0)||p_{\cwind [0]}(0)||c_{\cwind [0]}||m_{\cwind [0]}(1)||p_{\cwind [0]}(1)$
            \item Set $CW^i_{\cwind [1]} = c_{\cwind [1]}||m_{{\cwind [1]}}(0)||p_{\cwind [1]}(0)||c_{\cwind [1]}||m_{\cwind [1]}(1)||p_{\cwind [1]}(1)$
            \item Compute $\tau^i = \tau^i \oplus u_{i-1}(\evalprefix)\cdot CW^i_{\cwind [0]} \oplus q_{i-1}(\evalprefix)\cdot CW^i_{\cwind [1]}$
        \item Parse $\tau^i = s_{i}(\evalprefix||0)||u_{i}(\evalprefix||0)||q_{i}(\evalprefix||0)||$\\
        $s_{i}(\evalprefix||1)||u_{i}(\evalprefix||1)||q_{i}(\evalprefix||1)\in\{0,1\}^{2(\lambda+k)}$
    \end{itemize}
    \item Denote $i=\depth +1$ and $\evalprefix$ the $(i-1)$-bit prefix of $x$. Also denote $\cwind = \hashes_i(\evalprefix)$.
    \item Parse $CW^{\depth +1} = c_{0}(0) ||  ... ||c_{0}(\blocksize-1) || ... ||  c_{\const k-1}(0) || ... ||  c_{\const k-1}({B-1}) $
    \item Expand and parse $s_{i}'(\evalprefix||0) ||...|| s_{i}'(\evalprefix||\blocksize-1) \leftarrow G'(s_{\depth }(\evalprefix))$%
    \item Convert $s_{i}'(\evalprefix||j) := convert(s_{i}'(\evalprefix||j) )$ for $j\in[\blocksize]$
        \item For $j\in[\blocksize]$, compute $s_i(\evalprefix||j) = s_i'(\evalprefix||j) + u_{i-1}(\evalprefix)\cdot c_{\cwind [0]}(j)+ q_{i-1}(\evalprefix)\cdot c_{\cwind [1]}(j)$
    \item Return $(-1)^b\cdot s_i(x)$
\end{enumerate}\end{description}
\end{Boxfig}

\section{Proofs of Validity}
\label{sec:proofs}

We construct an efficient proof-system that allows a client to prove that it shared a valid block-sparse DPF. The proof is divided into two components:
\begin{enumerate}
\item $k$ correction-bit sparse. The client proves that at most $k$ of the (secret shared) correction bits are non-zero.

\item $k$ block-sparse. Given that at most $k$ of the correction bits are non-zero, the client proves that there are at most $k$ non-zero blocks in the output.
\end{enumerate}

We detail these components below. The proof system is sound against a malicious client, but we assume semi-honest behavior by the servers.

\begin{theorem} \label{thm:verifiable-DPF}

The scheme of Theorem \ref{thrm:construction} can be augmented to be a verifiable DPF for
the same function family ($k$-block-sparse functions). The construction incurs an additional round of interaction between the client and the servers (this can be eliminated using the Fiat-Shamir heuristic). The soundness error is $(\poly(k)/2^{\lambda})$.

The additional cost for the proof (on top of the construction above) is $O(k \cdot d \cdot \lambda)$ communication, $(k \cdot \poly(d,\lambda))$ client work, $(k \cdot 2^d \cdot \poly(d,\lambda))$ server work for each server.
\end{theorem}

\paragraph{Proving $k$-sparsity of the correction bits.} The secret shares for the correction bits are in $\bitset$ (the correction bit is ``on'' if these bit values are not identical). The client proves that the vector of $2 \cdot 2^d$ correction bits ($2^d$ pairs) is $k$-sparse, i.e. at most $k$ of the bits are non-zero. We use the efficient construction of 1-sparse DPFs from \cite{CCS:BoyGilIsh16}, which comes with an efficient proof system (the construction in \cite{BonehBCGI21} also handles malicious servers, but we do not treat this case here). The client sends $k$ 1-sparse DPFs (unit vectors over $\bitset^{2^{d+1}}$) whose sum equals the vector of correction bits: if these extra DPFs are indeed one-sparse and sum up to the vector of correction bits, then that vector must be $k$-sparse. We remark that we are agnostic to the field used for secret-sharing these additional DPFs (the secret shares of the correction bits are treated as the 0 and the 1 element in the field being used). Soundness and zero-knowledge follow from the properties of the DPF of \cite{CCS:BoyGilIsh16}. The communication is $O(k \cdot d \cdot \lambda)$, the client runtime is $(k \cdot \poly(d,\lambda))$, and the server runtime is $(k \cdot 2^d \cdot \poly(d,\lambda))$. The soundness error is $O(k/2^{\lambda})$.

\paragraph{Proving $k$-sparsity of the output blocks.} Given that at most $k$ of the correction bits are non-zero, the client needs to prove that there are at most $k$ non-zero blocks in the output. Consider the final layer of the DPF tree: in a zero block, the two PRG seeds held by the servers are identical, whereas in a non-zero block, they are different. Rather than expanding the seeds to $B$ group elements (as in the vanilla construction above), we add another $\lambda$ bits to the output, and we also add $\lambda$ corresponding bits to each correction word. We refer to these as the check-bits of the PRG outputs / correction words, and we refer to the original outputs (the $B$ group elements) as the payload bits. In the zero blocks the check-bits of the outputs should be identical: subtracting them should results in a zero vector. In each non-zero blocks, the check-bits of the (appropriate) correction word are chosen so that subtracting them from the (subtraction of the) check-bits of the PRG outputs also results in a zero vector. Thus, in our proof system, the servers verify that, in each block, the appropriate subtraction of the check-bits in the two PRG outputs together with the check-bits of the appropriate correction word (if any) are zero. Building on the notation of Section \ref{sec:ksparsedf}, taking $x\in\{0,1\}^{d}$ to be a node in the final layer of the DPF tree, and taking $s^{\chk}(x)$ and $t^{\chk}(x)$ to be the check-bits of the PRG output on the node $x$ for the two servers (respectively) and taking $c^{\chk}_m$ to be the check-bits of the $m$-th correction word, the servers check that for each node $x$ in the final layer:
\begin{align*} %
    0 = s_d^{\chk}(x) - t_d^{\chk}(x) & - \left( (u_d(x) + v_d(x)) \cdot c^{\chk}_{f_{d+1}(x)[0]} \right) \\
    &- \left( (q_d(x) + r_d(x)) \cdot c^{\chk}_{f_{d+1}(x)[1]} \right).
\end{align*}
To verify that equality to zero holds for all $x$ simultaneously, the servers can take a random linear combination of their individual summands and check only that the linear combination equals zero (this boils down to computing a linear function over their secret shared values, the random linear combination can be derandomized by taking the powers of a random field element). This only requires exchanging a constant number of field elements. 

This part of the proof maintains zero knowledge. The new information revealed to the servers are the check-bits in the correction words. The concern could be that these expose something about the locations or the values of non-zero blocks. However, the check bits of the correction words are pseudorandom even given all seeds held by a single server, and given all the payload values of all correction words, and thus each server's view can be simulated and zero-knowledge is maintained.

The above construction is appealing, but it is not quite sound: intuitively, given that there are only $k$ active correction bits (within the $k$ correction bit pairs), the correction words are only applied to at most $k$ of the blocks. If the check passes, this means that the  check bits of all but at most $k$ of the blocks had to have been 0 (except for a small error probability in the choice of the linear combination). However, it might be the case that the check-bits are 0, but the payload is not: i.e. we have two PRG seeds whose outputs are identical in their $\lambda$ bit suffix, but not in the prefix. One way to resolve this issue would be by assuming that the PRG is injective (in its suffix), or collision intractable. We prefer not to make such assumptions, and instead use an additional round of interaction to ensure that soundness holds (the interaction can be eliminated using the Fiat-Shamir heuristic). The interactive construction is as follows:
\begin{enumerate}
    \item The client sends all information for the DPF except the correction words (payload and check bits) for the last layer ($B$ group elements and $\lambda$ bits per node).

    \item The servers choose a pairwise-independent function $h$ mapping the range of the last layer's PRG to the same space and reveal it to the client.

    \item The client computes the correction words for the last layer, where the PRG used for that layer is the composition $(h \circ G')$ (the pairwise independent hash function applied to the PRG's output).
\end{enumerate}

Zero-knowledge is maintained because $(h \circ G')$ is still a PRG. Soundness now holds because the seeds for the final layer are determined before $h$ is chosen. The probability that two non-identical seeds collide in their last $\lambda$ bits, taken over the choice of $h$, is $2^{-\lambda}$. We take a union bound over all the seed-pairs and soundness follows.
 
\section{Additional Details on Experiments}
\label{app:experiment_details}
We provide more details about our experimental results in Section~\ref{sec:exp_setup} and provide additional experiments and plots in Section~\ref{sec:add-exp}.

\subsection{Experimental setup for private model training}
\label{sec:exp_setup}
Here we provide the full details for our private training experiments that were used to produce the plots in Figure~\ref{fig:plot-privacy-tradeoff}.

\paragraph{Noise comparison for DP-SGD (\Cref{fig:sigma_dpsgd}).} For our experiment that compares the standard deviation of the noise of the Gaussian mechanism and our approach, we consider a private training setup where we have a model of size $\fulldim = 2^{20} \approx  10^6$ and number of data points $n = 6 \cdot 10^5$. We user a standard setting of the parameters of DP-SGD where we have clipping norm $1.0$, $\eps = 1.0$, $\delta = 10^{-6}$, and number of epochs is $10$. We calculate the standard deviation of the noise required for the Gaussian mechanism using the dp-accounting library in python. For our approach, we calculate the standard deviation based on our description in the paper using the accounting techniques of~\cite{FeldmanS25}. For our method, we fix the communication cost $k \cdot B = 32768$ and 
then plot the ratio of the noise needed for our method over the noise of the Gaussian mechanism as a function of the block size $B$. We repeat this for different values of batch size in $\{512, 1028, 4096, 6 \cdot 10^5 \}$ and report the results in~\Cref{fig:sigma_dpsgd}.

\paragraph{MNIST experiment (\Cref{fig:mnist}).} For MNIST, 
we follow the experimental setup of~\cite{AsiFNNT23} and train a neural network with $69050$ parameters (see full description in~\Cref{tab:conv}). We run DP-SGD with fixed learning rate $0.1$, momentum $0.9$, and batch size $600$ for $10$ epochs. To privatize the gradients at each batch, we clip each individual gradient to have $\ell_2$ norm at most $1$ and use the standard Gaussian mechanism or our partitioned subsampling approach to release private gradients. We calculate the standard deviation of the noise required for the Gaussian mechanism using the dp\_accounting library in python. We set our privacy parameters to be $\eps = 2.0$ and $\delta = 10^{-6}$. For our partitioned subsampling approach, we use a block size $B = 920$ and number of blocks $k = 20$, and clip the $\ell_2$ norm of each block to be at most $L = 1.02 \sqrt{B/\fulldim}$ where $\fulldim = 69050 $ is the number of parameters in the model. We repeat this process $10$ times, each time recoding the accuracy per epoch for each method, and plot the median accuracy with 90\% confidence intervals in~\Cref{fig:mnist}.

\renewcommand{\arraystretch}{1}
\begin{table}[h]
\begin{center}
\begin{tabular}{ll}
            \hline
            \textbf{Layer} & \textbf{Parameters} \\
            \hline
            Convolution $+\tanh$ & 16 filters of $8 \times 8$, stride 2, padding 2 \\
            Average pooling & $2 \times 2$, stride 1 \\
            Convolution $+\tanh$ & 32 filters of $4 \times 4$, stride 2, padding 0 \\
            Average pooling & $2 \times 2$, stride 1 \\
            Fully connected $+\tanh$ & 32 units \\
            Fully connected $+\tanh$ & 10 units \\
            \hline
\end{tabular}
\end{center}
\caption{Architecture for convolutional network model.}
\label{tab:conv}
\end{table}
\renewcommand{\arraystretch}{3}

\paragraph{CIFAR10 experiment (\Cref{fig:cifar}).} For CIFAR, 
we produce CLIP embedding (using the version ViT-B/32) for the CIFAR10 images and train a simple two-layer neural network with $66954$ parameters: our network is a sequence of two fully connected layers: the first has dimensions $512 \times 128$ and the second $128 \times 10$. Then, we run DP-SGD with initial learning rate $4.0$, momentum $0.9$, weight decay $5 \cdot 10^{-4}$ and full batch for $10$ epochs. 
We use a stepLR scheduler for the learning rate which reduces the learning rate by a factor of $0.9$ every $5$ epochs.
Similarly to MNIST, we use the Gaussian mechanism and our partitioned subsampling approach to release private gradients, where we have clipping norm $1$ for the gradients, $\eps = 2.0$ and $\delta = 10^{-6}$. For our partitioned subsampling approach, we use a block size $B = 920$ and number of blocks $k = 25$, and clip the $\ell_2$ norm of each block to be at most $L = 1.02 \sqrt{B/\fulldim}$ where $\fulldim = 66954 $ is the number of parameters in the model. We repeat this process $10$ times, each time recoding the accuracy per epoch for each method, and plot the median accuracy with 90\% confidence intervals in~\Cref{fig:cifar}.

All of our experiments were run locally on a Macbook Pro equipped with Apple M1 Pro chip (with 10 cores), and 32GB RAM. The time for each epoch depends on the mechanism, dataset choice and batch size. For the Gaussian mechanism, each epoch takes from a few seconds up to 30 seconds, while for the partitioned subsampling approach each epochs takes about $1$-$2$ minutes.

\subsection{Additional experiments}
\label{sec:add-exp}

In this section, we present additional experimental results in different regimes than the ones presented in the main paper.
We begin in~\cref{fig:apdx-plot-communication-comp} where we compare our approach to the Gaussian mechanism and plot the error for estimating the sum of unit vectors in different dimensions. We can see that even for small communication complexity, sometime a factor of $32$ smaller than the dimension, our approach becomes competitive with the Gaussian mechanism. \cref{fig:apdx-plot-communication-comp-sample-comp} presents a similar plot where we show that the same behavior holds for different number of samples $n$.

\begin{figure}
    \centering
     \begin{subfigure}{0.32\textwidth}
        \centering
        \includegraphics[width=\textwidth]{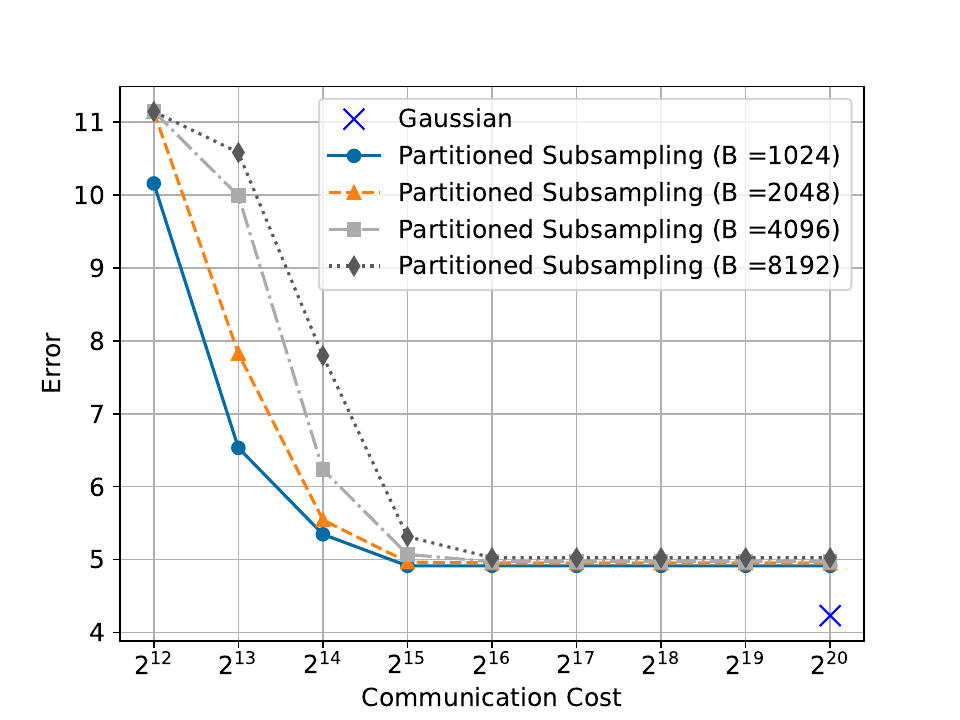}
        \caption{}
    \end{subfigure}
    \begin{subfigure}{0.32\textwidth}
        \centering
        \includegraphics[width=\textwidth]{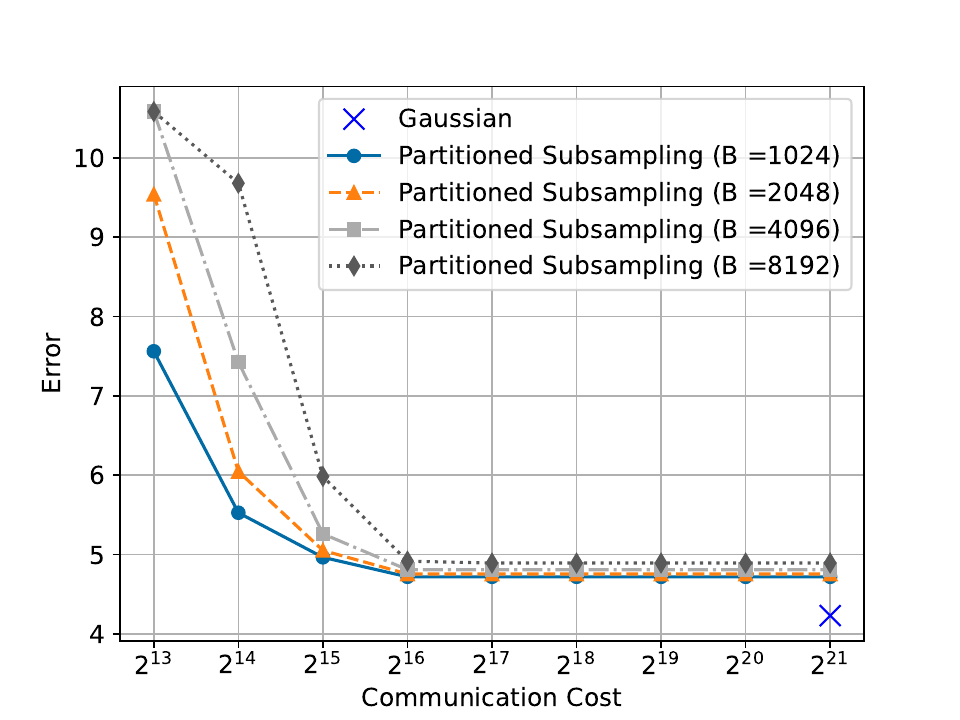}
        \caption{}
    \end{subfigure}
    \begin{subfigure}{0.32\textwidth}
        \centering
        \includegraphics[width=\textwidth]{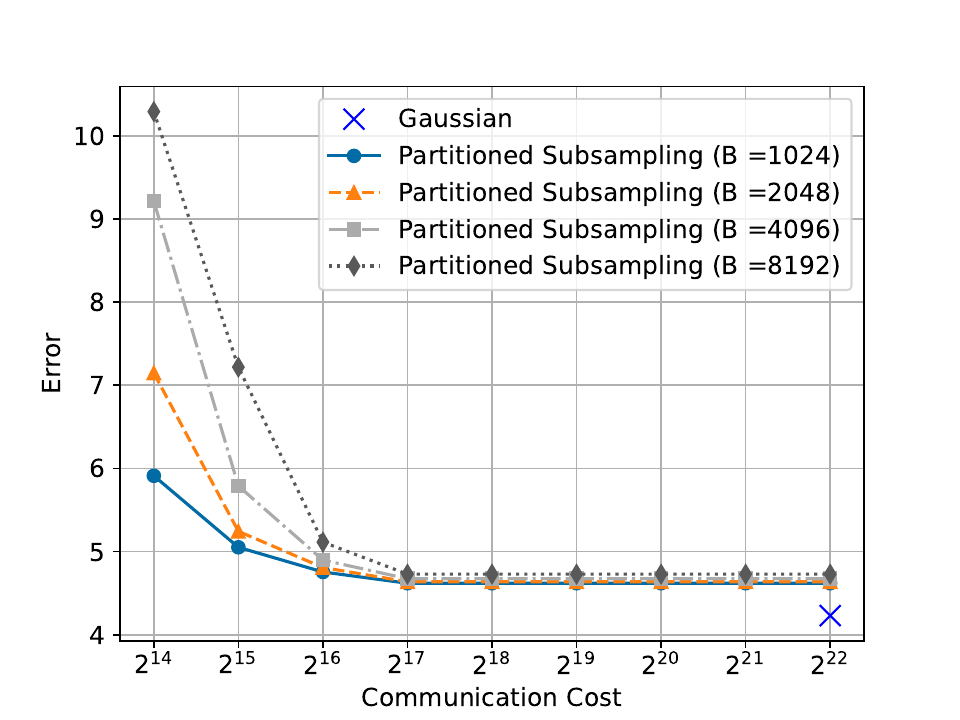}
        \caption{}
    \end{subfigure}
    \caption{\small The trade-off between the standard deviation of the error (per coordinate) and per-client communication $C = k B$, when computing the sum of $n=10^5$ vectors with dimension (a) $\fulldim = 2^{20}$, (b) $\fulldim = 2^{21}$, and (c) $\fulldim = 2^{22}$, with $(1.0, 10^{-6})$-DP. The blue 'x' shows the baseline approach of sending the whole vector.}%
    \label{fig:apdx-plot-communication-comp}
\end{figure}

\begin{figure}
    \centering
     \begin{subfigure}{0.32\textwidth}
        \centering
        \includegraphics[width=\textwidth]{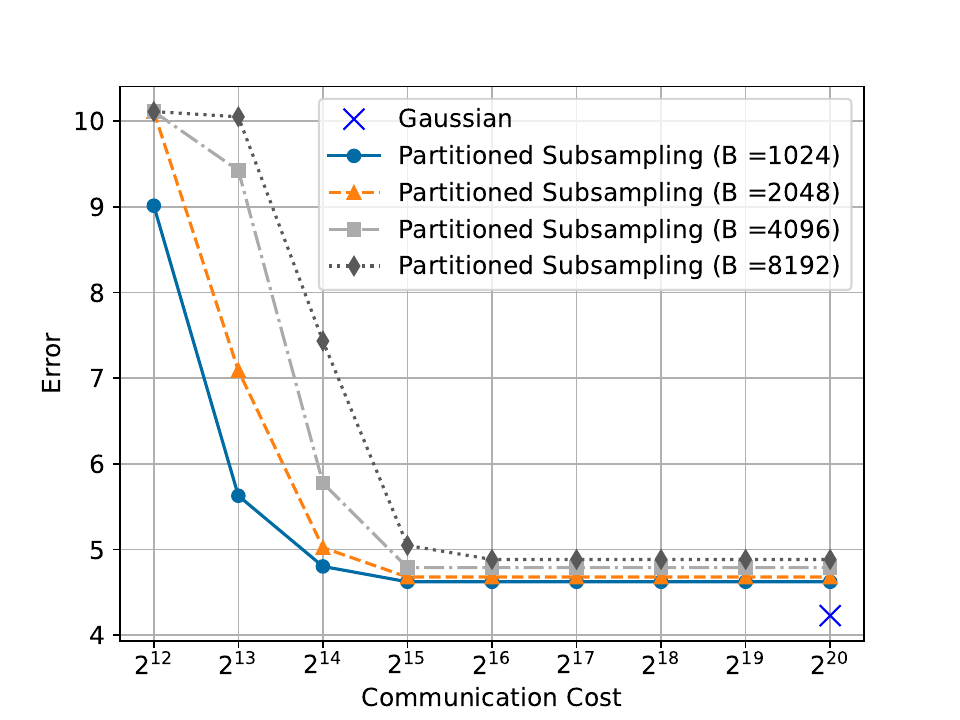}
        \caption{}
    \end{subfigure}
    \begin{subfigure}{0.32\textwidth}
        \centering
        \includegraphics[width=\textwidth]{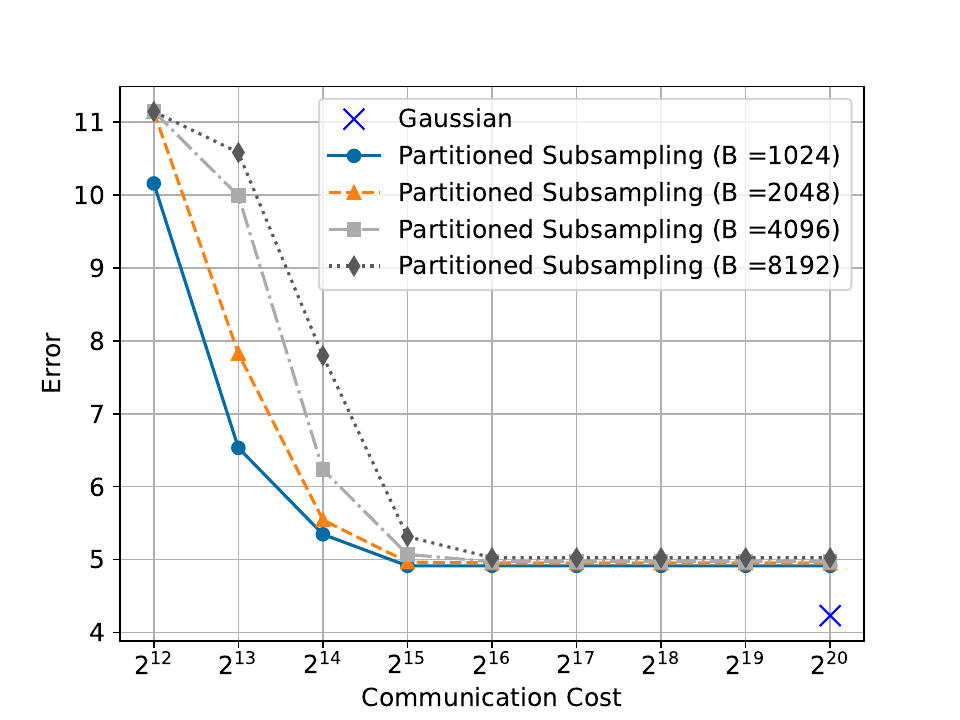}
        \caption{}
    \end{subfigure}
    \begin{subfigure}{0.32\textwidth}
        \centering
        \includegraphics[width=\textwidth]{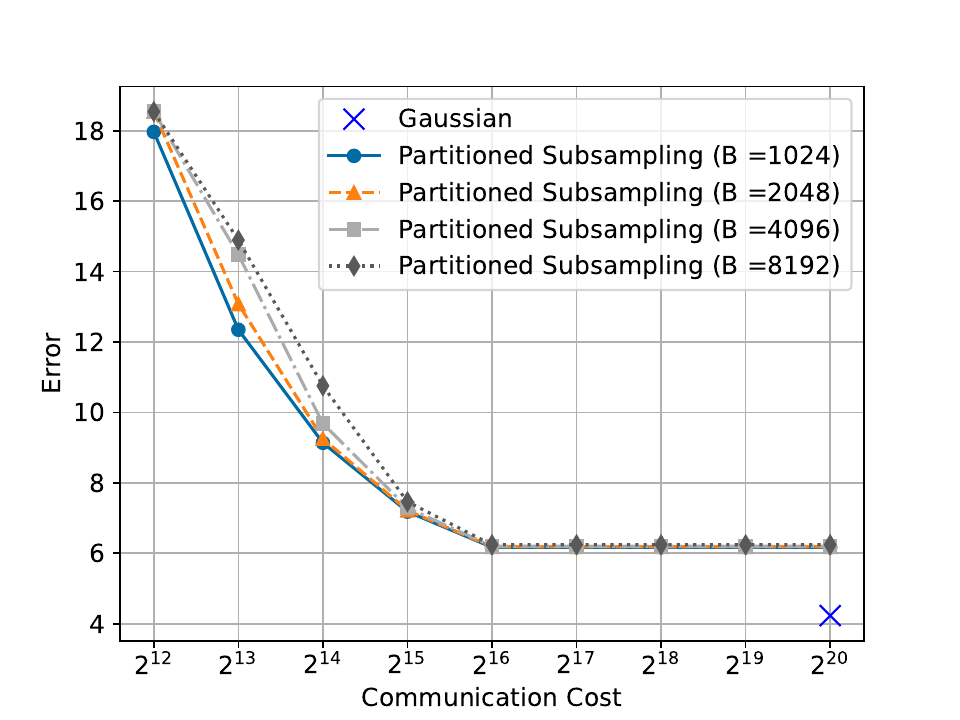}
        \caption{}
    \end{subfigure}
    \caption{\small The trade-off between the standard deviation of the error (per coordinate) and per-client communication $C = k B$, when computing the sum of $\fulldim = 2^{20}$-dimensional vectors with sample size (a) $n = 10^4$, (b) $n = 10^5$, and (c) $n = 10^6$, with $(1.0, 10^{-6})$-DP. The blue 'x' shows the baseline approach of sending the whole vector.}%
    \label{fig:apdx-plot-communication-comp-sample-comp}
\end{figure}

In \cref{fig:apdx-plot-trunc-poiss} we compare the performance of the partitioned subsampling scheme and the truncated Poisson, where the plots show that each method is favorable in different regimes: the truncated Poisson obtains better error if more communication is allowed, getting closer to the error of the Gaussian mechanism.  This is partly due to the analysis of partitioned subsampling building on RDP analysis, which even for the Gaussian mechanism yields standard deviation bounds slightly larger than the analytic Gaussian mechanism. The truncated Poisson analysis uses the tighter PRV accounting.

\begin{figure}
    \centering
    \begin{subfigure}{0.46\textwidth}
        \centering
        \includegraphics[width=\textwidth]{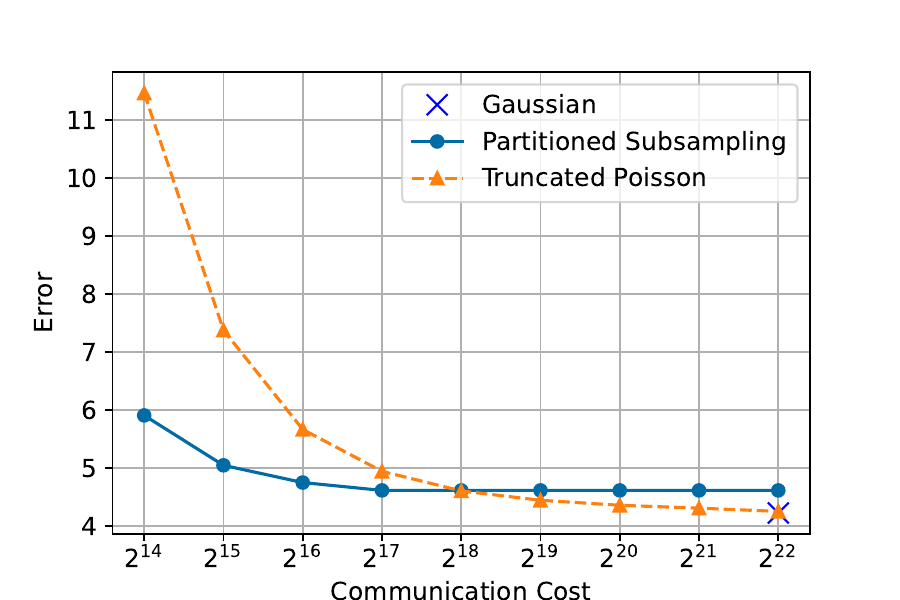}
        \caption{}
    \end{subfigure}
    \begin{subfigure}{0.46\textwidth}
        \centering
        \includegraphics[width=\textwidth]{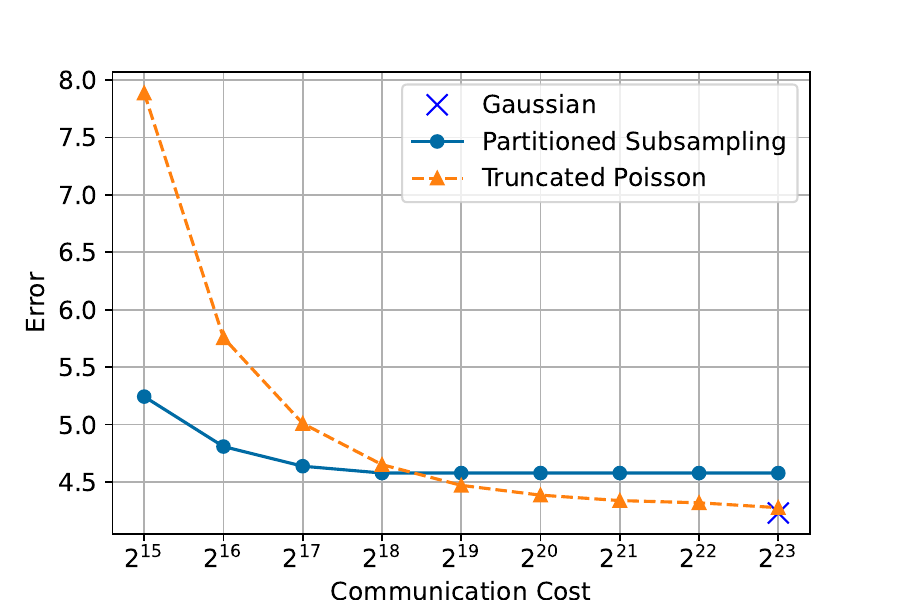}
        \caption{}
    \end{subfigure}
    \caption{\small
    The trade-off between the standard deviation of the error (per coordinate) and per-client communication for the Partitioned Subsampling scheme and the Truncated Poisson. These plots are for aggregating $n=10^5$ vectors with dimension (a) $\fulldim = 2^{22}$, (b) $\fulldim = 2^{23}$, block size $B = 2^{10}$ and $(1.0, 10^{-6})$-DP.   }
    \label{fig:apdx-plot-trunc-poiss}
\end{figure}

Furthermore, we evaluate our method for estimating the mean of real data. Specifically, we compare our method to the Gaussian mechanism for estimating the average gradient in a particular epoch during the training of a model over the CIFAR10 dataset with CLIP embeddings. We save $1024$ gradients, each of dimension $\fulldim = 66954$, and employ our alternative method to estimate the average gradient under $(2.0,10^{-6})$-DP.
We present the results in~\cref{fig:apdx-plot-trunc-err-grad}. These results corroborate our findings in the main paper for synthetic data (see \cref{fig:error_v_comm}), demonstrating that the same behavior is observed for realistic data.

\begin{figure}
    \centering
     \begin{subfigure}{0.46\textwidth}
        \centering \includegraphics[width=\textwidth]{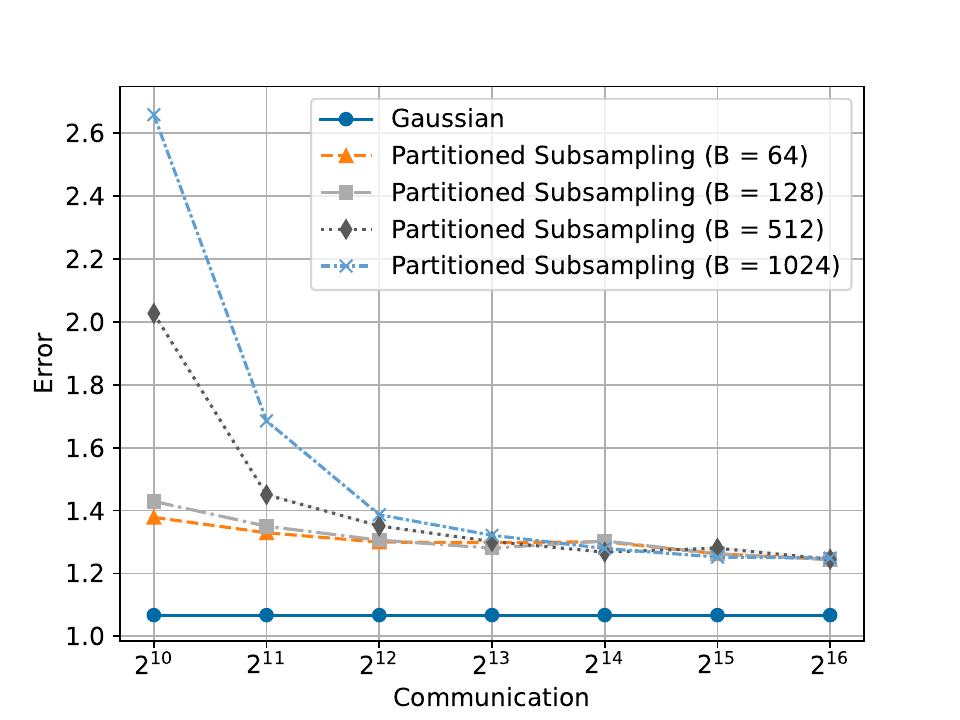}
        \caption{}
    \end{subfigure}
    \begin{subfigure}{0.46\textwidth}
        \includegraphics[width=\textwidth]{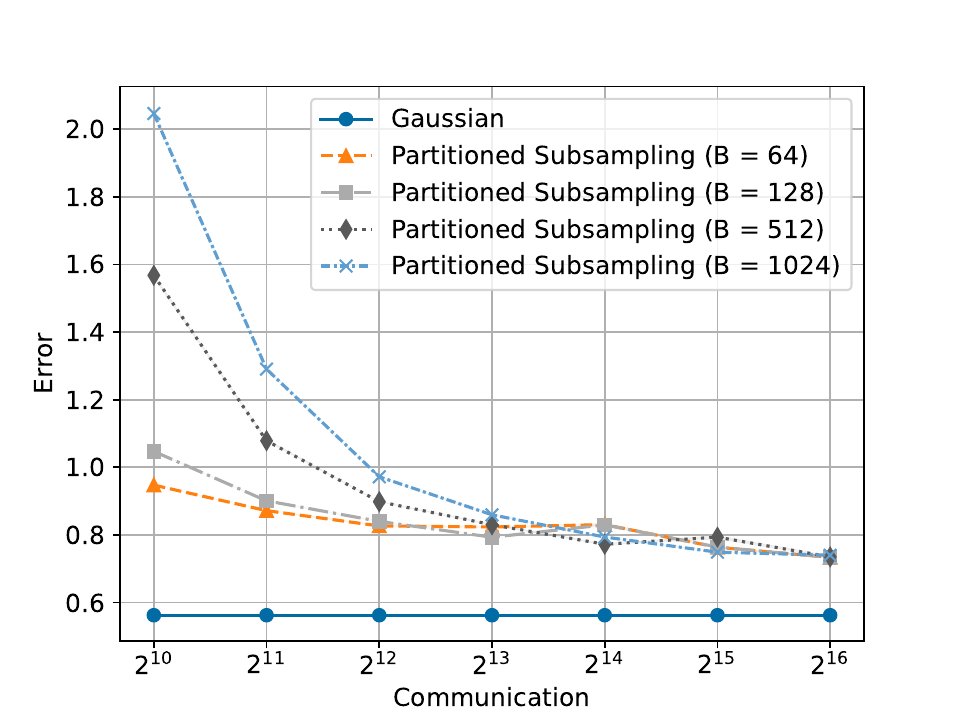}
        \centering
        \caption{}
    \end{subfigure}

    \caption{\small Error for estimating the average gradient during private model training of the CIFAR10 experiment where the model has $\fulldim = 66954$ parameters, comparing the Gaussian mechanism and partitioned subsampling for different communication costs, for (a) $\eps = 1$ and (b) $\eps = 2$. } %
    \label{fig:apdx-plot-trunc-err-grad}
\end{figure}

Finally, in~\cref{fig:apdx-plot-clip-rn50}, we present additional experiments for private model training on CIFAR10 using CLIP embeddings, employing the ResNet50 architecture. This experiment adheres to the same setup and parameters as~\cref{fig:cifar}. Furthermore, we experiment with a batch size of 1024 and a learning rate of 0.5 (while maintaining all other parameters at the same values).
We run each method 5 times and report the median accuracy as a function of epoch.
Our plots demonstrate that our method performs similarly to the Gaussian mechanism for a small batch size and full batch, with a significant reduction in communication.

\begin{figure}
    \centering
     \begin{subfigure}{0.46\textwidth}
        \centering \includegraphics[width=\textwidth]{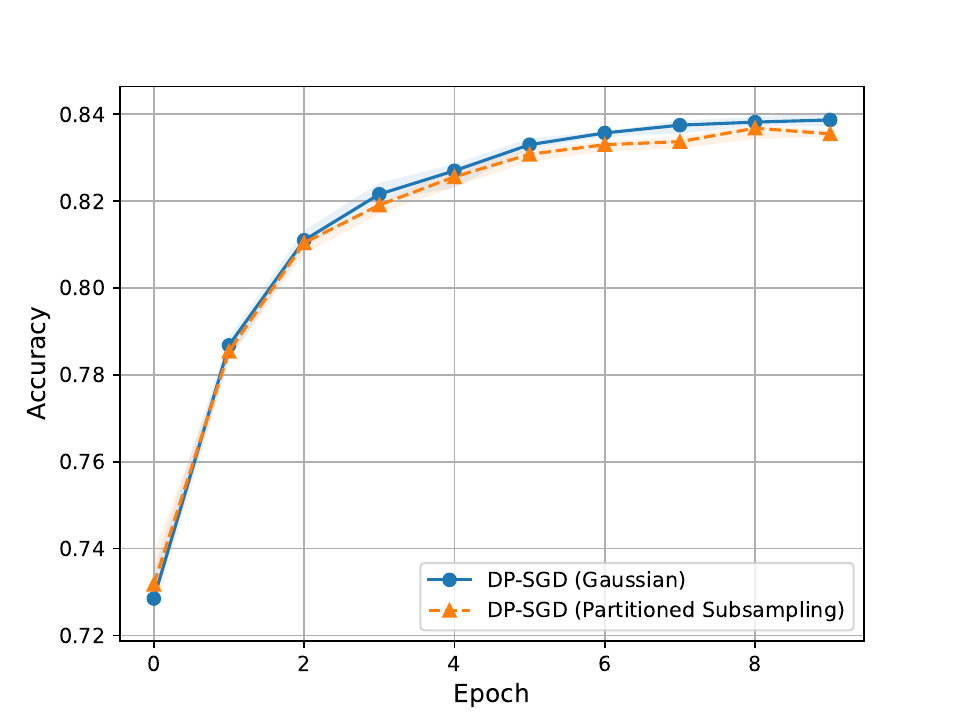}
        \caption{}
    \end{subfigure}
    \begin{subfigure}{0.46\textwidth}
        \includegraphics[width=\textwidth]{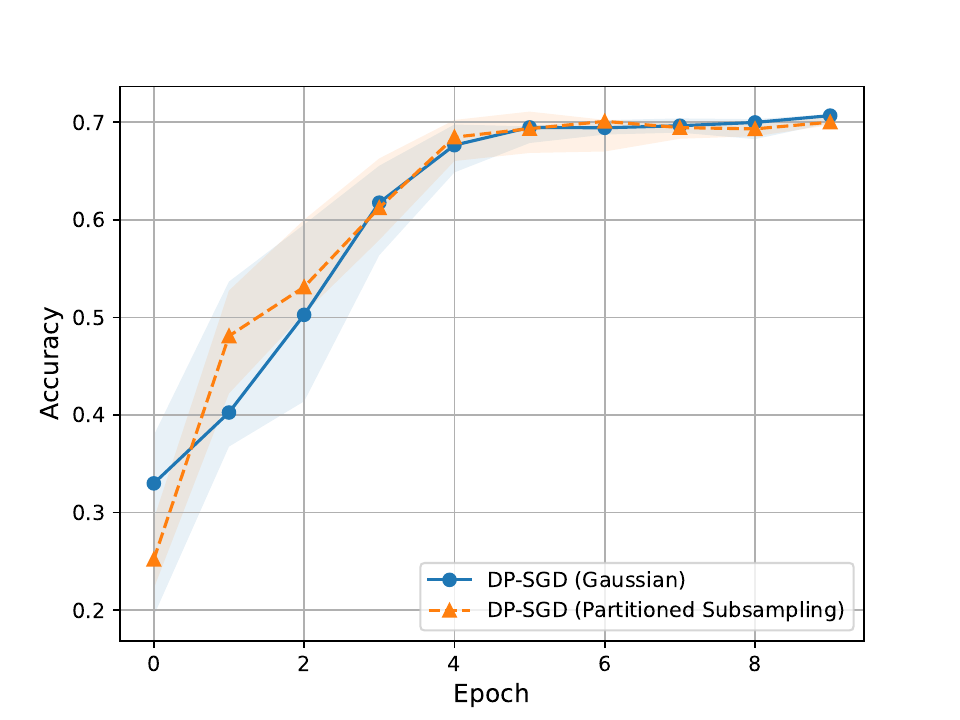}
        \centering
        \caption{}
    \end{subfigure}

    \caption{\small Comparison between PREAMBLE and the Gaussian mechanism for private model training over the CIFAR10 dataset using CLIP embeddings (version RN50 or Resnet50). We plot 90\% confidence intervals for (a) batch size of $1024$ and (b) full batch. }
    \label{fig:apdx-plot-clip-rn50}
\end{figure}

\end{document}